\documentclass[11pt, a4paper]{article}

\usepackage{amsmath,amssymb,amsthm}
\usepackage{parskip}
\usepackage{nicematrix}
\usepackage{textgreek}
\usepackage{fancyhdr}
\usepackage{bbm}
\usepackage{url}
\usepackage[overload]{empheq}
\usepackage[a4paper,margin=1in,footskip=0.25in,top=2cm]{geometry}

\RequirePackage[backend = biber, style= numeric, citestyle=numeric, sorting=nty, natbib=true,doi=false,isbn=false,url=false,eprint=false]{biblatex}

\bibliography{library}
\patchcmd{\thebibliography}{\section*{\refname}}{}{}{}

\newcommand{\floor}[1]{\lfloor #1 \rfloor}
\newcommand*\diff{\mathop{}\!\mathrm{d}}

\newtheoremstyle{theoremstyle}
  {3pt} 
  {3pt} 
  {\em} 
  {} 
  {\bfseries} 
  {.} 
  {.5em} 
  {} 
 
\theoremstyle{definition}
\newtheorem{definition}{Definition}[section]
\newtheorem{theorem}[definition]{Theorem}
\newtheorem{proposition}{Proposition}[section]
\newtheorem{remark}{Remark}[section]

\newtheorem{lemma}{Lemma}[section]
\newtheorem{notation}{Notation}[section]

\newtheorem{corollary}{Corollary}[section]
\newtheorem{assumption}{Assumption}[section]

\numberwithin{equation}{section}

\newcounter{mysubequations}

\newcommand*\samethanks[1][\value{footnote}]{\footnotemark[#1]}

\providecommand{\keywords}[1]
{
  \small	
  \textbf{\textit{Keywords:}} #1
}

\providecommand{\mscclass}[1]
{
  \small	
  \textbf{\textit{MSC2020 classification:}} #1
}

\usepackage{empheq}
\usepackage[colorlinks=true, allcolors=blue]{hyperref}
\usepackage[most]{tcolorbox}
\newtcbox{\mymath}[1][]{
    nobeforeafter, math upper, tcbox raise base,
    enhanced, colframe=black,
     standard jigsaw,
    opacityback=0,
     boxrule=0.5pt,
    #1}

\title{Approximation Rates for Deep Calibration of (Rough) Stochastic Volatility Models}

\author{
Francesca Biagini\thanks{Workgroup Financial and Insurance Mathematics, Department of Mathematics, Ludwig-Maximilians Uni-
versit\"at, Theresienstrasse 39, 80333 Munich, Germany. Emails: biagini@math.lmu.de, walter@math.lmu.de.} \and Lukas Gonon\thanks{Department of Mathematics, Imperial College London, London, SW7 1NE UK. Email:
l.gonon@imperial.ac.uk.} \and Niklas Walter\samethanks[1]\\
}

\date{\today}

\begin{document}

\maketitle

\begin{abstract}
    We derive quantitative error bounds for deep neural networks (DNNs) approximating option prices on a $d$-dimensional risky asset as functions of the underlying model parameters, payoff parameters and initial conditions. We cover a general class of stochastic volatility models of Markovian nature as well as the rough Bergomi model. In particular, under suitable assumptions we show that option prices can be learned by DNNs up to an arbitrary small error $\varepsilon \in (0,1/2)$ while the network size grows only sub-polynomially in the asset vector dimension $d$ and the reciprocal $\varepsilon^{-1}$ of the accuracy. Hence, the approximation does not suffer from the curse of dimensionality. As quantitative approximation results for DNNs applicable in our setting are formulated for functions on compact domains, we first consider the case of the asset price restricted to a compact set, then we extend these results to the general case by using convergence arguments for the option prices.
\end{abstract}


\keywords{deep neural network, volatility modelling, rough volatility, calibration, expression rate, curse of dimensionality, function approximation}

\mscclass{60H35, 68T07, 91G60}



\section{Introduction}
In this paper, we derive quantitative error bounds for the approximation of option pricing functions using deep neural networks (DNNs). In particular, we extend results in \cite{Elbrächter2022, Gonon2021a} to a calibration setting, establishing a mathematical relationship between a target accuracy and the corresponding size of the approximating network.\par 

The calibration of parametric financial models to available market data is a classical problem considered in financial mathematics and one of the main tasks faced by the industry. It can be seen as a reverse optimisation problem aiming to find parameter values such that the model output is as close as possible to what can be observed in the market. Attaining its solution usually requires multiple price simulations for different input model parameters or efficient approximate valuation techniques. As a consequence, the procedure might lead to extensive computation times restricting the use of complex models or the application to high-dimensional products like basket options. Our contribution broadens the existing body of research on applying machine learning techniques to improve the calibration process, providing a solid quantitative rationale for its application.\par 

Due to the use of deep learning techniques, this area is often referred to as \textit{deep calibration}. The concept itself can be divided into two main streams of ideas. The first one relies on learning the relationship between observable market quotes and the optimal model parameters directly. This approach was applied in the pioneering paper \cite{Hernandez2016} to calibrate the popular Hull and White short rate model introduced in \cite{Hull1990}, and also in \cite{Dimitroff2018} where convolutional neural networks were used to learn the optimal model parameters. One disadvantage of this method is that the corresponding map can possibly be quite irregular. The second approach consists in a two step procedure, see \cite{Bayer2018, Buechel2020, Horvath2020} and \cite{Liu2019}. In a first step, the relationship between the corresponding model parameters and the option prices of the underlying asset is learned by a neural network. Once the network is trained, it provides a deterministic function which can then be calibrated by classical methods. While \cite{Horvath2020} employs convolutional neural networks for approximating the pricing function, we follow the approach from \cite{Buechel2020} and \cite{Liu2019} and use feedforward neural networks to be more flexible with respect to the dimension of the output prices. \par 

More precisely, consider a $d$-dimensional risky asset price process $X^{x,v,\theta}$ with stochastic volatility process $V^{v,\theta}$ for a finite time horizon $T>0$ under a risk-neutral pricing measure $\mathbb{Q}$, where $x,v$ are the initial values of $X^{x,v,\theta}$ and $V^{v,\theta}$, respectively. In particular, we suppose that $X^{x,v,\theta}$ and $V^{v,\theta}$ are described by some parameters $\theta$ lying in a space $\Theta$ depending on the underlying process' dynamics. Following the two-step approach for model calibration described above, the core problem is to approximate $\mathbb{E}[\varphi(X^{x,v,\theta}_T,K)]$ as a function of $(x, v, \theta, K)$ by a DNN, where $\varphi$ is a payoff functional and $K$ some payoff parameter. We show that for any suitable probability measure $\mu^d$, appropriate constants $\mathfrak{C},\mathfrak{p},\mathfrak{q}>0$ and arbitrary target accuracy $\varepsilon\in(0,1/2)$ there exists a DNN $U_{\varepsilon,d,m}$ such that the $L^2(\mu^d)$-approximation error is at most $\varepsilon$ while the size of the DNN is bounded from above by $\mathfrak{C} d^{\mathfrak{p}} \varepsilon^{-\mathfrak{q}}$. In particular, the latter shows that the approximation does not suffer from the curse of dimensionality (CoD), meaning that the DNN size grows at most polynomially in the dimension of the risky assets, parameters and the inverse of the error rate.

We start by considering a Markovian stochastic volatility model.  Over the years a broad class of such models has been established including the popular Heston model \cite{Heston1993}, SABR model \cite{Hagan2002} and GARCH model \cite{Bollerslev1986}. For a comprehensive survery we refer to \cite{Shephard2005}. We start by modelling the process $(X^{x,v},V^v)$ in a general setting by using a stochastic differential equation (SDE) with measurable coefficient function satisfying a Lipschitz condition and being of sub-linear growth. This type of SDE was studied intensively in the literature and we refer to \cite{Applebaum2009, Oksendal2003} and \cite{Protter2005} for the existence and uniqueness of solutions. We apply these results to a possibly high-dimensional Heston model which allows a complex cross-correlation structure. It is worth noting that the square root coefficient function in the Heston model does not satisfy a Lipschitz assumption. Therefore, we introduce a stopping time ensuring that the variance process is bounded away from zero so that the prior results for the approximation of the pricing function are applicable. We are able to address the general case by providing a convergence result. In the last part we prove similar DNN approximation results for a version of the rough Bergomi model introduced in \cite{Bayer2015}. The latter falls into the class of so called rough volatility models. Those are characterised by their non-Markovian form and were developed because of the observation that high-frequency data of log-volatility behaves like a  fractional Brownian motion (fBM) with Hurst index $H<1/2$, see \cite{Gatheral2018}. For a technical introduction to fBMs we refer to \cite{Biagini2008} and \cite{Nualart2009}. Besides the rough Bergomi model, especially the rough Heston model \cite{ElEuch2018,ElEuch2018Hedge,ElEuch2019} and rough SABR model \cite{Akahori2017, Fukasawa2022} gained popularity in the recent years.

Our work builds on quantitative approximation results for DNNs with the ReLU activation function. For a general introduction to DNNs we refer to \cite{Goodfellow2016}. First studies about the approximation property of DNNs can be found in \cite{Cybenko1989,Hornik1989,Mhaskar1993}, normally referred to as \textit{universal approximation theorems}. Quantitative results establishing a connection between the size and depth of a DNN necessary to achieve a desired accuracy have been obtained later on for different classes of functions to approximate, see for example \cite{Barron1993,Mhaskar1993,Yarotsky2017}. In \cite{Gühring2020,Yarotsky2017} upper bounds for size and depth are presented for DNNs with the ReLU activation function for the Sobolev spaces $\mathcal{W}^{n,p}([0,1]^d)$ with $1\leq p \leq \infty$. For further studies on different function spaces, we refer to \cite{Elbrachter2021}. In \cite{Mhaskar2016} such expression rates are established for convolutional neural networks, and \cite{Bölcskei2017} obtains a link between neural network approximation rates and the complexity of function classes in $L^2(\mathbb{R}^d)$. Besides the approximation property of DNNs for different classes of functions, in the recent years many works focused on approximating solutions to partial differential equations (PDEs) without suffering from the CoD. In particular, we refer to \cite{Elbrächter2022,Gonon2021,Grohs2023,Hutzenthaler2020,Kutyniok2022,Reisinger2020} and the survey done in \cite{Beck2023}.

While the idea of using deep learning techniques for the calibration of stochastic volatility models was already considered before in \cite{Bayer2018, Buechel2020, Dimitroff2018, Hernandez2016, Horvath2020, Liu2019}, these works focus on the implementation of the approximation methods. In this paper, we derive explicit error bounds for the approximation of the parametric pricing function. In the first part of the paper we build upon the findings in \cite{Gonon2021,Reisinger2020} to analyse a general stochastic volatility model and incorporate the corresponding model parameters as input variables for the approximating DNN. Moreover, we apply the results to a high-dimensional Heston model to highlight their range of applicability. Going one step further, we also develop similar results for the rough Bergomi model from  \cite{Bayer2015} to cover rough volatility dynamics driven by fractional processes. In particular, we demonstrate how the approximation scheme for Brownian semistationary processes proposed in \cite{Bennedsen2017} can be used to develop a pathwise DNN approximation.
                                          
The paper is structured as follows. In Section \ref{Section3}, we consider a general stochastic volatility model for the risky asset process $X^{x,v}$ under a pricing measure $\mathbb{Q}$, where the coefficient functions satisfy some Lipschitz and sub-linear growth condition. In Section \ref{Section4} we apply the results from Section \ref{Section3} to a possibly high-dimensional Heston model. The non-Lipschitz behaviour of the square-root function is addressed by first introducing a stopping time and studying the approximation results for a stopped process. A martingale convergence result states that the same results hold asymptotically when the stopping time goes to infinity. In the last Section \ref{Section5} we consider the rBergomi model to show that DNN approximation results can also be obtained for a rough volatility models. In Section \ref{Section2}, we give a brief introduction to DNNs and the relevant mathematical notation and properties. 


\section{Deep Calibration for Stochastic Volatility Models with global Lipschitz Coefficients}\label{Section3}
Consider a finite time horizon $T>0$. Let $(\Omega,\mathcal{F},\mathbb{F}, \mathbb{Q})$ be a filtered probability space, where $\mathbb{F}=(\mathcal{F}_t)_{t\in[0,T]}$ is assumed to satisfy the usual condition of right-continuity and completeness. Let $W = (W_t)_{t\in[0,T]}$ and $B = (B_t)_{t\in[0,T]}$ be two possibly correlated $r$-dimensional standard $\mathbb{F}$-Brownian motions with correlation matrix $\rho = (\rho_{i,j})_{i,j=1,...,r}$, $\rho_{i,j}\in(-1,1)$. We consider a $d$-dimensional discounted stock price process $X^{x,v} = (X^{x,v}_t)_{t\in[0,T]}$ and a $d$-dimensional variance process $V^v = (V^v_t)_{t\in[0,T]}$ such that
\begin{equation}\label{ModelDynamics}
    \begin{cases}
        \diff{X}^{x,v}_t = \sigma^\theta(X^{x,v}_t, V^{v}_t)\diff{W}_t,\; X^{x,v}_0 = x,
        \\ \diff{V}^v_t = \bar{\mu}^\theta(V^{v}_t)\diff{t} + \bar{\sigma}^\theta(V^{v}_t)\diff{B}_t,\; V^v_0 = v,
        \\ \diff{\langle B^i,W^j \rangle}_t = \rho_{i,j} \diff{t},\; i,j = 1,...,d,
    \end{cases}
\end{equation}

where $x\in\mathbb{R}^d$, $x_i\geq 0$ for all $i=1,...,d$, $v\in\mathbb{R}^d$, $v_i\geq 0$ for all $i=1,...,d$ and for some measurable coefficient functions $\bar{\mu}^\theta:\mathbb{R}^d \rightarrow \mathbb{R}^d$, $\sigma^\theta:\mathbb{R}^d\times \mathbb{R}^d \rightarrow \mathbb{R}^{d\times r}$ and $\bar{\sigma}^\theta:\mathbb{R}^d \rightarrow \mathbb{R}^{d\times r}$. The superscript $\theta$ indicates the dependence of the coefficient functions on some parameter vector 
\begin{equation}\label{ParameterVector}
	\theta\in\Theta\subseteq\mathbb{R}^{md}
\end{equation}

for some $m\in\mathbb{N}$. To simplify the notation we omit the dependence on $\theta$ in $X^{x,v}$ and $V^v$. We call $\Theta$ the parameter space of the model. Note that the dimension of $\Theta$ grows proportionally to the state space dimension $d$. We introduce the variable 
$$
    \bar{d} := \max(d,r),
$$
which captures the complexity of the SDE \eqref{ModelDynamics}. The goal of this paper is to construct a DNN approximating prices of European options on the risky asset $X^{x,v}$ as a function of the model parameters, the initial value of the stock price and the variance process. In particular, we show that the DNN does not suffer under the curse of dimensionality (CoD), which means that its size grows at most polynomially in $\bar{d}$.

\begin{remark}
Note that we work directly under the assumption that the discounted stock price is a local martingale under $\mathbb{Q}$. However, our approach holds without loss of generality in a more general setting when $X^{x,v}$ is assumed to be a semimartingale under suitable assumptions on its drift under a physical measure $\mathbb{P}$.
\end{remark}

The following regularity assumptions for the coefficient functions guarantee existence and uniqueness of the solutions of the stochastic differential equation \eqref{ModelDynamics} and strong convergence of the Euler-Maruyama scheme introduced in the next section.
\begin{assumption}[Structural Assumptions]\label{StrucAssumptionV}
    There exist constants $L, K_1, K_2 > 0$ not depending on $d,r\in\mathbb{N}$ and $\theta\in\Theta$ such that
    \begin{enumerate}
        \item For all $x,y,u,v\in\mathbb{R}^d$ it holds
        \begin{align}
            \Vert \sigma^\theta(x,u) - \sigma^\theta(y,v)\Vert^2_F + \Vert \bar{\mu}^\theta(u) - \bar{\mu}^\theta(v) \Vert^2 \nonumber + \Vert \bar{\sigma}^\theta(u) - \bar{\sigma}^\theta(v)\Vert^2_F \leq L\left( \Vert x-y \Vert^2 + \Vert u-v \Vert^2 \right),
        \end{align}
        
        where $\Vert \cdot \Vert_F$ denotes the Frobenius norm.
        \item For all $x,v\in\mathbb{R}^d$, $i\in\lbrace 1,...,d \rbrace$ and $j\in\{1,...,r\}$ it holds
        \begin{equation}
           \vert \sigma^\theta_{i,j}(x,v)\vert^2 \leq K_1(1 + \Vert x\Vert^2 + \Vert v \Vert^2)
        \end{equation}
        
        and
        \begin{equation}
            \vert \bar{\mu}^\theta_i(v) \vert^2 + \vert \bar{\sigma}^\theta_{i,j}(v)\vert^2 \leq K_2(1 + \Vert v \Vert^2).
        \end{equation}

    \end{enumerate}
\end{assumption}

Under Assumption \ref{StrucAssumptionV} the SDE \eqref{ModelDynamics} admits a pathwise unique and continuous solution $\mathbb{Q}\text{-a.s.}$\hspace{-0.5mm} by Theorem 6.2.9 in \cite{Applebaum2009}. In the remainder of this section we prove some basic properties of the solution processes $X^{x,v}$ and $V^v$. 

\begin{remark}
    Throughout this work we make frequent use of the fact that for any $n\in\mathbb{N}$ and $x_1,...,x_n \in \mathbb{R}$ it holds
    \begin{equation}\label{SquaredTriangular}
        \vert x_1 + ... + x_n \vert^2 \leq n(\vert x_1 \vert^2 + ... + \vert x_n \vert^2).
    \end{equation}
\end{remark}

\begin{lemma}\label{LemmaL2XV}
It holds
    \begin{equation}
        \mathbb{E}\left[\int_0^T  \Vert X^{x,v}_t \Vert^2 \diff{t}\right] < \infty \text{ and } \mathbb{E}\left[ \int_0^T  \Vert V^v_t \Vert^2 \diff{t} \right] <\infty.
    \end{equation}
\end{lemma}

\begin{proof}
    Follows directly from Lemma 4.1. in \cite{Gonon2021}.
\end{proof}

Based on this result, we are able to find bounds for the second moments of the solution processes. In particular, we show that the second moments of $X^{x,v}$ and $V^v$ grow at most polynomially in $\Vert v \Vert$ and $\bar{d}$ as well as $\Vert x \Vert$,$\Vert v\Vert$ and $\bar{d}$, respectively. The proof follows by the same arguments as Lemma 4.1. in \cite{Gonon2021}.

\begin{lemma}\label{LemmaBoundSecMomV}
Suppose Assumption \ref{StrucAssumptionV} holds. Then there exist constants $c_1, c_2 > 0$ such that
\begin{equation}
    \mathbb{E}\left[\sup_{t\in[0,T]}\Vert V^v_t \Vert^2\right] \leq c_1 \Vert v \Vert^2 + c_2 \bar{d}^2,
\end{equation}

where the constants $c_1,c_2$ only depend on $T$, $L$ and $K_2$.
\end{lemma}

\begin{proof}
    Define the function $G(t):= \mathbb{E}\left[\sup_{s\in[0,t]}\Vert V^v_s \Vert^2\right]$ for $t\in[0,T]$. By the Cauchy-Schwarz inequality, Doob's maximal inequality, Itô's isometry, Lemma \ref{LemmaL2XV} and \eqref{SquaredTriangular} it follows
    \begin{align}\label{Lemma4.2.Bound1}
        G(t) &\leq 3\Vert v \Vert^2 + 3\mathbb{E}\left[\sup_{s\in[0,t]}\left\Vert \int_0^s \bar{\mu}^\theta(V^v_r)\diff{r}\right\Vert^2\right] + 3 \mathbb{E}\left[\sup_{s\in[0,t]}\left\Vert\int_0^s \bar{\sigma}^\theta(V^v_r)\diff{B}_r\right\Vert^2\right] \nonumber\\
        &\leq 3\Vert v \Vert^2 + 3t \mathbb{E}\left[\int_0^t \left\Vert\bar{\mu}^\theta(V^v_s)\right\Vert^2\diff{s}\right] + 12 \mathbb{E}\left[\left\Vert\int_0^t \bar{\sigma}^\theta(V^v_s)\diff{B}_s\right\Vert^2\right]\nonumber\\
        &= 3\Vert v \Vert^2 + 3t\int_0^t \mathbb{E}\left[\left\Vert  \bar{\mu}^\theta(V^v_s)\right\Vert^2\right]\diff{s} + 12 \int_0^t\mathbb{E}\left[\left\Vert \bar{\sigma}^\theta(V^v_s)\right\Vert^2_F\right]\diff{s}.
    \end{align}
    
    Assumption \ref{StrucAssumptionV} yields that
    \begin{align}\label{Lemma4.2.Bound2}
        \int_0^t \mathbb{E}\left[\left\Vert  \bar{\mu}^\theta(V^v_s)\right\Vert^2\right]\diff{s} &\leq 2\int_0^t \mathbb{E}\left[\left\Vert \bar{\mu}^\theta(V^v_s) - \bar{\mu}^\theta(0)\right\Vert^2\right]\diff{s} + 2\int_0^t \mathbb{E}\left[\left\Vert  \bar{\mu}^\theta(0)\right\Vert^2\right]\diff{s}\nonumber\\
        &\leq 2 L\int_0^t \mathbb{E}\left[\Vert V^v_s \Vert^2 \right] \diff{s} + 2 d K_2 t \nonumber\\
        &\leq 2L \int_0^t G(s)\diff{s} + 2 d K_2 t
    \end{align}
    
    and
    \begin{align}\label{Lemma4.2.Bound3}
        \int_0^t\mathbb{E}\left[\left\Vert \bar{\sigma}^\theta(V^v_s)\right\Vert^2_F\right]\diff{s} &\leq 2\int_0^t\mathbb{E}\left[\left\Vert \bar{\sigma}^\theta(V^v_s) - \bar{\sigma}^\theta(0) \right\Vert^2_F\right]\diff{s} + 2 \int_0^t\mathbb{E}\left[\left\Vert \bar{\sigma}^\theta(0) \right\Vert^2_F\right]\diff{s} \nonumber\\
        &\leq 2 L \int_0^t\mathbb{E}\left[\left\Vert V^v_s \right\Vert^2\right]\diff{s} + 2t dr K_2\nonumber\\
        &\leq 2 L \int_0^t G(s)\diff{s} + 2 dr K_2 t .
    \end{align}
    
    Therefore, by Lemma \ref{LemmaL2XV}, \eqref{Lemma4.2.Bound2} and \eqref{Lemma4.2.Bound3} it follows that $G(t)\in L^1([0,T])$. By \eqref{Lemma4.2.Bound1}, \eqref{Lemma4.2.Bound2} and \eqref{Lemma4.2.Bound3} we obtain that
    \begin{equation*}
        G(t) \leq 3\Vert v \Vert^2 + 24 \Bar{d}^2 K_2 T^2 + 3\max(2T,8)L\left(\int_0^t G(s) \diff{s}\right).
    \end{equation*}
    
    So by Gronwall's inequality as stated in Proposition 6.1.4 in \cite{Applebaum2009} we get 
    \begin{equation}
         G(t) \leq 3\left(\Vert v\Vert^2 + 8 \Bar{d}^2 K_2 T^2\right)\exp(3\max(2T,8)LT).
    \end{equation}
    
    Finally, for all $t\in[0,T]$ we obtain
    \begin{equation}
        G(t) \leq c_1 \Vert v\Vert^2 + c_2 \bar{d}^2,
    \end{equation}
    
    where
    \begin{equation*}
        c_1 := 3\exp(3\max(2T,8)LT)\text{ and } c_2 := 24 K_2 T^2 \exp(3\max(2T,8)LT).
    \end{equation*}
\end{proof}

A similar result can be found for the process $X^{x,v}$. However, due to its additional dependence on $V^v$, its second moment is increasing faster in $\bar{d}$, but still at most polynomial.

\begin{lemma}\label{LemmaBoundSecMomX}
Suppose Assumption \ref{StrucAssumptionV} holds. Then there exist constants $c_3, c_4, c_5 > 0$ such that
\begin{equation}
    \mathbb{E}\left[ \sup_{t\in[0,T]} \Vert X^{x,v}_t \Vert^2\right] \leq c_3 \Vert x \Vert^2 + c_4 \Vert v \Vert^2 \bar{d}^2 + c_5 \bar{d}^4,
\end{equation}

where the constants $c_3,c_4,c_5$ only depend on $T, L$ and $K_1$.
\end{lemma}

\begin{proof}
   The proof follows by the same arguments as Lemma \ref{LemmaBoundSecMomV}. For completeness, we include it in the Appendix \ref{ProofSection3}.
\end{proof}


\subsection{Discrete-time approximation of the system}
In this section we study a discrete-time approximation for the dynamics given in \eqref{ModelDynamics}. For the finite time horizon $T>0$ consider now the discrete time points
\begin{equation}\label{DiscreteTimePoints}
    0 = t_0 < ...< t_k < ... < t_N = T
\end{equation}

for some $N\in\mathbb{N}$. We assume that $t_k = kh$ for $h := T/N$ and $k=0,...,N$. Consider the Euler-Maruyama scheme given by

\begin{equation}\label{EulerX}
    \hat{X}^{x,v}_{t_{k+1}} = \hat{X}^{x,v}_{t_k} + \sigma^\theta(\hat{X}^{x,v}_{t_k}, \hat{V}^v_{t_k})(W_{t_{k+1}} - W_{t_k}),\; \hat{X}_0 = x
\end{equation}

and
\begin{equation}\label{EulerV}
    \hat{V}^v_{t_{k+1}} = \hat{V}^v_{t_k} + \bar{\mu}^\theta(\hat{V}^v_{t_k})h + \bar{\sigma}^\theta(\hat{V}^v_{t_k})(B_{t_{k+1}} - B_{t_k}),\; \hat{V}_0 = v.
\end{equation}

Moreover, we introduce the notation $\floor{t}:= \max\lbrace t_k\mid k=0,...,N,\; t_k \leq t \rbrace$. Then the continuous-time Euler-Maruyama approximation is defined as the solution to
\begin{equation}\label{ContEuler}
    \begin{cases}
        \diff\bar{X}^{x,v}_t = \sigma^\theta(\bar{X}^{x,v}_{\floor{t}}, \bar{V}^v_{\floor{t}})\diff{W}_t,\; \bar{X}_0 = x\\[2mm]
        \diff{\bar{V}^v_t} = \bar{\mu}^\theta(\bar{V}^v_{\floor{t}})\diff{t} + \bar{\sigma}^\theta(\bar{V}^v_{\floor{t}})\diff{B}_t,\; \bar{V}_0 = v.
    \end{cases}
\end{equation}

One can directly see that the solution to \eqref{ContEuler} coincides with \eqref{EulerX} and \eqref{EulerV} at any time point $t=t_k$. The following results are well-known in the literature.

\begin{lemma}\label{L2floor}
It holds
    \begin{equation}
        \int_0^T \mathbb{E}\left[ \Vert \bar{X}^{x,v}_{\floor{t}} \Vert^2 \right]\diff{t}<\infty \text{ and } \int_0^T \mathbb{E}\left[ \Vert \bar{V}^v_{\floor{t}} \Vert^2 \right]\diff{t}<\infty.
    \end{equation}
\end{lemma}

The next two propositions show the strong convergence of the interpolation scheme \eqref{ContEuler} under Assumption \ref{StrucAssumptionV} and that the approximation of the processes $X^{x,v}$ and $V^v$ is not affected by the CoD. Hence, we obtain that the approximation error grows at most polynomial in $\bar{d}$. The proof of these results is a modification of, e.g., Lemma 4.3 in \cite{Gonon2021}. For the reader's convenience we include the proof in the Appendix.

\begin{proposition}\label{DiffVNormalBar}
Suppose Assumption \ref{StrucAssumptionV} holds. Then there exist constants $c_6, c_7 > 0$ such that
\begin{equation}
    \mathbb{E}\left[ \sup_{t\in[0,T]} \Vert V^v_t - \bar{V}^v_t \Vert^2 \right]\leq h(c_6 \Vert v \Vert^2 \bar{d}^2 + c_7 \bar{d}^4),
\end{equation}
where $h=T/N$ and the constants $c_6,c_7$ only depend on $L,T$ and $K_2$.
\end{proposition}

\begin{proof}
    The proof is moved to the Appendix \ref{ProofSection3}.
\end{proof}

\begin{proposition}\label{DiffXNormalBar}
Suppose Assumption \ref{StrucAssumptionV} holds. Then there exist constants $c_8, c_9, c_{10}>0$ such that
\begin{equation}
    \mathbb{E}\left[ \sup_{t\in[0,T]} \Vert X^{x,v}_t - \bar{X}^{x,v}_t \Vert^2 \right]\leq h\left( c_8\Vert x \Vert^2 \bar{d}^2 +  c_{9}\Vert v \Vert^2 \bar{d}^4 + c_{10} \bar{d}^6 \right),
\end{equation}
where $h=T/N$ and the constants $c_8, c_9, c_{10}$ only depend on $L,T$ and $K_1$.
\end{proposition}

\begin{proof}
    The proof is analogous to the one of Proposition \ref{DiffVNormalBar}.
\end{proof}
 

\subsection{DNN approximation of the coefficient functions}
Besides the regularity assumptions we assume that the coefficient functions in \eqref{ModelDynamics} can be approximated by DNNs in the sense of the following assumption. A brief summary of DNNs and the terminology used in this work is given in Appendix \ref{Section2}. One should note that this work solely uses DNNs with the ReLU activation function, also specified in the Appendix.
\begin{assumption}[DNN Approximation of the Coefficient Functions]\label{AssumpCoeffApprox}
    There exist constants $C>0$, $p,q,l\geq 0$ not depending on $d,r\in\mathbb{N}$ and $\varepsilon\in(0,1/2)$ and for each $\varepsilon\in(0,1/2)$ there exist neural networks $\bar{\mu}_\varepsilon:\mathbb{R}^d\times \mathbb{R}^{md} \rightarrow \mathbb{R}^d$, $\sigma_{\varepsilon,j}:\mathbb{R}^d\times\mathbb{R}^d\times\mathbb{R}^{md} \rightarrow \mathbb{R}^d$ and $\bar{\sigma}_{\varepsilon,j}:\mathbb{R}^d\times\mathbb{R}^{md}\rightarrow\mathbb{R}^d$, $j=1,...,r$, such that for all $x,v\in\mathbb{R}^d,\theta \in \Theta$ it holds 
    \begin{enumerate}
        \item[(i)] \begin{align*}& \Vert \sigma^\theta(x,v) - \sigma_\varepsilon(x,v,\theta) \Vert^2_F + \Vert \bar{\mu}^\theta(x) - \bar{\mu}_\varepsilon(x,\theta) \Vert^2\\[2mm] + &\Vert \bar{\sigma}^\theta(x) - \bar{\sigma}_\varepsilon(x,\theta) \Vert^2_F \leq C\varepsilon^{2l+1} \bar{d}^p (1+\Vert x\Vert^2 + \Vert v \Vert^2 + \Vert \theta \Vert^2).\end{align*}
        
        \item[(ii)] \begin{align*}\Vert \bar{\mu}_\varepsilon(v,\theta) \Vert + \Vert \sigma_\varepsilon(x,v,\theta) \Vert_F + \Vert \bar{\sigma}_\varepsilon(v,\theta) \Vert_F \leq C(\bar{d}^p + \Vert x \Vert + \Vert v \Vert + \Vert \theta \Vert).\end{align*}
        
        \item[(iii)] \begin{align*}
        \text{size}(\bar{\mu}_\varepsilon) + \sum_{j=1}^d \text{size}(\sigma_{\varepsilon,j}) + \text{size}(\bar{\sigma}_{\varepsilon,j})\leq C \bar{d}^p \varepsilon^{-q}.
        \end{align*}
    \end{enumerate}
\end{assumption}

To motivate the assumptions made on the coefficient functions, we present an exemplary model, whose coefficient functions satisfy the (i)-(iii) from Assumption \ref{AssumpCoeffApprox}, in Section \ref{Section4}.

After having approximated the processes $X^{x,v}$ and $V^v$ via an Euler-Maruyama scheme in the last section, we now proceed with the next step of approximation by using Assumption \ref{AssumpCoeffApprox}. For $d,r,m\in\mathbb{N}$ and $\varepsilon\in(0,1/2)$ we consider the continuous-time processes $\tilde{X}^{\varepsilon,x,v,\theta}$ and $\tilde{V}^{\varepsilon,v,\theta}$ with dynamics for $t\in[0,T]$
\begin{equation}\label{XTildEps}
    \begin{cases}
        \diff{\tilde{X}}^{\varepsilon,x,v,\theta}_t = \sigma_\varepsilon(\tilde{X}^{\varepsilon,x,v,\theta}_{\floor{t}},\tilde{V}^{\varepsilon,v,\theta}_{\floor{t}},\theta)\diff{W}_t,\; \tilde{X}^{\varepsilon,x,v,\theta}_0 = x, \\
        \diff{\tilde{V}}^{\varepsilon,v,\theta}_t = \bar{\mu}_\varepsilon(\tilde{V}^{\varepsilon,v,\theta}_{\floor{t}},\theta)\diff{t} + \bar{\sigma}_\varepsilon(\tilde{V}^{\varepsilon,v,\theta}_{\floor{t}},\theta)\diff{B}_t,\; \tilde{V}^{\varepsilon,v,\theta}_0 = v,\\
        \diff{\langle B^i,W^j \rangle}_t = \rho_{i,j} \diff{t},\; i,j = 1,...,r.
    \end{cases}
\end{equation}

Similarly to the last section, we can also show for these approximating processes that their second moments increase at most polynomial in $\bar{d}$ and $\Vert \theta \Vert$ as well as $\Vert v \Vert$ and $\Vert x \Vert$ respectively. For the sake of simplicity throughout this section we use the notation $\bar{X} := \bar{X}^{x,v}$, $\bar{V} := \bar{V}^{v}$, $\tilde{X}^\varepsilon := \tilde{X}^{\varepsilon,x,v,\theta}$ and $\tilde{V}^\varepsilon := \tilde{V}^{\varepsilon,v,\theta}$.

\begin{lemma}\label{SecMomVTild}
Suppose Assumption \ref{AssumpCoeffApprox} holds. Then there exist constants $c_{11},c_{12},c_{13}>0$ such that
\begin{equation}
    \mathbb{E}\left[ \sup_{t\in[0,T]} \Vert \tilde{V}^\varepsilon_t \Vert^2\right] \leq c_{11}\Vert v \Vert^2 + c_{12}\Vert \theta\Vert^2 + c_{13}\bar{d}^{2p},
\end{equation}

where the constants $c_{11},c_{12},c_{13}$ only depend on $T$ and $C$.
\end{lemma}

\begin{proof}
The proof is moved to Appendix \ref{ProofSection3}.
\end{proof}

\begin{lemma}\label{SecMomXTild}
Suppose Assumption \ref{AssumpCoeffApprox} holds. Then there exist constants $c_{14}, c_{15},c_{16},c_{17}>0$ such that
\begin{equation}
    \mathbb{E}\left[\sup_{t\in[0,T]}\Vert \tilde{X}^\varepsilon_t \Vert^2\right] \leq c_{14} \Vert x \Vert^2 + c_{15} \Vert v \Vert^2 + c_{16} \Vert \theta \Vert^2 + c_{17} \bar{d}^{2p},
\end{equation}

where the constants $c_{14},c_{15},c_{16},c_{17}$ only depend on $T$ and $C$.
\end{lemma}

\begin{proof}
The proof is moved to Appendix \ref{ProofSection3}
\end{proof}

The next two propositions show that the system \eqref{XTildEps} strongly approximates \eqref{ContEuler}.

\begin{proposition}\label{DiffVBarTilde}
Suppose Assumptions \ref{StrucAssumptionV} and \ref{AssumpCoeffApprox} hold. Then there exists a constant $c_{18}>0$ such that for all $\varepsilon\in(0,1/2)$
\begin{equation}
    \mathbb{E}\left[\sup_{t\in[0,T]} \Vert \bar{V}_t - \tilde{V}^\varepsilon_t \Vert^2\right] \leq c_{18} \bar{d}^{3p} \varepsilon^{2l+1} (1 + \Vert v \Vert^2 + \Vert \theta \Vert^2),
\end{equation}

where the constant $c_{18}$ only depends on $T,C$ and $L$.
\end{proposition}

\begin{proof}
Define the function $G(t):= \mathbb{E}[\sup_{s\in[0,t]} \Vert \bar{V}_s - \tilde{V}^\varepsilon_s \Vert^2]$ for all $t\in[0,T]$. Then it follows by Hölder's inequality, Doob's martingale inequality, Itô's isometry, Lemma \ref{L2floor} and \eqref{SquaredTriangular} that
\begin{align*}
    & G(t)\\
    \leq & 2 \mathbb{E}\left[ \sup_{s\in[0,t]} \left\Vert \int_0^s \bar{\mu}^\theta(\bar{V}_{\floor{r}}) - \bar{\mu}_\varepsilon(\tilde{V}^\varepsilon_{\floor{r}},\theta)\diff{r} \right\Vert^2 \right] + 2 \mathbb{E}\left[ \sup_{s\in[0,t]} \left\Vert \int_0^s \bar{\sigma}^\theta(\bar{V}_{\floor{r}}) - \bar{\sigma}_\varepsilon(\tilde{V}^\varepsilon_{\floor{r}},\theta)\diff{W}_r \right\Vert^2 \right]\\
    \leq & 2t \int_0^t \mathbb{E}[\Vert \bar{\mu}^\theta(\bar{V}_{\floor{r}}) - \bar{\mu}_\varepsilon(\tilde{V}^\varepsilon_{\floor{r}},\theta) \Vert^2] \diff{r} + 8\mathbb{E}\left[ \left\Vert \int_0^t \bar{\sigma}^\theta(\bar{V}_{\floor{r}}) - \bar{\sigma}_\varepsilon(\tilde{V}^\varepsilon_{\floor{r}},\theta)\diff{W}_r \right\Vert^2 \right]\\
    =& 2t \int_0^t \mathbb{E}[\Vert \bar{\mu}^\theta(\bar{V}_{\floor{r}}) - \bar{\mu}_\varepsilon(\tilde{V}^\varepsilon_{\floor{r}},\theta) \Vert^2] \diff{r} + 8 \int_0^t \mathbb{E}[\Vert \bar{\sigma}^\theta(\bar{V}_{\floor{r}}) - \bar{\sigma}_\varepsilon(\tilde{V}^\varepsilon_{\floor{r}},\theta) \Vert^2_F] \diff{r}
\end{align*}

for all $t\in[0,T]$. Looking at the integrands individually yields together with Assumption \ref{StrucAssumptionV} and Assumption \ref{AssumpCoeffApprox}
\begin{align*}
    \Vert \bar{\mu}^\theta(\bar{V}_{\floor{r}}) - \bar{\mu}_\varepsilon(\tilde{V}^\varepsilon_{\floor{r}},\theta) \Vert^2 & \leq 2 \Vert \bar{\mu}^\theta(\bar{V}_{\floor{r}}) - \bar{\mu}^\theta(\tilde{V}^\varepsilon_{\floor{r}}) \Vert^2 + 2 \Vert \bar{\mu}^\theta(\tilde{V}^\varepsilon_{\floor{r}}) - \bar{\mu}_\varepsilon(\tilde{V}^\varepsilon_{\floor{r}},\theta) \Vert^2\\
    &\leq 2L\Vert \bar{V}_{\floor{r}} - \tilde{V}^\varepsilon_{\floor{r}} \Vert^2 + 2C\varepsilon^{2l+1} d^p (1+\Vert \tilde{V}^\varepsilon_{\floor{r}}\Vert^2 + \Vert \theta \Vert^2)
\end{align*}

and
\begin{align*}
    \Vert \bar{\sigma}^\theta(\bar{V}_{\floor{r}}) - \bar{\sigma}_\varepsilon(\tilde{V}^\varepsilon_{\floor{r}},\theta) \Vert^2_F &\leq 2 \Vert \bar{\sigma}^\theta(\bar{V}_{\floor{r}}) - \bar{\sigma}^\theta(\tilde{V}^\varepsilon_{\floor{r}}) \Vert^2_F + 2 \Vert \bar{\sigma}^\theta(\tilde{V}^\varepsilon_{\floor{r}}) - \bar{\sigma}_\varepsilon(\tilde{V}^\varepsilon_{\floor{r}},\theta) \Vert^2_F \\
    &\leq 2L \Vert \bar{V}_{\floor{r}} - \tilde{V}^\varepsilon_{\floor{r}} \Vert^2 + 2C\varepsilon^{2l+1} d^p (1+\Vert \tilde{V}^\varepsilon_{\floor{r}}\Vert^2 + \Vert \theta\Vert^2).
\end{align*}

Using these bounds we see that
\begin{align*}
    G(t) \leq (4t+16)L \int_0^t \mathbb{E}[\Vert \Bar{V}_{\floor{r}} - \Tilde{V}^\varepsilon_{\floor{r}} \Vert^2] \diff{r} + (4t+16)C\varepsilon^{2l+1} \bar{d}^p \int_0^t (1 + \mathbb{E}[\Vert \Tilde{V}^\varepsilon_{\floor{r}}\Vert^2] + \Vert \theta \Vert^2)\diff{r}
\end{align*}

and so with Lemma \ref{SecMomVTild}
\begin{align*}
    G(t) \leq (4T+16)L\int_0^t G(s)\diff{s} + (4T+16)C\bar{d}^{3p} \varepsilon^{2l+1} T \max(c_{11},1+c_{12}, 1+ c_{13})(1+\Vert v \Vert^2 + \Vert \theta \Vert^2).
\end{align*}

By Gronwall's inequality it finally follows
\begin{align*}
    G(t) &\leq (4T+16)C\bar{d}^{3p}\varepsilon^{2l+1} T \max(c_{11},1+c_{12}, 1+ c_{13})(1+\Vert v \Vert^2 + \Vert \theta \Vert^2) \exp(2(2T+8)LT)\\
    &= c_{18} \bar{d}^{3p} \varepsilon^{2l+1} (1 + \Vert v \Vert^2 + \Vert \theta \Vert^2)
\end{align*}

for $c_{18} := (4T+16) CT \max(c_{11},1+c_{12}, 1+ c_{13})\exp((4T+16)LT)$. 
\end{proof}

\begin{proposition}\label{DiffXBarTilde}
Suppose Assumptions \ref{StrucAssumptionV} and \ref{AssumpCoeffApprox} hold. Then there exists a constant $c_{19}>0$ such that for all $\varepsilon\in(0,1/2)$,
\begin{equation}
    \mathbb{E}\left[\sup_{t\in[0,T]} \Vert \bar{X}_t - \tilde{X}^\varepsilon_t \Vert^2\right] \leq c_{19}  \bar{d}^{3p} \varepsilon^{2l+1} (1+ \Vert v \Vert^2 + \Vert x \Vert^2 + \Vert \theta \Vert^2),
\end{equation}

where the constant $c_{19}$ only depends on $T,C$ and $L$.
\end{proposition}

\begin{proof}
The proof is analogous to the one of Proposition \ref{DiffVBarTilde}.
\end{proof}


\subsection{DNN approximation of the pricing function}
Let $n\in\mathbb{N}$. For each $d\in\mathbb{N}$ we consider a functional $\varphi_d:\mathbb{R}^d \times \mathbb{R}^{nd}\rightarrow\mathbb{R}$. In this section our goal is to approximate the map
\begin{equation}\label{PricingMap}
    (x,v,\theta, K)\mapsto U^{\tau}_{d,r}(x,v,\theta,K) := \mathbb{E}[\varphi_d(X^{x,v,\theta}_{\tau},K)]
\end{equation}

by a neural network for any fixed $[0,T]$-valued stopping time $\tau$ without suffering from the CoD. Note that for $\tau=T$ the map \eqref{PricingMap} describes the dependence between the time zero price of a European payoff with maturity $T$ and the initial values $x$ and $v$ of the price process and variance process, respectively, as well as the model parameters $\theta$ and a payoff parameter $K$. For example $K$ can represent the strike price for a call or put option. Before we show how \eqref{PricingMap} can be approximated by DNNs, we make the assumption that the payoff function itself can be approximated by a DNN.

\begin{assumption}\label{AssumPricingApprox}
    For the constants $C>0$ and $p,q,l\geq 0$ from Assumption \ref{AssumpCoeffApprox} and for each $d\in\mathbb{N}$, $\varepsilon\in(0,1/2)$ there exists a neural network $\phi_{\varepsilon,d}: \mathbb{R}^d \times \mathbb{R}^{nd} \rightarrow \mathbb{R}$ such that for each $d\in\mathbb{N}$ and $\varepsilon\in(0,1/2)$ it holds 
    \begin{enumerate}
        \item[(i)] for all $x \in \mathbb{R}^d$ and $K\in\mathbb{R}^{nd}$
        \begin{equation*}
            \vert \varphi_d(x,K) - \phi_{\varepsilon, d}(x,K)\vert \leq C d^p \varepsilon(1+\Vert x \Vert + \Vert K \Vert)
        \end{equation*}
        
        \item[(ii)] 
        \begin{equation*}
            \text{size}(\phi_{\varepsilon, d}) \leq C d^p \varepsilon^{-q}
        \end{equation*}
        
        \item[(iii)] 
        \begin{equation*}
            \text{Lip}(\phi_{\varepsilon, d}) := \sup_{(x_1,K_1),(x_2,K_2)\in\mathbb{R}^d\times \mathbb{R}^{nd}} \frac{\vert\phi_{\varepsilon,d}(x_1,K_1) - \phi_{\varepsilon,d}(x_2,K_2)\vert}{\Vert x_1 - x_2\Vert + \Vert K_1 - K_2 \Vert} \leq C d^p \varepsilon^{-l}.
        \end{equation*}
    \end{enumerate}
    
    In addition, for each $d\in\mathbb{N}$ it holds for a probability measure $\mu^d$ on $\mathbb{R}^d \times \mathbb{R}^d \times \Theta \times \mathbb{R}^{nd}$ that
    \newpage
    \begin{enumerate}
        \item[(iv)] 
        \begin{equation}\label{Assum3.3.ProbMeasure}
            \int_{\mathbb{R}^d \times \mathbb{R}^d \times \Theta \times \mathbb{R}^{nd}}(1 + \Vert x \Vert^2 + \Vert v \Vert^2 + \Vert \theta \Vert^2 + \Vert K \Vert^2) \mu^d(\diff{x},\diff{v},\diff{\theta},\diff{K}) \leq C d^p.
        \end{equation}
    \end{enumerate}
\end{assumption}

\begin{remark}
    It is worth noting that many relevant payoff functions traditionally considered in mathematical finance, e.g. call and put options on a basket of underlyings, can be explicitly written as DNNs with the ReLU activation function. In particular, in practice it often holds $\varphi_d = \phi_{\varepsilon,d}$ and (ii), (iii) in Assumption \ref{AssumPricingApprox} are satisfied with $q=l=0$. We refer to Section 4 in \cite{Grohs2023} for more details and examples. 
\end{remark}

\begin{theorem}\label{MainTheorem}
Assume that Assumptions \ref{StrucAssumptionV}, \ref{AssumpCoeffApprox} and \ref{AssumPricingApprox} hold. Then for a given $[0,T]$-valued stopping time $\tau$ and $m,n\in\mathbb{N}$ there exist constants $\mathfrak{C},\mathfrak{p},\mathfrak{q}>0$ and for any $d,r\in\mathbb{N}$ and target accuracy $\bar{\varepsilon}\in(0,1/2)$ a neural network $U^{\tau}_{\bar{\varepsilon},d,r}:\mathbb{R}^d \times \mathbb{R}^d \times \Theta \times \mathbb{R}^{nd}$ such that
\begin{itemize}
    \item[(i)] \begin{equation}\label{Theorem4.1.Prop1}
        \text{size}(U^{\tau}_{\bar{\varepsilon},d,r})\leq \mathfrak{C} \bar{d}^\mathfrak{p} \bar{\varepsilon}^{-\mathfrak{q}}
    \end{equation}
    \item[(ii)] 
    \begin{equation}\label{Theorem4.1.Prop2}
        \left( \int \left\vert U^{\tau}_{d,r}(x,v,\theta,K) - U^{\tau}_{\bar{\varepsilon}, d, r}(x,v,\theta,K) \right\vert^2 \mu^d(\diff{x},\diff{v},\diff{\theta},\diff{K}) \right)^{1/2} < \bar{\varepsilon},
    \end{equation}
\end{itemize}
where $\mu^d$ is the probability measure from \eqref{Assum3.3.ProbMeasure}. The constants $\mathfrak{C},\mathfrak{p},\mathfrak{q}$ depend on the constants $L$, $K_1$, $K_2$, $C$, $p$, $l$, $q$ from Assumptions \ref{StrucAssumptionV}, \ref{AssumpCoeffApprox} and \ref{AssumPricingApprox}, but they do not depend on $d,r,\bar{\varepsilon}$.
\end{theorem}

\begin{proof}
Let $\tau$ be a stopping time, $m,n\in\mathbb{N}$, $\bar{\varepsilon}\in(0,1/2)$ and consider $M\in\mathbb{N}$. Then a Monte-Carlo approach would suggest to use the approximation
\begin{equation*}
    U^{\tau}_{d,r}(x,v,\theta,K) \sim \frac{1}{M}\sum_{i=1}^M \phi_{\varepsilon,d}(\tilde{X}^i_{\tau(\omega)}(\omega),K),
\end{equation*}

where $(\tilde{X}^i_{\tau})_{i=1,...,M}$ are i.i.d. copies of $\tilde{X}^{\varepsilon,x,v,\theta}_\tau$, where $\Tilde{X}^{\varepsilon,x,v,\theta}$ satisfies \eqref{Theorem4.1.Prop1}. Therefore, to prove the theorem we now show that there exist $\omega\in\Omega$ such that
\begin{equation*}
    \left( \int_{\mathbb{R}^d\times \mathbb{R}^d \times \Theta \times \mathbb{R}^{nd}} \left\vert U^{\tau}_{d,r}(x,v,\theta,K) - \frac{1}{M}\sum_{i=1}^M \phi_{\varepsilon,d}(\tilde{X}^i_{\tau(\omega)}(\omega),K) \right\vert^2 \mu(\diff{x},\diff{v},\diff{\theta,\diff{K})} \right)^{1/2}<\bar{\varepsilon}
\end{equation*}

and that $\frac{1}{M}\sum_{i=1}^M \phi_{\varepsilon,d}(\tilde{X}^i(\omega)_{\tau(\omega)},K)$ is indeed a neural network satisfying \eqref{Theorem4.1.Prop1}. First of all, we look at the $L^2$-approximation error 
\begin{align*}
    &\int_{\mathbb{R}^d\times \mathbb{R}^d \times \Theta \times \mathbb{R}^{nd}} \mathbb{E}\left[\left\vert U^{\tau}_{d,r}(x,v,\theta,K) - \frac{1}{M}\sum_{i=1}^M \phi_{\varepsilon,d}(\tilde{X}^i_{\tau},K) \right\vert^2  \right]\mu(\diff{x},\diff{v},\diff{\theta},\diff{K})\\[2mm]
    =& \int_{\mathbb{R}^d\times \mathbb{R}^d \times \Theta \times \mathbb{R}^{nd}}\left\vert U^{\tau}_{d,r}(x,v,\theta,K) - \mathbb{E}[\phi_{\varepsilon,d}(\tilde{X}^1_{\tau},K)] \right\vert^2\\[2mm]
    & \phantom{5cm}+ \mathbb{E}\left[\left\vert \mathbb{E}[\phi_{\varepsilon,d}(\tilde{X}^1_{\tau},K)] - \frac{1}{M}\sum_{i=1}^M \phi_{\varepsilon,d}(\tilde{X}^i_{\tau},K) \right\vert^2  \right]\mu(\diff{x},\diff{v},\diff{\theta},\diff{K}) \\[2mm]
    =& \int_{\mathbb{R}^d\times \mathbb{R}^d \times \Theta \times \mathbb{R}^{nd}}\left\vert U^{\tau}_{d,r}(x,v,\theta,K) - \mathbb{E}[\phi_{\varepsilon,d}(\tilde{X}^1_{\tau},K)] \right\vert^2\\[2mm] &\phantom{5cm} + \frac{1}{M}\mathbb{E}\left[\left\vert \mathbb{E}[\bar{\phi}_{\varepsilon,d}(\tilde{X}^1_{\tau},K)] - \bar{\phi}_{\varepsilon,d}(\tilde{X}^1_{\tau},K) \right\vert^2  \right]\mu(\diff{x},\diff{v},\diff{\theta},\diff{K}),
\end{align*}

where $\bar{\phi}_{\varepsilon,d}(\tilde{X}^1_{\tau},K) := \phi_{\varepsilon,d}(\tilde{X}^1_{\tau},K) - \phi_{\varepsilon,d}(0,K)$. We study the two terms of the integrand individually. First, by the triangle inequality and Assumption \ref{AssumPricingApprox} we see that
\begin{align}\label{ProofMainThBound1}
    & \left\vert U^{\tau}_{d,r}(x,v,\theta,K) - \mathbb{E}[\phi_{\varepsilon,d}(\tilde{X}^1_{\tau},K)] \right\vert\nonumber\\[2mm]
    \leq & \mathbb{E}\left[ \left\vert \varphi_d(X^{x,v}_{\tau},K) - \phi_{\varepsilon,d}(X^{x,v}_{\tau},K) \right\vert\right] + \mathbb{E}\left[\left\vert \phi_{\varepsilon,d}(X^{x,v}_{\tau},K) - \phi_{\varepsilon,d}(\tilde{X}^1_{\tau},K) \right\vert \right]\nonumber\\[2mm]
    \leq & \mathbb{E}\left[ C\bar{d}^p\varepsilon\left(1+ \sup_{t\in[0,T]}\Vert X^{x,v}_{t} \Vert + \Vert K \Vert\right) \right] + \mathbb{E}\left[ C \bar{d}^p \varepsilon^{-l} \sup_{t\in[0,T]}\Vert X^{x,v}_{t} - \tilde{X}^1_{t} \Vert \right].
\end{align}

By Propositions \ref{DiffXNormalBar} and \ref{DiffXBarTilde} it holds 
\begin{align}\label{BoundDiffXNormalTildeFirstMoment}
    &\mathbb{E}\left[\sup_{t\in[0,T]}\Vert X^{x,v}_{t} - \tilde{X}^1_{t} \Vert \right] \leq \mathbb{E}\left[\sup_{t\in[0,T]}\Vert X^{x,v}_{t} - \bar{X}^{x,v}_{t} \Vert \right] + \mathbb{E}\left[\sup_{t\in[0,T]}\Vert \bar{X}^{x,v}_{t} - \tilde{X}^1_{t} \Vert \right]\nonumber\\[2mm]
    \leq& \left(h(c_8 \Vert x \Vert^2 \bar{d}^2 + c_9 \Vert v \Vert^2 \bar{d}^4 + c_{10} \bar{d}^6)\right)^{1/2} + \left( c_{19}\varepsilon^{2l+1} \bar{d}^{3p}(1+\Vert v \Vert^2 + \Vert x \Vert^2 + \Vert \theta \Vert^2) \right)^{1/2},
\end{align}

and by Lemma \ref{LemmaBoundSecMomX}
\begin{equation}\label{BoundFirstMomX}
    \mathbb{E}\left[\sup_{t\in[0,T]}\Vert X^{x,v}_{t} \Vert^2 \right] \leq c_3 \Vert x \Vert^2 + c_4 \Vert v \Vert^2 \bar{d}^2 + c_5 \bar{d}^4.
\end{equation}

Note that by Jensen's inequality, \eqref{SquaredTriangular} and \eqref{ProofMainThBound1} it holds
\begin{align*}
    &\left\vert U^{\tau}_{d,r}(x,v,\theta,K) - \mathbb{E}[\phi_{\varepsilon,d}(\tilde{X}^1_{\tau},K)] \right\vert^2\\[2mm]
    \leq & 6C^2 \bar{d}^{2p} \varepsilon^2\left(1+ \mathbb{E}\left[\sup_{t\in[0,T]}\Vert X^{x,v}_{t} \Vert\right]^2 + \Vert K \Vert^2\right) + 2 C^2 \bar{d}^{2p} \varepsilon^{-2l} \mathbb{E}\left[\left(\sup_{t\in[0,T]}\Vert X^{x,v}_{t} - \Tilde{X}^1_{t}\Vert\right)\right]^2.
\end{align*}

Inserting \eqref{BoundDiffXNormalTildeFirstMoment} and \eqref{BoundFirstMomX} in \eqref{ProofMainThBound1} therefore yields
\begin{align*}
    \vert &U^{\tau}_{d,r}(x,v,\theta,K) - \mathbb{E}[\phi_{\varepsilon,d}(\tilde{X}^1_{\tau},K)]\vert^2\\[2mm]
    \leq & 6C^2 \bar{d}^{2p} \varepsilon^2(1 + c_3 \Vert x \Vert^2 + c_4 \Vert v \Vert^2 \bar{d}^2 + c_5 \bar{d}^4 + \Vert K \Vert^2) \\[2mm]
    &+ 4C^2 \bar{d}^{2p} \varepsilon^{-2l}\bigg( h(c_8 \Vert x \Vert^2 \bar{d}^2 + c_9 \Vert v \Vert^2 \bar{d}^4 + c_{10} \bar{d}^6) + c_{19}\varepsilon^{2l+1} \bar{d}^{3p} (1+\Vert v \Vert^2 + \Vert x \Vert^2 + \Vert \theta \Vert^2) \bigg)\\[2mm]
    \leq & \left[6C^2 \bar{d}^{2p+4}\varepsilon^2\max(1,c_3,c_4,1+c_5) + 4C^2 \bar{d}^{2p+6}\varepsilon^{-2l}h\max(c_8,c_9,c_{10}) + 4C^2 \bar{d}^{5p} \varepsilon c_{19}\right]\\[2mm]
    &(1+\Vert x \Vert^2 + \Vert v \Vert^2 + \Vert \theta\Vert^2 + \Vert K \Vert^2)
\end{align*}

On the other hand, we have 
\begin{align*}
    &\mathbb{E}\left[\left\vert \mathbb{E}[\bar{\phi}_{\varepsilon,d}(\tilde{X}^1_{\tau},K)] - \bar{\phi}_{\varepsilon,d}(\tilde{X}^1_{\tau},K) \right\vert^2  \right] \leq \mathbb{E}\left[\vert \phi_{\varepsilon,d}(\Tilde{X}^1_{\tau},K) - \phi_{\varepsilon,d}(0,K) \vert^2 \right]\\ 
    \leq & C^2 \bar{d}^2 \varepsilon^{-2l} \mathbb{E}\left[ \sup_{t\in[0,T]} \Vert \Tilde{X}^1_{t} \Vert^2\right] \leq C^2 \bar{d}^{2} \varepsilon^{-2l}(c_{14}\Vert x\Vert^2 + c_{15}\Vert v \Vert^2 + c_{16}\Vert \theta \Vert^2 + c_{17}\bar{d}^{2p})\\[2mm]
    \leq & C^2 \bar{d}^{2p+2}\varepsilon^{-2l}\max(c_{14},c_{15},c_{16},c_{17})(1+\Vert x \Vert^2 + \Vert v \Vert^2 + \Vert \theta\Vert^2).
\end{align*}

So it follows that
\begin{align*}
    &\int_{\mathbb{R}^d\times \mathbb{R}^d \times \Theta \times \mathbb{R}^{nd}} \mathbb{E}\left[\left\vert U^{\tau}_{d,r}(x,v,\theta,K) - \frac{1}{M}\sum_{i=1}^M \phi_{\varepsilon,d}(\tilde{X}^i_{\tau},K) \right\vert^2  \right]\mu(\diff{x},\diff{v},\diff{\theta},\diff{K}) \\[2mm]
    &\leq \bigg[6C^2 \bar{d}^{2p+4}\varepsilon^2\max(1,c_3,c_4,1+c_5) + 4C^2 \bar{d}^{2p+6}\varepsilon^{-2l}h\max(c_8,c_9,c_{10})\\[2mm]
    &\phantom{5cm} + 4C^2 \bar{d}^{5p} \varepsilon c_{19} + M^{-1}C^2 \bar{d}^{2p+2}\varepsilon^{-2l}\max(c_{14},c_{15},c_{16},c_{17})\bigg]\\[2mm]
    &\phantom{5cm}\int_{\mathbb{R}^d\times \mathbb{R}^d \times \Theta \times \mathbb{R}^{nd}} (1+\Vert x \Vert^2 + \Vert v \Vert^2 + \Vert \theta\Vert^2 + \Vert K \Vert^2) \mu(\diff{x},\diff{v},\diff{\theta},\diff{K}).
\end{align*}

By \eqref{Assum3.3.ProbMeasure} from Assumption \ref{AssumPricingApprox} we finally get
\begin{align*}
    &\int_{\mathbb{R}^d\times \mathbb{R}^d \times \Theta \times \mathbb{R}^{nd}} \mathbb{E}\left[\left\vert U^{\tau}_{d,r}(x,v,\theta,K) - \frac{1}{M}\sum_{i=1}^M \phi_{\varepsilon,d}(\tilde{X}^i_{\tau},K) \right\vert^2  \right]\mu(\diff{x},\diff{v},\diff{\theta},\diff{K}) \\[2mm] \leq & \bar{C}\bar{d}^{r} [\varepsilon^2 + \varepsilon^{-2l}[h + \varepsilon^{2l+1}] + M^{-1}\varepsilon^{-2l}]
\end{align*}

for $s := \max(3p+4,3p+6,6p,3p+2)$, $\Bar{C}:= 6C^3\max(1, c_3, c_4, 1+c_5,  c_8,c_9, c_{10},c_{14}, c_{15}, c_{16}, c_{17}, c_{18},c_{19})$. Let $\Bar{\varepsilon}\in(0,1)$. Choosing $M:= \lceil 3 \varepsilon^{-2l}\Bar{C}\bar{d}^{s}\Bar{\varepsilon}^{-2} \rceil$, $h:= \bar{\varepsilon}^2(6\Bar{C}\bar{d}^{s}\varepsilon^{-2l})^{-1}$ and $\varepsilon = \Bar{\varepsilon}(4\Bar{C}\bar{d}^{s})^{-1}$, Fubini's theorem finally yields 

\begin{equation*}
    \mathbb{E}\left[\int_{\mathbb{R}^d\times \mathbb{R}^d \times \Theta \times \mathbb{R}^{nd}} \left\vert U^{\tau}_{d,r}(x,v,\theta,K) - \frac{1}{M}\sum_{i=1}^M \phi_{\varepsilon,d}(\tilde{X}^i_{\tau},K) \right\vert^2  \mu(\diff{x},\diff{v},\diff{\theta},\diff{K})\right] < \Bar{\varepsilon},
\end{equation*}

which finishes the first part of the proof.\par

So it remains to be shown that $\frac{1}{M}\sum_{i=1}^M \phi_{\varepsilon,d}(\tilde{X}^i(\omega)_{\tau(\omega)},K)$ is indeed a DNN satisfying \eqref{Theorem4.1.Prop1}. We denote by $l^\sigma_{j} := \text{depth}(\sigma_{\varepsilon,j})$, $l^{\bar{\mu}} := \text{depth}(\bar{\mu}_\varepsilon)$ and $l^{\bar{\sigma}}_{j} := \text{depth}(\bar{\sigma}_{\varepsilon,j})$ for $j=1,...,d$ and introduce 
\begin{equation*}
    l^{\text{max}}_X := \max \lbrace 2, l^\sigma_1,...,l^\sigma_d \rbrace, \; l^{\text{max}}_V := \max \lbrace 2, l^{\bar{\mu}}, l^{\bar{\sigma}}_1,...,l^{\bar{\sigma}}_d \rbrace.
\end{equation*}

For any $l\in\mathbb{N}$ denote by $\mathcal{I}_{d,l}$ a $l$-layer DNN emulating the identity on $\mathbb{R}^d$. By Remark 2.4. of \cite{Petersen2018} we can choose $\mathcal{I}_{d,l}$ with $\text{size}(\mathcal{I}_{d,l})\leq 2dl$. For a fixed $\omega\in\Omega$ we find $k\in\{0,...,N-1\}$ such that $\tau(\omega)\in[t_k,t_{k+1}]$ so that it holds
\begin{align*}
    \tilde{X}^i_{\tau} = \mathcal{I}&_{d,l^{\text{max}}_X}(\tilde{X}^i_{t_k}) 
    + \sum_{j=1}^d \sigma^\varepsilon_j(\mathcal{I}_{d,l^{\text{max}}_X-l^\sigma}(\tilde{X}^i_{t_k}), \mathcal{I}_{d,l^{\text{max}}_X-l^\sigma}(\tilde{V}^i_{t_k}), \mathcal{I}_{md,l^{\text{max}}_X - l^\sigma}(\theta))(W^i_{\tau,j} - W^i_{t_k,j})
\end{align*}

and
\begin{align*}
    \tilde{V}^i_{\tau} = \mathcal{I}&_{d,l^{\text{max}}_V}(\tilde{V}^i_{t_k}) + \bar{\mu}_\varepsilon(\mathcal{I}_{d,l^{\text{max}}_V - l^{\bar{\mu}}}(\tilde{V}^i_{t_k}), \mathcal{I}_{md,l^{\text{max}}_V - l^{\bar{\mu}}}(\theta))(\tau-t_k)\\ &+ \sum_{j=1}^d \bar{\sigma}^\varepsilon_j(\mathcal{I}_{d,l^{\text{max}}_V-l^{\bar{\sigma}}}(\tilde{V}^i_{t_k}), \mathcal{I}_{md,l^{\text{max}}_V - l^{\bar{\sigma}}}(\theta))(B^i_{\tau,j} - B^i_{t_k,j}).
\end{align*}

By Definition 2.10. in \cite{Petersen2022} we know that for any two DNNs $\Phi^1:\mathbb{R}^{d_1}\rightarrow \mathbb{R}^{d_2}$ and $\Phi^2:\mathbb{R}^{d_3}\rightarrow \mathbb{R}^{d_4}$ with $L$ layers there exists a DNN $\text{FP}(\Phi^1,\Phi^2)$ such that for all $x\in\mathbb{R}^{d_1}$, $y\in\mathbb{R}^{d_3}$ it holds
\begin{equation}\label{ParallelisationDNN}
    \text{FP}(\Phi^1,\Phi^2)(x,y) = (\Phi^1(x),\Phi^2(y)).
\end{equation}

This property is referred to as parallelisation of DNNs with different input dimension, and it can be also applied to more than two DNNs. Therefore, we can write 
\begin{align*}
     \tilde{V}^i_{\tau} = \mathcal{I}_{d,l^{\text{max}}_V}(\tilde{V}^i_{\tau_k}) + &\bar{\mu}_\varepsilon(\text{FP}(\mathcal{I}_{d,l^{\text{max}}_V - l^{\bar{\mu}}},\mathcal{I}_{md,l^{\text{max}}_V - l^{\bar{\mu}}})(\tilde{V}^i_{t_k}, \theta))(\tau-t_k)\\[2mm]
     &+ \sum_{j=1}^d \bar{\sigma}^\varepsilon_j(\text{FP}(\mathcal{I}_{d,l^{\text{max}}_V - l^{\bar{\sigma}}},\mathcal{I}_{md,l^{\text{max}}_V - l^{\bar{\sigma}}})(\tilde{V}^i_{t_k},\theta))(B^i_{\tau,j} - B^i_{t_k,j}).
\end{align*}

As shown in Proposition 2.2 from \cite{Opschoor2020}, the composition of a DNN $\Phi^1$ with $L_1$ layers and a DNN $\Phi^2$ with $L_2$ layers can be realized by a DNN $\Phi = \Phi^1 \odot \Phi^2$ with $L_1 + L_2$ layers and $\text{size}(\Phi)\leq 2(\text{size}(\Phi_1) + \text{size}(\Phi_2))$. Moreover, by Lemma 3.2 from \cite{Gonon2021a} a (randomly) weighted sum of DNNs of the same depth $L$ can be again realized by a DNN with depth $L$. Therefore, it holds
\begin{equation*}
    \tilde{V}^i_{\tau(\omega)}(\omega) = \Phi^i_{\tau}\left(\tilde{V}^i_{t_k}(\omega), \theta \right)
\end{equation*}

for a DNN $\Phi^i_{\tau}:\mathbb{R}^d \times \Theta \rightarrow \mathbb{R}^d$ satisfying 
\begin{align}\label{SizePhiTauI}
    \text{size}\left( \Phi^i_{\tau} \right) &\leq \text{size}\left( \mathcal{I}_{d,l^{\text{max}}_V}\right) + \text{size}\left( \bar{\mu}_\varepsilon \odot \text{FP}(\mathcal{I}_{d,l^{\text{max}}_V - l^{\bar{\mu}}},\mathcal{I}_{md,l^{\text{max}}_V - l^{\bar{\mu}}})\right)\nonumber\\[2mm]
    &\phantom{....}+ \sum_{j=1}^d \text{size}\left(\bar{\sigma}_{\varepsilon,j} \odot \text{FP}(\mathcal{I}_{d,l^{\text{max}}_V - l^{\bar{\sigma}}},\mathcal{I}_{md,l^{\text{max}}_V - l^{\bar{\sigma}}})\right) \nonumber \\
    &\leq (10d + 8d^2)l^{\text{max}}_V + 2\text{size}(\bar{\mu}_\varepsilon) + 2 \sum_{j=1}^d \text{size}(\bar{\sigma}_{\varepsilon,j})\nonumber\\
    &\leq (1 + 6\bar{d} + 4\bar{d}^2)2mC \bar{d}^p  \varepsilon^{-q}.
\end{align}

Analogously we can define DNNs $\Phi^i_{t_j}:\mathbb{R}^d \times \mathbb{R}^{md}\rightarrow \mathbb{R}^d$, $j=1,...,k$, so that
\begin{equation}\label{DNNRecurV}
    \Tilde{V}^i_{t_j}(\omega) = \Phi^i_{t_j}\left(\Tilde{V}^i_{t_{i-1}}(\omega),\theta\right)
\end{equation}

with size$(\Phi^i_{t_j})$ also satisfying \eqref{SizePhiTauI} and $l^i_j:= \text{depth}(\Phi^i_{t_j})$ for $j=1,...,k$. Therefore, we get recursively
\begin{equation*}
     \tilde{V}^i(\omega)_{\tau(\omega)} = \Phi^i_{\tau}\left( \Phi^i_{t_k}\left(\tilde{V}^i_{t_{k-1}}(\omega), \theta \right), \mathcal{I}_{md,l^i_k}(\theta) \right) = \Phi^i_{\tau}\left( \tilde{\Phi}^i_{t_k} \left( \tilde{V}^i_{t_{k-1}}(\omega), \theta \right) \right),
\end{equation*}

where we can choose $\tilde{\Phi}^i_{t_k}:=\text{FP}(\Phi^i_{t_k},\mathcal{I}_{md,l^i_k})$. Continuing this pattern yields 
\begin{equation*}
    \tilde{V}^i(\omega)_{\tau(\omega)} = \left( \Phi^i_{\tau} \circ \tilde{\Phi}^i_{t_k} \circ \dots \circ \tilde{\Phi}^i_{t_1} \right)(v,\theta) = \bar{\Phi}^i_{\tau}(v,\theta)
\end{equation*}

for $\bar{\Phi}^i_{\tau} := \Phi^i_{\tau} \odot \tilde{\Phi}^i_{t_k} \odot \dots \odot \tilde{\Phi}^i_{t_1}$. From Proposition 2.2 in \cite{Opschoor2020} it holds 
$$\text{size}(\Phi^1 \odot \Phi^2) \leq 2 \text{size}(\Phi^1) + \text{size}_{out}(\Phi^2) + \text{size}(\Phi^2)$$
and 
$$
\text{size}_{out}(\Phi^1 \odot \Phi^2) = \text{size}_{out}(\Phi^1),
$$
if $\Phi^1$ has at least one hidden layer. By using \eqref{SizePhiTauI} for $k=0,...,N-1$ we get
\begin{align*}
    \text{size}(\bar{\Phi}^i_{\tau}) &\leq 2 \text{size}(\Phi^i_{\tau}) + \text{size}_{out}(\tilde{\Phi}^i_{t_k}\odot \dots \odot\tilde{\Phi}^i_{t_1}) + \text{size}(\tilde{\Phi}^i_{t_k} \odot \dots \odot \tilde{\Phi}^i_{t_1}) \\ &=
    2\text{size}(\Phi^i_{\tau}) + \text{size}_{out}(\tilde{\Phi}^i_{t_k}) + \text{size}(\tilde{\Phi}^i_k \odot \dots \odot \tilde{\Phi}^i_{t_1}) \\ &\leq 2\text{size}(\Phi^i_{\tau}) + \text{size}_{out}(\tilde{\Phi}^i_{t_k}) + 2\text{size}(\tilde{\Phi}^i_{t_k}) + \text{size}_{out}(\tilde{\Phi}^i_{t_k}) + \text{size}(\tilde{\Phi}^i_{t_{k-1}} \odot \dots \odot \tilde{\Phi}^i_{t_1}) \\ &\leq ... \\ &\leq 2 \text{size}(\Phi^i_{\tau}) + 3\sum_{j=1}^k \text{size}(\tilde{\Phi}^i_{t_j}) =  2\text{size}(\Phi^i_{\tau}) + 3\sum_{j=1}^k [\text{size}(\Phi^i_{t_j}) + \text{size}(\mathcal{I}_{d,l^\text{max}_V})] \\ 
    &\leq 2 \text{size}(\Phi^i_{\tau}) + 3\sum_{j=1}^k \text{size}(\tilde{\Phi}^i_{t_j}) =  2\text{size}(\Phi^i_{\tau}) + 3\sum_{j=1}^k (1+2md)\text{size}(\Phi^i_{t_j}) \\[2mm]
    &\leq 2C \bar{d}^p \varepsilon^{-q}(1+6\bar{d} + 4 \bar{d}^2)(2+3k(1+m\bar{d})).
\end{align*} 

Similarly to above, we can write 
\begin{align*}
    \tilde{X}^i_{\tau} = \mathcal{I}_{d,l^\text{max}_X}(\tilde{X}^i_{t_k}) & + \mu_\varepsilon\left(\text{FP}(\mathcal{I}_{d,l^\text{max}_X - l^\mu},\mathcal{I}_{md,l^\text{max}_X - l^\mu})(\tilde{X}^i_{t_k}, \theta) \right)(\tau -t_k)\\[2mm]
    & + \sum_{j=1}^d \sigma^\varepsilon_j\left(\text{FP}(\mathcal{I}_{d,l^\text{max}_X - l^\sigma},\mathcal{I}_{d,l^\text{max}_X - l^\sigma},\mathcal{I}_{md,l^\text{max}_X - l^\sigma})(\tilde{X}^i_{t_k}, \tilde{V}^i_{t_k}, \theta)\right)(W^i_{\tau,j} - W^i_{t_k,j}).
\end{align*}

Again, since the composition and weighted sums of DNNs can be again expressed as a DNN, it holds 
\begin{equation}\label{DNNRecurX}
    \tilde{X}^i(\omega)_{\tau(\omega)} = \Psi^i_{\tau}\left( \tilde{X}^i_{t_k}(\omega), \tilde{V}^i_{t_k}(\omega), \theta \right)
\end{equation}

for a DNN $\Psi^i_{\tau}:\mathbb{R}^d \times \mathbb{R}^d \times \Theta \rightarrow \mathbb{R}^d$ satisfying
\begin{align}\label{SizePsi}
    \text{size}\left( \Psi^i_{\tau} \right) &\leq \text{size}(\mathcal{I}_{d,l^\text{max}_X}) + \text{size}(\mu_\varepsilon \odot \text{FP}(\mathcal{I}_{d,l^\text{max}_X - l^\mu},\mathcal{I}_{md,l^\text{max}_X - l^\mu}))\\[2mm] &\phantom{....} + \sum_{j=1}^d \text{size}(\sigma_{\varepsilon,j} \odot \text{FP}(\mathcal{I}_{d,l^\text{max}_X - l^\sigma},\mathcal{I}_{d,l^\text{max}_X - l^\sigma},\mathcal{I}_{md,l^\text{max}_X - l^\sigma})) \nonumber \\[2mm]
    &\leq (1 + 6 \bar{d} + 6\bar{d}^2)2mC\bar{d}^p \varepsilon^{-q},
\end{align}

where we have used Assumption \ref{AssumpCoeffApprox}. Again we can define DNNs $\Psi^i_{t_j}:\mathbb{R}^d \times \mathbb{R}^d \times \mathbb{R}^{md} \rightarrow \mathbb{R}^d$, $j=1,...,k$, such that
\begin{equation*}
    \Tilde{X}^i_{t_j}(\omega) = \Psi^i_{t_j}\left(\Tilde{X}^i_{t_{j-1}}(\omega), \Tilde{V}^i_{t_{j-1}}(\omega), \theta \right)
\end{equation*}

with size$(\Psi^i_{t_j})$ satisfying \eqref{SizePsi} and 
$\Tilde{l}^i_j := \text{depth}(\Psi^i_{t_j}).$ Using \eqref{DNNRecurV} and \eqref{DNNRecurX}, we get recursively
\begin{equation*}
    \tilde{X}^i_{\tau}(\omega)(\omega) = \Psi^i_{\tau}\left(\Psi^i_{t_k}\left(\tilde{X}^i_{t_{k-1}}(\omega), \tilde{V}^i_{t_{k-1}}(\omega), \theta\right), \mathcal{I}_{d,\tilde{l}^i_{k} - l^i_k}\circ\Phi^i_{t_k}\left(\tilde{V}^i_{t_{k-1}}(\omega), \theta\right), \mathcal{I}_{md,\tilde{l}^i_{k}}(\theta)\right).
\end{equation*}

We can set $\tilde{\Psi}^i_{t_k}:=\text{FP}(\Psi^i_{t_k},\mathcal{I}_{d,\tilde{l}^i_k - l^i_k} \odot \Phi^i_{t_k},\mathcal{I}_{md,\tilde{l}^i_k})$. Then
\begin{equation*}
    \tilde{X}^i_{\tau(\omega)}(\omega) = \left(\Psi^i_{\tau}\circ \tilde{\Psi}^i_{t_k} \circ \dots \circ \tilde{\Psi}^i_{t_1}\right)(x,v,\theta) = \bar{\Psi}^i_{\tau}(x,v,\theta)
\end{equation*}

for $\bar{\Psi}^i_{\tau} := \Psi^i_{\tau}\odot \tilde{\Psi}^i_{t_k} \odot \dots \odot \tilde{\Psi}^i_{t_1}$. By using Proposition 2.2 from \cite{Opschoor2020} and \eqref{SizePsi}, we obtain
\begin{align}\label{SizePsiBar}
    &\text{size}(\bar{\Psi}^i_{\tau})\nonumber\\[1.5mm]  \leq& 2\text{size}(\Psi^i_{\tau}) + \text{size}_{out}(\tilde{\Psi}^i_{t_k} \odot \dots \odot \tilde{\Psi}^i_{t_1}) + \text{size}(\tilde{\Psi}^i_{t_k} \odot \dots \odot \tilde{\Psi}^i_{t_1}) \nonumber \\[1.5mm]
    = &2\text{size}(\Psi^i_{\tau}) + \text{size}_{out}(\tilde{\Psi}^i_{t_k}) + \text{size}(\tilde{\Psi}^i_{t_k} \odot \dots \odot \tilde{\Psi}^i_{t_1}) \nonumber \\[1.5mm]
    \leq & 2\text{size}(\Psi^i_{\tau}) + \text{size}_{out}(\tilde{\Psi}^i_{t_k}) + 2\text{size}(\tilde{\Psi}^i_{t_k}) + \text{size}_{out}(\tilde{\Psi}^i_{t_{k-1}}) + \text{size}(\tilde{\Psi}^i_{t_{k-1}} \odot \dots \odot \tilde{\Psi}^i_{t_1}) \nonumber \\[1.5mm]
    \leq & ... \nonumber\\
    \leq & 2\text{size}(\Psi^i_{\tau}) + 3 \sum_{j=1}^k [\text{size}(\Psi^i_{t_j}) + \text{size}(\mathcal{I}_{d,\tilde{l}^i_j - l^i_j} \odot \Phi^i_{t_j}) + \text{size}(\mathcal{I}_{md,\tilde{l}^i_j})] \nonumber \\[2mm]
    \leq & (1+6\bar{d} + 6\bar{d}^2)2mC\bar{d}^p \varepsilon^{-q} (2+6mk + 16mk\bar{d}).
\end{align}

Therefore, we can finally write for $\bar{l}^i_{\tau} := \text{depth}(\bar{\Psi}^i_{\tau})$
\begin{equation*}
    \frac{1}{M}\sum_{i=1}^M \phi_{\varepsilon,d}(\Tilde{X}^i(\omega)_{\tau(\omega)}, K) = \frac{1}{M}\sum_{i=1}^M \phi_{\varepsilon,d}\left(\Bar{\Psi}^i_{\tau}(x,v,\theta), \mathcal{I}_{n,\bar{l}^i_{\tau}}(K)\right).
\end{equation*}

By parallelisation it therefore follows
\begin{equation*}
    \frac{1}{M}\sum_{i=1}^M \phi_{\varepsilon,d}\left(\Bar{\Psi}^i_{\tau}(x,v,\theta), \mathcal{I}_{md,\bar{l}^i_{\tau}}(K)\right) = \frac{1}{M}\sum_{i=1}^M \phi_{\varepsilon,d}(\text{FP}(\bar{\Psi}^i_\tau, \mathcal{I}_{n,\bar{l}^i_\tau})(x,v,\theta,K)) = \frac{1}{M}\sum_{i=1}^M \bar{\Xi}^i_{\tau}(x,v,\theta,K)
\end{equation*}

for $\Bar{\Xi}^i_{\tau} := \phi_{\varepsilon,d} \odot \text{FP}(\bar{\Psi}^i_\tau, \mathcal{I}_{n,\bar{l}^i_\tau})$. By Lemma 3.2 in \cite{Gonon2021a} the sum of neural networks can be expressed by a neural network. Therefore, we can indeed choose $U^{\tau}_{\bar{\varepsilon},d,r}(x,v,\theta,K) = \frac{1}{M}\sum_{i=1}^M \Bar{\Xi}^i_{\tau}(x,v,\theta,K)$. It only remains to show that for this choice it holds $\text{size}(U^{\tau}_{\bar{\varepsilon},d,r})\leq \mathfrak{C} \bar{d}^\mathfrak{p}  \bar{\varepsilon}^{-\mathfrak{q}}$ for some constants $\mathfrak{C}, \mathfrak{p}, \mathfrak{q} > 0$. From Assumption \ref{AssumPricingApprox} and \eqref{SizePsiBar} it follows that
\begin{align*}
    &\text{size}(U^{\tau}_{\bar{\varepsilon},d,r}) \leq 2M \text{size}(\phi_{\varepsilon,d}) + 2M \text{size}(\mathcal{I}_{n,\bar{l}^i_\tau}) + 2\sum_{i=1}^M \text{size}(\bar{\Psi}^i_{\tau}) \\
    \leq & 2M\left[C\bar{d}^p \varepsilon^{-q} + (2n+1)(1+6\bar{d} + 6\bar{d}^2)2mC\bar{d}^p\varepsilon^{-q}(2+6mk+16mk\bar{d})\right] \\[2mm]
    \leq & 4MNCm\bar{d}^p \varepsilon^{-q}[1+(2n+1)(1+6\bar{d}+6\bar{d}^2)(8+16\bar{d})] \\[2mm]
    \leq & 15024 N\bar{C}Cmn\varepsilon^{-(q+2l)}\bar{d}^{p+3+s}\bar{\varepsilon}^{-2} \\[2mm]
    \leq & 90144 T m n C \bar{C}^{4+2l+q} 4^{2l+q+2} \bar{\varepsilon}^{-(4+4l+q)}\bar{d}^{p+3+s(2+4l+q)},
\end{align*}

where we used our choices for $M$, $h$ $\varepsilon$ and $r$ from the previous part of this proof. The proof is finished by setting $\mathfrak{C} := 90144 T m n C \bar{C}^{4+2l+q} 4^{2l+q+2}$, $\mathfrak{p} := p+3+s(2+4l+q)$ and $\mathfrak{q} := 4+4l+q$.
\end{proof}

\begin{remark}
For $\tau=T$ in Theorem \ref{MainTheorem} we get the existence of a DNN $U_{\bar{\varepsilon},d,m}$ approximating the time zero value of a European option with fixed maturity $T$, i.e. 
$$
    U_{\bar{\varepsilon},d,m} \sim \mathbb{E}[\varphi_d(X^{x,v}_T,K)],
$$
as a function of the spot price and variance of the risky asset, the parameters of the model and some payoff factor $K$ in the sense of Theorem \ref{MainTheorem}.
\end{remark}



\section{Deep Calibration of a Cross-Correlated Heston Model (CCH)}\label{Section4}
In this section, we want to apply our results to a multivariate version of the classical Heston model firstly proposed in \cite{Heston1993}. On the filtered probability space $(\Omega,\mathcal{F},\mathbb{F}=(\mathcal{F}_t)_{t\in[0,T]},\mathbb{Q})$ introduced in Section \ref{Section3}, let $\tilde{W}^0,\tilde{W}^1,..., \tilde{W}^d, \tilde{B}^1,...,\tilde{B}^d$ by independent one-dimensional standard $\mathbb{F}$-Brownian motions. We assume that the $d$-dimensional risky assets process $X = (X^1_t,...,X^d_t)^T_{t\in[0,T]}$ with variance process $V = (V^1_t,...,V^d_t)^T_{t\in[0,T]}$ satisfy the system 
\begin{equation}\label{CCHeston}
\begin{cases}
    \diff X^i_t = \sqrt{V^i_t}X^i_t \diff{\bar{W}}^{i}_t,\; X^i_0 > 0, \\[3mm]
    \diff V^i_t  = a_i(b_i - V^i_t) \diff{t} + \nu_i\sqrt{V^i_t} \diff{\bar{B}}^i_t,\; V^i_0 > 0,\\[3mm]
    \bar{W}^i_t = \rho_{X,i} \tilde{W}^0_t + \sqrt{1-\rho_{X,i}^2} \tilde{W}^i_t\\[3mm]
    \bar{B}^i_t = \rho_{V,i} \bar{W}^i_t + \sqrt{1-\rho_{V,i}^2} \tilde{B}^i_t
\end{cases}
\end{equation}

for the parameters $a_i, b_i, \nu_i >0$ and $\rho_{X,i},\rho_{V,i}\in[-1, 1]$. Thus the model parameters to calibrate are given by the vector $\theta = (a_1, b_1, \nu_1, \rho_{X,1}, \rho_{V,1},...,a_d, b_d, \nu_d, \rho_{X,d}, \rho_{V,d})^T \in \Theta^{\text{CCH}} \subset \mathbb{R}^{5d}$, where $\Theta^{\text{CCH}}$ denotes a given parameter space for the cross-correlation Heston model. We specify $\Theta^{\text{CCH}}$ later in more detail. Note that each single asset-variance pair $(X^i,V^i)$ follows a classical Heston model.

\begin{remark}
The common Brownian motion $\tilde{W}^0$ can be interpreted as a underlying stochastic driver of the entire financial market affecting each asset $X^i$, $i=1,...,d$, whereas $\tilde{W}^i$ is an idiosyncratic noise reflecting fluctuations individual to the $i$-th asset.
\end{remark}

In order to apply the results of Section \ref{Section3}, we now rewrite system \eqref{CCHeston} by explicitly including the correlation parameters $\rho_{X,1},...,\rho_{X,d}$ and $\rho_{V,1},...,\rho_{V,d}$ in the coefficient functions. Therefore, we introduce the modified system 
\begin{equation}\label{EnlargedCCHeston}
    \begin{cases}
        \diff X^{x,v}_t = \sigma^{\theta,\text{CCH}}(X^{x,v}_t,V^v_t) \diff{W}_t \\[2mm]
        \diff V^v_t  = \bar{\mu}^{\theta,\text{CCH}}(V^v_t)\diff{t} + \bar{\sigma}^{\theta,\text{CCH}}(V^v_t) \diff{W}_t \\[2mm]
        W_t = \left(\tilde{W}^0_t,\tilde{W}^1_t,...,\tilde{W}^d_t,\tilde{B}^1_t,...,\tilde{B}^d_t\right)^T
    \end{cases}
\end{equation}

for coefficient functions $\sigma^{\theta,\text{CCH}}:\mathbb{R}^d \times \mathbb{R}^d \rightarrow \mathbb{R}^{d\times (2d+1)}$, $\bar{\mu}^{\theta, \text{CCH}}:\mathbb{R}^d \rightarrow \mathbb{R}^d$ and $\bar{\sigma}^{\theta,\text{CCH}}: \mathbb{R}^d \rightarrow \mathbb{R}^{d\times (2d+1)}$ defined for any $x,v\in\mathbb{R}^d$ as

\begin{gather*}
    \sigma^{\theta,\text{CCH}}(x,v) := \\[2mm]
    \setcounter{MaxMatrixCols}{11}
    \begin{bNiceMatrix}[parallelize-diags=true, xdots/shorten=5mm]
    \sqrt{v_1}x_1 \rho_{X,1} & \sqrt{v_1} x_1\overline{\rho_{X,1}} & 0 & & \Cdots & 0 & 0 & \Cdots & \Cdots &  \Cdots & 0 \\
    \Vdots & 0 &    & \Ddots & & \Vdots & \Vdots & & & & \Vdots \\
    & & \Ddots & \Ddots & & & & & & & \\
    & \Vdots  & & & & 0 & & & & & \\
    \sqrt{v_d}x_d\rho_{X,d} & 0 & \Cdots & & 0 &  \sqrt{v_d}x_d\overline{\rho_{X,d}} & 0 & \Cdots & \Cdots & \Cdots & 0
    \end{bNiceMatrix},
\end{gather*}
\vspace{5mm}
\begin{equation*}
    \bar{\mu}^{\theta, \text{CCH}}(v) := 
    \begin{bNiceMatrix}
    a_1(b_1 - v_1) \\
    \Vdots \\
    a_d(b_d - v_d)
    \end{bNiceMatrix}
\end{equation*}

and 
\begin{equation*}
    \bar{\sigma}^{\theta,\text{CCH}}(v):=
\end{equation*}
\begin{align*}
    \setcounter{MaxMatrixCols}{11}
    \renewcommand{\arraystretch}{1.8}
    \setlength{\arraycolsep}{3.5pt}
    \begin{bNiceMatrix}[parallelize-diags=true, xdots/shorten=3mm]
    \sqrt{v_1}\nu_1 \rho_{X,1} \rho_{V,1} & \sqrt{v_1} \nu_1 \overline{\rho_{X,1}} \rho_{V,1}  &  0  &  & \Cdots & 0 &  \sqrt{v_1}\nu_1\overline{\rho_{V,1}} & 0 & & \Cdots & 0\\
    \Vdots & 0 &  & \Ddots & & \Vdots & 0 & \Ddots & \Ddots & & \Vdots\\
     & & \Ddots & \Ddots & & & \Ddots & \Ddots & & & \\
     & \Vdots & & & & 0 & \Vdots & & & & 0\\
     \sqrt{v_d}\nu_d \rho_{X,d} \rho_{V,d} & 0 & \Cdots & & 0 & \sqrt{v_d} \nu_d \overline{\rho_{X,d}} \rho_{V,d} & 0 & \Cdots & &  0 & \sqrt{v_d}\nu_d\overline{\rho_{V,d}}\\
    \end{bNiceMatrix}
\end{align*}

with $\overline{\rho_{X,i}} := \sqrt{1-\rho_{X,i}^2}$ and $\overline{\rho_{V,i}} := \sqrt{1-\rho_{V,i}^2}$. The SDE \eqref{EnlargedCCHeston} is now a special case of \eqref{ModelDynamics} with state space dimension $d$ as well as $r=2d+1$, $\bar{d} = 2d+1$ and $m=5$. In addition, we make the usual assumption that the Feller ratios satisfy
$$
    \frac{2a_i b_i}{\nu^2_i} > 0,\; i=1,...,d,
$$

which ensures that the processes $V^i$ are strictly positive. Note that the coefficient functions of \eqref{EnlargedCCHeston} do not satisfy the properties in Assumptions \ref{StrucAssumptionV} and \ref{AssumpCoeffApprox} due to the square root function. However, we can guarantee that the coefficient functions are Lipschitz and can be approximated by DNNs, by restricting their values to some compact intervals bounded away from the origin. Then we can apply Theorem \ref{MainTheorem} to approximate the option prices of European options on $X^{x,v}$ by a DNN. In a first step we consider some compact intervals $\mathcal{X},\mathcal{V}\subseteq \mathbb{R}^{+}$ chosen such that the initial conditions satisfy $x\in\mathcal{X}^d:= \mathcal{X}\times\dots\times\mathcal{X}$ and $v\in\mathcal{V}^d:=\mathcal{V}\times \dots\times\mathcal{V}$ respectively. Then the functions $\sigma^{\Theta,\text{CCH}}\vert_{\mathcal{X}^d \times \mathcal{V}^d}$, $\bar{\mu}^{\Theta,\text{CCH}}\vert_{\mathcal{V}^d}$ and $\bar{\sigma}^{\Theta,\text{CCH}}\vert_{\mathcal{V}^d}$ restricted to the compact sets $\mathcal{X}^d \times \mathcal{V}^d$ and $\mathcal{V}^d$ respectively, satisfy the hypotheses of Assumption \ref{StrucAssumptionV}.


\subsection{DNN approximation of the coefficient functions}
We need to prove the existence of ReLU DNNs approximating the coefficient functions of the SDE in \eqref{EnlargedCCHeston} with DNN sizes that grow at most  polynomially in $\bar{d}$. For a general introduction to ReLU neural network approximation we refer to \cite{Yarotsky2017}, \cite{Petersen2022}, \cite{Elbrachter2021} and \cite{Petersen2018}. Note that quantitative uniform approximation results assume a bounded domain for the function to be approximated. Since we want to approximate the option prices as functions not only of the initial values of the processes $X^{x,v}$ and $V^v$ and the strike price but also of the model parameters $\theta$, we assume that the model parameter can only vary in a certain bounded domain. Note that this a common assumption in model calibration as described in Section 3.2.2 in \cite{Horvath2020}.

\begin{assumption}\label{AssumBoundedParaCCH}
    There exist constants $\bar{a}, \bar{b}$ and $\bar{\nu}>0$ such that $0<a_i\leq \bar{a}, 0<b_i\leq \bar{b}$ and $0<\nu_i \leq \bar{\nu}$ and $a_i,b_i,\nu_i$ are such that the condition for the existence for the second moment of $X^i$ from Proposition 3.1 in \cite{Andersen2006} is satisfied for all $i=1,\dots,d$.
\end{assumption}

The parameter space is thus given by $\Theta^{\text{CCH}} := (0,\bar{a}]^d \times (0,\bar{b}]^d \times (0,\bar{\nu}]^d \times [-1,1]^d \times [-1,1]^d$. The approximation of the coefficient functions is mainly based on the fact that any one-dimensional function in $\mathcal{C}^1$ defined on a compact domain can be approximated arbitrarily well by a DNN as in Lemma \ref{DNNApproxC1}. Recall that for a constant $K>0$ and a function $f\in \mathcal{C}^1([-K,K]^d)$ the $\mathcal{C}^1$-norm is defined as
$$
    \Vert f \Vert_{\mathcal{C}^1([-K,K]^d)} := \max\{ \Vert f \Vert_\infty, \Vert \partial_1 f \Vert_\infty,..., \Vert \partial_d f \Vert_\infty \},
$$

where $\partial_i f$, $i=1,...,d$, denotes the first partial derivative of $f$ with respect to the $i$-th variable. For notational convenience we now introduce a single-layered DNN only performing a linear transformation on the input parameters.

\begin{notation}
    For fixed $i,j\in\mathbb{N}$ we denote by $\Phi^{\text{LT}}_{A,b}:\mathbb{R}^i \rightarrow \mathbb{R}^j$ the single layered DNN defined as
    \begin{equation}\label{LinearTransNetwork}
        \Phi^{\text{LT}}_{A,b}(x) = Ax + b,\;x\in\mathbb{R}^i,
    \end{equation}
    for a weight matrix $A\in\mathbb{R}^{j\times i}$ and bias $b\in\mathbb{R}^j$. In case that $b=0$ we just write $\Phi^{\text{LT}}_A$. Note that $\text{size}(\Phi^{\text{LT}}_{A,b})\leq j(i+1)$.
\end{notation}

\begin{lemma}\label{DNNApproxC1}
Let $K>0$, $M\geq 1$, $n\in\mathbb{N}$ and $p\in(0,\infty]$. Then there exists a constant $C>0$ such that for all $f\in \mathcal{C}^1([-K,K]^n)$ with $\Vert f \Vert_{\mathcal{C}^1([-K,K]^n)} \leq M$ and every $\varepsilon\in(0,1/2)$ there exists a DNN $\Phi^{f,\varepsilon}$ such that 
$$
    \Vert f - \Phi^{f,\varepsilon}\Vert_{L^p([-K,K]^n)} \leq \varepsilon
$$

with size$(\Phi^{f,\varepsilon})\leq C \varepsilon^{-2n}$ and depth$(\Phi^{f,\varepsilon})\leq C\varepsilon^{-1}$.
\end{lemma}

\begin{proof}
Let $K>0$, $M\geq 1$, $n\in\mathbb{N}$ and $p\in(0,\infty]$. Consider $f\in \mathcal{C}^1([-K,K]^n)$ with $\Vert f \Vert_{\mathcal{C}^1([-K,K]^n)}\leq M$ and define the function $g\in \mathcal{C}^1([0,1]^n)$ as
$$
    g(x) = f\left( (2\mathcal{I}_n x - \mathbf{1}_n)K \right),
$$

where $\mathcal{I}_n\in\mathbb{R}^{n\times n}$ denotes the identity matrix and $\mathbf{1}_n = (1,...,1)\in\mathbb{R}^n$. Then for the function $h:= g/M \in \mathcal{C}^1([0,1]^n)$ it holds $\Vert h \Vert_{\mathcal{C}^1([0,1]^n)}\leq 1$. So by Theorem 3.19 in \cite{Petersen2022} we find a constant $C^\prime>0$ such that for every $\hat{\varepsilon}\in(0,1/(2M))$ there exists a DNN $\Phi^{h,\hat{\varepsilon}}:[0,1]^n\rightarrow\mathbb{R}$ with size$(\Phi^{h,\hat{\varepsilon}})\leq C^\prime \hat{\varepsilon}^{-2n}$ and depth$(\Phi^{h,\hat{\varepsilon}})\leq C^\prime \hat{\varepsilon}^{-1}$ so that it holds 
$$
    \Vert h - \Phi^{h,\hat{\varepsilon}}\Vert_{L^p([0,1]^n)} \leq \hat{\varepsilon}.
$$

Consider $\varepsilon := \hat{\varepsilon}M \in (0,1/2)$ and the DNN $\Phi^{g,\varepsilon} := M \Phi^{h,\hat{\varepsilon}}$. Note that it holds
$$
    \Vert g - \Phi^{g,\varepsilon} \Vert_{L^p([0,1]^n)} = M \Vert h - \Phi^{h,\hat{\varepsilon}}\Vert_{L^p([0,1]^n)} \leq M \hat{\varepsilon} = \varepsilon.
$$

For $A:= \frac{1}{2K}\mathcal{I}_n$ and $b:= \frac{1}{2}\mathbf{1}_n$ we introduce the DNN $\Phi^{f,\varepsilon} := \Phi^{g,\varepsilon} \odot \Phi^{\text{LT}}_{A,b}:[-K,K]^n\rightarrow\mathbb{R}$, where $\Phi^{\text{LT}}_{A,b}$ is defined in \eqref{LinearTransNetwork}. So for all $z\in[-K,K]^n$ we have
$$
    \Vert f - \Phi^{f,\varepsilon} \Vert_{L^p([-K,K]^n)} = \left\Vert g - \Phi^{g,\varepsilon}\right\Vert_{L^p([0,1]^n)} \leq \varepsilon.
$$

Finally, by using that size$(\Phi^{g,\varepsilon}) = 2n$ we obtain that
$$
    \text{size}(\Phi^{f,\varepsilon}) \leq 2\text{size}(\Phi^{g,\varepsilon}) + 2\text{size}(\Phi^{\text{LT}}_{A,b}) \leq 2C^\prime \hat{\varepsilon}^{-2n} + 4n \leq 4\max(C^\prime M^{2n}, 2n)\varepsilon^{-2n}
$$

and with Definition 2.9 in \cite{Petersen2022} that
$$
    \text{depth}(\Phi^{f,\varepsilon}) = \text{depth}(\Phi^{g,\varepsilon}) + \text{depth}(\Phi^{\text{LT}}_{A,b}) - 1 = \text{depth}(\Phi^{g,\varepsilon}) \leq C^\prime M\varepsilon^{-1}.
$$
So the result follows by choosing $C := 4\max(C^\prime M^{2n}, 2n)$.
\end{proof}

\begin{remark}\label{CoDInDomainRemark}
It is worth noting that the size of the approximating DNN in Lemma \ref{DNNApproxC1} grows at most exponentially in the dimension $n$ of the underlying domain.
\end{remark}

Using Lemma \ref{DNNApproxC1} we can now prove the existence of DNNs approximating the coefficient functions of the SDE \eqref{EnlargedCCHeston} in the sense of Assumption \ref{AssumpCoeffApprox}, as shown in the following Lemma.

\begin{lemma}\label{ApproxCCHestonCoeff}
Under Assumption \ref{AssumBoundedParaCCH} and with the choice $m = 5$ in \eqref{ParameterVector}, there exist constants $C>0$ and $l,p,q\geq 0$ and for $\Bar{d} = 2d+1$ and any $\varepsilon\in(0,1/2)$ there exist neural networks $\bar{\mu}_{\varepsilon}^{\text{CCH}}:\mathcal{V}^d\times\Theta^{\text{CCH}}\rightarrow\mathbb{R}^d$, $\sigma_{\varepsilon,j}^{\text{CCH}}:\mathcal{X}^d\times\mathcal{V}^d\times\Theta^{\text{CCH}}\rightarrow\mathbb{R}^d$, $\bar{\sigma}_{\varepsilon,j}^{\text{CCH}}:\mathcal{V}^d\times\Theta^{\text{CCH}}\rightarrow\mathbb{R}^d$ for $j=1,...,2d+1$ such that for all $x\in\mathcal{X}^d$, $v\in\mathcal{V}^d$ and $\theta\in\Theta^{\text{CCH}}$ it holds
    \begin{enumerate}
        \item[(i)] \begin{align}\label{ApproxCCH01} & \Vert \sigma^{\theta,\text{CCH}}(x,v) - \sigma_\varepsilon^{\text{CCH}}(x,v,\theta) \Vert^2_F + \Vert \bar{\mu}^{\theta,\text{CCH}}(v) - \bar{\mu}_\varepsilon^{\text{CCH}}(v,\theta) \Vert^2 \nonumber \\[2mm] +& \Vert \bar{\sigma}^{\theta,\text{CCH}}(v) - \bar{\sigma}_\varepsilon^{\text{CCH}}(v,\theta) \Vert^2_F \leq  C\varepsilon^{2l+1} \bar{d}^p (1+\Vert x\Vert^2 + \Vert v \Vert^2 + \Vert \theta \Vert^2).\end{align}
        
        \item[(ii)] \begin{align}\label{ApproxCCH02} \Vert \sigma_\varepsilon^{\text{CCH}} (x,v,\theta) \Vert_F + \Vert \bar{\mu}_\varepsilon^{\text{CCH}} (v,\theta) \Vert + \Vert \bar{\sigma}_\varepsilon^{\text{CCH}}(v,\theta) \Vert_F \leq C(\bar{d}^p + \Vert x \Vert + \Vert v \Vert + \Vert \theta \Vert).\end{align}
        
        \item[(iii)] \begin{align}\label{ApproxCCH03}
        \text{size}(\bar{\mu_\varepsilon}^{\text{CCH}}) + \sum_{j=1}^{2d+1} \text{size}(\sigma_{\epsilon,j}^{\text{CCH}}) + \text{size}(\bar{\sigma}_{\varepsilon,j}^{\text{CCH}})\leq C \bar{d}^p \varepsilon^{-q},
        \end{align}
    \end{enumerate}

    where $\Vert \cdot \Vert_F$ denotes the Frobenius norm. The constants $C$, $l$, $p$, $q$ do not depend on $d$ and $\varepsilon$.
\end{lemma}

\begin{proof}
The proof consists of two steps. In a first step, we approximate the one-dimensional non-zero entries of the coefficient functions of \eqref{EnlargedCCHeston}. This can be done using Lemma \ref{DNNApproxC1}. We also exploit the fact that the dimension of the domains depends on $m$ and not on the state space dimension $\bar{d}$. In a second step, we use the parallelisation property of DNNs and the fact that the size of the resulting network grows linearly in the number of networks in the corresponding parallelisation. 

\textit{Step One:} Assume that Assumption \ref{AssumBoundedParaCCH} holds and choose $m=5$. Then we consider the functions
\begin{gather*}
    f_1(x,v,a,b,\nu,\rho_X,\rho_V) = x\sqrt{v}\rho_X,\\[2mm]
    f_2(x,v,a,b,\nu,\rho_X,\rho_V) =  x\sqrt{v}\sqrt{1-\rho_X^2}
\end{gather*}

on $\mathcal{X} \times \mathcal{V} \times (0,\bar{a}] \times (0,\bar{b}] \times (0,\bar{\nu}] \times [-1,1]^2$ and 
\begin{gather*}
    f_3(v,a,b,\nu,\rho_X,\rho_V) = a(b-v),\\[3mm]
    f_4(v,a,b,\nu,\rho_X,\rho_V) = \sqrt{v} \nu \rho_X \rho_V,\\[2mm]
    f_5(v,a,b,\nu,\rho_X,\rho_V) = \sqrt{v} \nu \sqrt{1-\rho_X^2} \rho_V,\\[2mm]
    f_6(v,a,b,\nu,\rho_X,\rho_V) = \sqrt{v} \nu \sqrt{1-\rho_V^2}
\end{gather*}

on $\mathcal{V} \times (0,\bar{a}] \times (0,\bar{b}] \times (0,\bar{\nu}] \times [-1,1]^2$. Note that for $K := \max\{1,\bar{x},\bar{v},\bar{a},\bar{b},\bar{\nu}\}$ it holds that $f_1,f_2 \in \mathcal{C}^1([-K,K]^{m+2})$ and $f_3,f_4,f_5,f_6 \in \mathcal{C}^1([-K,K]^{m+1})$ for $\bar{x}:= \max\{x \mid x \in\mathcal{X}\}$, $\bar{v}:= \max\{v \mid v \in\mathcal{V}\}$ and the constants introduced in Assumption \ref{AssumBoundedParaCCH}. Moreover, since $f_1,...,f_6$ are $\mathcal{C}^1$-functions with a compact domain, there exists a constant $M\geq 1$ such that 
\begin{equation}\label{MBoundCCH}
    \Vert f_i \Vert_{\mathcal{C}^1([-K,K]^{m+2})}, \Vert f_k \Vert_{\mathcal{C}^1([-K,K]^{m+1})} \leq M,\; i=1,2,\; k=3,4,5,6.
\end{equation}

So we can apply Lemma \ref{DNNApproxC1} with $p=\infty$ and get that there is a constant $\Tilde{C}>0$ such that for every $\varepsilon\in(0,1/2)$ there exist DNNs $\Phi^{f_1,\varepsilon},\Phi^{f_2,\varepsilon}:[-K,K]^{m+2} \rightarrow \mathbb{R}$ and $\Phi^{f_3,\varepsilon}, \Phi^{f_4,\varepsilon}, \Phi^{f_5,\varepsilon}, \Phi^{f_6,\varepsilon}:[-K,K]^{m+1} \rightarrow\mathbb{R}$ with
$$
    \Vert f_i - \Phi^{f_i,\varepsilon}\Vert_{L^\infty([-K,K]^{m+2})}  \leq \varepsilon,\; i=1,2, \; \Vert f_i - \Phi^{f_i,\varepsilon}\Vert_{L^\infty([-K,K]^{m+1})}  \leq \varepsilon,\; i=3,...,6
$$

and size$(\Phi^{f_i,\varepsilon})\leq \Tilde{C}\varepsilon^{-2(m+2)}$ for $i=1,...,6$.\par 

\textit{Step Two:} For any DNNs $\Phi^i:\mathbb{R}^{k}\rightarrow \mathbb{R}$, $i=1,...,n$, for some $k,n\in\mathbb{N}$ there a exists a DNN $\Phi^{\text{SUM}}:\mathbb{R}^{k}\rightarrow\mathbb{R}$ such that for all $x\in\mathbb{R}^k$
\begin{equation}\label{SumOfDNNs}
    \Phi^{\text{SUM}}(x) = \sum_{i=1}^k \Phi^i(x)
\end{equation}

with size$(\Phi^{\text{SUM}}) \leq \sum_{i=1}^k \text{size}(\Phi^i)$. Moreover, we define the matrices $I^i\in\mathbb{R}^{7\times 7d}$ and $\tilde{I}_i\in\mathbb{R}^{6\times 6d}$ for $i=1,...,d$ with entries 
\begin{equation*}
    \left[I_i\right]_{l,k} = \mathbbm{1}_{\{k = 7(i-1)+l\}} \text{ and }[\tilde{I}_i]_{l,k} = \mathbbm{1}_{\{k = 6(i-1)+l\}}.
\end{equation*}

Therefore, we can construct the DNN $\sigma^{\text{CCH}}_{\varepsilon,1}:\mathcal{X}^d\times \mathcal{V}^d \times \Theta^{\text{CCH}} \rightarrow \mathbb{R}^d$ as 
$$
    \sigma_{\varepsilon,1}^{\text{CCH}}(x,v,\theta) := \sum_{i=1}^d \left(\Phi^{\text{LT}}_{e_i} \circ \Phi^{f_1,\varepsilon} \circ \Phi^{\text{LT}}_{I_{i}}\right)\left([x_1,v_1,a_1,b_1,\nu_1,\rho_{X,1},\rho_{V,1},...,x_d,v_d,a_d,b_d,\nu_d,\rho_{X,d},\rho_{V,d}]^T\right),
$$

where size$(\Phi^{\text{LT}}_{I_{j-1}}) = 7$, size$(\Phi^{\text{LT}}_{e_i}) = 1$ and $e_k$ denoting the $k$-th unit vector of $\mathbb{R}^d$. Hence, it follows that size$(\sigma_{\varepsilon,1}^{\text{CCH}})\leq 2d \max(4\Tilde{C},18)\varepsilon^{-2(m+2)}$. Then we define $\sigma_{\varepsilon,j}^{\text{CCH}}:\mathcal{X}^d\times \mathcal{V}^d \times \Theta^{\text{CCH}} \rightarrow \mathbb{R}^d$, $j=2,...,d+1$, as
$$
    \sigma_{\varepsilon, j}^{\text{CCH}}(x,v,\theta) := \left(\Phi^{\text{LT}}_{e_{j-1}} \circ \Phi^{f_2,\varepsilon} \circ \Phi^{\text{LT}}_{I_{j-1}}\right)\left([x_1,v_1,a_1,b_1,\nu_1,\rho_{X,1},\rho_{V,1},...,x_d,v_d,a_d,b_d,\nu_d,\rho_{X,d},\rho_{V,d}]^T\right).
$$

By same techniques as in the proof of Theorem \ref{MainTheorem}, we obtain that size$(\sigma_{\varepsilon, j}^{\text{CCH}}) \leq 2 \max(4\Tilde{C},18)\varepsilon^{-2(m+2)}$. For the DNNs $\sigma_{\varepsilon, j}^{\text{CCH}}$ for $j=d+2,...,2d+1$, we can choose $\Phi^{\text{LT}}_{A,b}$ with $A$ and $b$ only having entries equal to zero and hence size$(\sigma^{\text{CCH}}_{\varepsilon, j}) = 0$. Using \eqref{SumOfDNNs} we can define the DNN $\Bar{\mu}_{\varepsilon}^{\text{CCH}}:\mathcal{V}^d \times \Theta^{\text{CCH}}\rightarrow\mathbb{R}^d$ as 
\begin{equation*}
    \Bar{\mu}_{\varepsilon}^{\text{CCH}}(v,\theta) := \sum_{i=1}^d \left(\Phi^{\text{LT}}_{e_i} \circ \Phi^{f_3,\varepsilon} \circ \Phi^{\text{LT}}_{\tilde{I}_{i}}\right)\left([v_1,a_1,b_1,\nu_1,\rho_{X,1},\rho_{V,1},...,v_d,a_d,b_d,\nu_d,\rho_{X,d},\rho_{V,d}]^T\right)
\end{equation*}

also with size$(\Bar{\mu}_{\varepsilon}^{\text{CCH}})\leq 2d \max(4\Tilde{C},18)\varepsilon^{-2(m+2)}$. Equivalently, we define $\Bar{\sigma}_{\varepsilon,1}^{\text{CCH}}:\mathcal{V}^d \times \Theta^{\text{CCH}}\rightarrow \mathbb{R}^d$ as 
\begin{equation*}
    \bar{\sigma}^{\text{CCH}}_{\varepsilon,1}(v,\theta) := \sum_{i=1}^d \left(\Phi^{\text{LT}}_{e_i} \circ \Phi^{f_4,\varepsilon} \circ \Phi^{\text{LT}}_{\tilde{I}_{i}}\right)\left([v_1,a_1,b_1,\nu_1,\rho_{X,1},\rho_{V,1},...,v_d,a_d,b_d,\nu_d,\rho_{X,d},\rho_{V,d}]^T\right)
\end{equation*}

with size$(\bar{\sigma}^{\text{CCH}}_{\varepsilon,1})\leq 2d \max(4\Tilde{C},18)\varepsilon^{-2(m+2)}$. Finally, we construct $\bar{\sigma}^{\text{CCH}}_{\varepsilon,j}:\mathcal{V}^d \times \Theta^{\text{CCH}} \rightarrow \mathbb{R}^d$ for $j=2,...,d+1$ as
\begin{equation*}
    \bar{\sigma}^{\text{CCH}}_{\varepsilon, j}(v,\theta) := \left(\Phi^{\text{LT}}_{e_{j-1}} \circ \Phi^{f_5,\varepsilon} \circ \Phi^{\text{LT}}_{\tilde{I}_{j-1}}\right)\left([v_1,a_1,b_1,\nu_1,\rho_{X,1},\rho_{V,1},...,v_d,a_d,b_d,\nu_d,\rho_{X,d},\rho_{V,d}]^T\right)
\end{equation*}

and for $j=d+2,...,2d+1$ as 
\begin{equation*}
    \bar{\sigma}^{\text{CCH}}_{\varepsilon, j}(v,\theta) := \left(\Phi^{\text{LT}}_{e_{j-(d+1)}} \circ \Phi^{f_6,\varepsilon} \circ \Phi^{\text{LT}}_{\tilde{I}_{j-(d+1)}}\right)\left([v_1,a_1,b_1,\nu_1,\rho_{X,1},\rho_{V,1},...,v_d,a_d,b_d,\nu_d,\rho_{X,d},\rho_{V,d}]^T\right),
\end{equation*}

where it holds size$(\bar{\sigma}_{\varepsilon, j}^{\text{CCH}}) \leq 2 \max(4\Tilde{C},18)\varepsilon^{-2(m+2)}$ for all $j=2,...,2d+1$. Therefore, it follows
\begin{align*}
    \text{size}(\Bar{\mu}^{\text{CCH}}_\varepsilon) + \sum_{j=1}^{2d+1}\text{size}(\sigma_{\epsilon,j}^{\text{CCH}}) + \text{size}(\bar{\sigma}_{\varepsilon,j}^{\text{CCH}}) 
    \leq 12d\max(4\Tilde{C},18)\varepsilon^{-2(m+2)}.
\end{align*}

Moreover, we obtain for all $x\in\mathcal{X}^d$, $v\in\mathcal{V}^d$ and $\theta\in\Theta^{\text{CCH}}$
\begin{align}\label{BoundDiffCoeffCCH}
    &\Vert \sigma^{\theta, \text{CCH}}(x,v) -  \sigma^{\text{CCH}}_{\varepsilon}(x,v,\theta)\Vert^2_F + \Vert \bar{\mu}^{\theta,\text{CCH}}(v) - \bar{\mu}^{\text{CCH}}_{\varepsilon}(v,\theta) \Vert^2\nonumber\\[2mm]
    & \phantom{1cm}+ \Vert \bar{\sigma}^{\theta,\text{CCH}}(v) - \bar{\sigma}^{\text{CCH}}_{\varepsilon}(v,\theta) \Vert^2_F \nonumber\\[2mm]
    =& \sum_{k=1}^d \bigg(\vert \sigma^{\theta,\text{CCH}}_{1,k}(x,v) - \left(\sigma^{\text{CCH}}_{\varepsilon,1}(x,v,\theta)\right)_k \vert^2 + \vert \sigma^{\theta,\text{CCH}}_{k,k+1}(x,v) - \left(\sigma^{\text{CCH}}_{\varepsilon,k+1}(x,v,\theta)\right)_{k} \vert^2\bigg) \nonumber\\[2mm]
    & \phantom{1cm} + \sum_{k=1}^d \vert \bar{\mu}^{\theta,\text{CCH}}_{1,k}(v) - \left(\bar{\mu}^{\text{CCH}}_{\varepsilon,1}(v,\theta)\right)_k \vert^2 + \sum_{k=1}^d \bigg(\vert \bar{\sigma}^{\theta,\text{CCH}}_{1,k}(v) - \left(\bar{\sigma}^{\text{CCH}}_{\varepsilon,1}(v,\theta)\right)_k \vert^2 \nonumber\\[2mm]
    & \phantom{1cm} + \vert \bar{\sigma}^{\theta,\text{CCH}}_{k,k+1}(v) - \left(\bar{\sigma}^{\text{CCH}}_{\varepsilon,k+1}(v,\theta)\right)_{k} \vert^2 + \vert \bar{\sigma}^{\theta,\text{CCH}}_{k,d+k+1}(v) - \left(\bar{\sigma}^{\text{CCH}}_{\varepsilon,d+k+1}(v,\theta)\right)_{k} \vert^2 \bigg) \nonumber\\[2mm]
    & \leq \bar{d} \sum_{i=1}^6 \Vert f_i - \Phi^{f_i,\varepsilon} \Vert^2_\infty \leq 6 \bar{d} \varepsilon^2(1 + \Vert x \Vert^2 + \Vert v \Vert^2 + \Vert \theta \Vert^2).
\end{align}

Lastly, we get for all $x\in\mathcal{X}^d$, $v\in\mathcal{V}^d$ and $\theta \in \Theta^{\text{CCH}}$ using \eqref{MBoundCCH}, \eqref{BoundDiffCoeffCCH} and the triangular inequality that
\begin{align*}
    & \Vert \sigma_{\varepsilon}^{\text{CCH}}(x,v,\theta)\Vert_F + \Vert \bar{\mu}_{\varepsilon}^{\text{CCH}}(v,\theta)\Vert + \Vert \bar{\sigma}_{\varepsilon}^{\text{CCH}}(v,\theta)\Vert_F \\[2mm]
    \leq &  \Vert \sigma_{\varepsilon}^{\text{CCH}}(x,v,\theta) - \sigma^{\theta,\text{CCH}}(x,v)\Vert_F +  \Vert \sigma^{\theta,\text{CCH}}(x,v)\Vert_F + \Vert \bar{\mu}^{\text{CCH}}_{\varepsilon}(v,\theta) - \bar{\mu}^{\theta,\text{CCH}}(v)\Vert \\[2mm]
    &\phantom{1cm} + \Vert \bar{\mu}^{\theta,\text{CCH}}(v)\Vert + \Vert \bar{\sigma}_{\varepsilon}^{\text{CCH}}(v,\theta) - \bar{\sigma}^{\theta,\text{CCH}}(v)\Vert_F +  \Vert \bar{\sigma}^{\theta,\text{CCH}}(v)\Vert_F \\[2mm]
    \leq & 3\sqrt{6}d^{1/2}\varepsilon + \sqrt{2}d^{1/2}\sqrt{M} + d^{1/2}\sqrt{M} + \sqrt{3}d^{1/2}\sqrt{M}\\[2mm]
    \leq & 24 M (\bar{d} + \Vert x \Vert + \Vert v \Vert + \Vert \theta \Vert).
\end{align*}

So the result follows by the choice $C=\max(24M, 12\max(4\Tilde{C},18))$, $l=1/2$, $p=1$ and $q = 2(m+2)$ in \eqref{ApproxCCH01}, \eqref{ApproxCCH02} and \eqref{ApproxCCH03}, respectively.
\end{proof}


\subsection{DNN approximation of the pricing function}
The Lipschitz property of the coefficient functions and Lemma \ref{ApproxCCHestonCoeff} hold only if restricted to the space $\mathcal{X}^d\times\mathcal{V}^d\times \Theta^{\text{CCH}}$. However, in general the processes $X^{x,v}$ and $V^v$ may attain values outside of $\mathcal{X}^d$ and $\mathcal{V}^d$, respectively. Therefore, we introduce the stopping times 
$$
    \tau_{\mathcal{X}} := \inf\{t\in[0,T] \mid X_t^{x,v}\notin\mathcal{X}^d\} \text{ and } \tau_{\mathcal{V}} := \inf\{t\in[0,T] \mid V_t^{v}\notin\mathcal{V}^d\} 
$$

and for $\tau := \tau_{\mathcal{X}} \wedge \tau_{\mathcal{V}}$ we consider the stopped processes $X^{x,v,\tau} = (X^{x,v,\tau}_t)_{t\in[0,T]}$ and $V^{v,\tau} = (V^{v,\tau}_t)_{t\in[0,T]}$ given by

$$
    X^{x,v,\tau}_t = X^{x,v}_t \mathbbm{1}_{\{t\leq \tau\}} + X^{x,v}_{\tau} \mathbbm{1}_{\{t > \tau\}} \text{ and } V^{v,\tau}_t = V^v_t \mathbbm{1}_{\{t \leq \tau\}} + V^v_{\tau} \mathbbm{1}_{\{t > \tau\}},\; t\in[0,T].
$$

Note that $\tau \neq 0$ due to the fact that $\mathcal{X}^d$ and $\mathcal{V}^d$ are chosen such that they contain the initial conditions $x$ and $v$ of $X^{x,v}$ and $V^{v}$ respectively. Moreover, by choosing the sets $\mathcal{X}$ and $\mathcal{V}$ sufficiently large the difference between the option prices written on the stopped and original price process can be made arbitrarily small. The following lemma shows that we can approximate the price of a European option written on the risky asset from \eqref{EnlargedCCHeston} stopped at $\tau$. 

\begin{corollary}\label{MainTheoremCCH}
For $n\in\mathbb{N}$, let $\varphi:\mathbb{R}^d\times \mathbb{R}^{nd}\rightarrow\mathbb{R}$ be a payoff function satisfying the properties of Assumption \ref{AssumPricingApprox}. Then for $m=5$ there exist constants $\mathfrak{C},\mathfrak{p},\mathfrak{q}>0$ not depending on $d, \bar{\varepsilon}$ and for any target accuracy $\bar{\varepsilon}\in(0,1/2)$ a neural network $U^{\text{CCH}}_{\bar{\varepsilon}}:\mathcal{X}^d\times\mathcal{V}^d\times\Theta^{\text{CCH}}\times \mathbb{R}^{kd}$ such that
\begin{itemize}
    \item[(i)] \begin{equation}\label{Theorem4.1.Prop1CCH}
        \text{size}(U^{\text{CCH}}_{\bar{\varepsilon}})\leq \mathfrak{C} \bar{d}^\mathfrak{p} \bar{\varepsilon}^{-\mathfrak{q}},
    \end{equation}
    \item[(ii)] 
    \begin{equation}\label{Theorem4.1.Prop2CCH}
        \left( \int \left\vert \mathbb{E}[\varphi(X^{x,v,\tau}_T,K)] - U^{\text{CCH}}_{\bar{\varepsilon}}(x,v,\theta,K) \right\vert^2 \mu(\diff{x},\diff{v},\diff{\theta},\diff{K}) \right)^{1/2} < \bar{\varepsilon}.
    \end{equation}
\end{itemize}
\end{corollary}

\begin{proof}
Let $n\in\mathbb{N}$ and $\varphi:\mathbb{R}^d\times \mathbb{R}^{nd}\rightarrow\mathbb{R}$ be a payoff function satisfying the properties of Assumption \ref{AssumPricingApprox}. For $m=5$ we have by the Lipschitz property on compact domains and Lemma \ref{ApproxCCHestonCoeff} that the coefficient functions of the system \eqref{EnlargedCCHeston} satisfy the Assumption \ref{StrucAssumptionV} and \ref{AssumpCoeffApprox} with $\Bar{d}=2d+1$ on  $\mathcal{X}^d\times\mathcal{V}^d\times\Theta^{\text{CCH}}\times\mathbb{R}^{nd}$. Therefore, the result follows directly from Theorem \ref{MainTheorem} and the observation that $\mathbb{E}[\varphi(X^{x,v,\tau}_{\tau},K)] = \mathbb{E}[\varphi(X^{x,v,\tau}_T,K)]$.
\end{proof}

Note that Lemma \ref{MainTheoremCCH} only gives the existence of a DNN approximating the option price written on the stopped asset process. In the following proposition we show how to generalise this result.

\begin{proposition}
For $n\in\mathbb{N}$ let $\varphi:\mathbb{R}^d\times \mathbb{R}^{nd}\rightarrow\mathbb{R}$ be a payoff function satisfying the properties of Assumption \ref{AssumPricingApprox}. Then for any $K\in\mathbb{R}^{nd}$ it holds
$$
     \mathbb{E}[\varphi(X^{x,v,\tau}_T,K)]\longrightarrow\mathbb{E}[\varphi(X^{x,v}_T,K)] 
$$

for $\tau \rightarrow T$.
\end{proposition}

\begin{proof}
Let $n\in\mathbb{N}$ and $\varphi:\mathbb{R}^d\times \mathbb{R}^{nd}\rightarrow\mathbb{R}$ be a payoff function satisfying the properties of Assumption \ref{AssumPricingApprox}. Then for any $K\in\mathbb{R}^{nd}$ and every $\varepsilon\in(0,1/2)$ we obtain by the triangle inequality that
\begin{align}\label{TriInequalityConv}
    &\bigg\vert \mathbb{E}[\varphi(X^{x,v}_T,K)] - \mathbb{E}[\varphi(X^{x,v,\tau}_T,K)] \bigg\vert \nonumber \\[2mm]
    \leq \hspace{1mm} & \mathbb{E}\left[\vert \varphi(X^{x,v}_T,K) - \phi_{\varepsilon,d}(X^{x,v}_T,K) \vert\right] + \mathbb{E}\left[\vert \phi_{\varepsilon,d}(X^{x,v}_T,K) - \phi_{\varepsilon,d}(X^{x,v,\tau}_T,K) \vert\right]\nonumber \\[2mm]
    &\hspace{3mm} + \mathbb{E}\left[\vert \phi_{\varepsilon,d}(X^{x,v,\tau}_T,K) - \varphi(X^{x,v,\tau}_T,K) \vert\right] \nonumber \\[2mm]
    \leq \hspace{1mm} & C \bar{d}^p \varepsilon \left( 2 + \mathbb{E}[\Vert X^{x,v}_T\Vert] + \mathbb{E}[\Vert X^{x,v,\tau}_T\Vert] + 2 \Vert K \Vert \right) + C\bar{d}^p \varepsilon^{-l} \mathbb{E}\left[\Vert X^{x,v}_T - X^{x,v,\tau}_T \Vert\right],
\end{align}

where $\phi_{\varepsilon,d}$ is the DNN from Assumption \ref{AssumPricingApprox}. For every $i=1,\dots,d$ the pair $(X^i,V^i)$ from \eqref{CCHeston} are a special case of the SDE considered in \cite{Andersen2006} with $\lambda = 1$, \textkappa $\;= a_i$, $\theta = b_i$, $\varepsilon = \nu_i$, $p=1/2$ and $\rho = \rho_{V,i}$. Hence by Proposition 2.5 in \cite{Andersen2006} the processes $X^i$, $i=1,\dots,d$, are true $\mathbb{Q}$-martingales. So by Doob's martingale inequality and Assumption \ref{AssumBoundedParaCCH} it follows that
\begin{equation}\label{SecondMomemntFiniteCCH}
    \mathbb{E}\left[\sup_{t\in[0,T]} \vert X^i_t\vert^2\right] \leq 4\mathbb{E}[\vert X^i_T\vert^2] < \infty
\end{equation}
    
and so by dominated convergence 
\begin{equation}\label{L2ConvTerminalValues}
    \lim_{\tau \rightarrow T} \mathbb{E}\left[\vert X^{i}_T - X^{i,\tau}_T \vert^2\right] = \lim_{\tau \rightarrow T} \mathbb{E}\left[\vert X^{i}_T - X^{i}_\tau\vert^2\right] = 0
\end{equation}

for all $i=1,\dots,d$. Finally, by applying \eqref{SecondMomemntFiniteCCH} to \eqref{TriInequalityConv} we obtain by Jensen's inequality
\begin{align*}
    &\bigg\vert \mathbb{E}[\varphi(X^{x,v}_T,K)] - \mathbb{E}[\varphi(X^{x,v,\tau}_T,K)] \bigg\vert \\[2mm]
    \leq \hspace{1mm} & 2C \bar{d}^{p} \varepsilon \left( 1 + \mathbb{E}\left[\sup_{t\in[0,T]} \Vert X^{x,v}_t \Vert^2\right]^{1/2} + \Vert K \Vert \right) + C\bar{d}^p \varepsilon^{-l} \mathbb{E}\left[\Vert X^{x,v}_T - X^{x,v,\tau}_T \Vert\right]\\[2mm]
    \leq \hspace{1mm} & 2C \bar{d}^{p} \varepsilon \left( 1 + \left(\sum_{i=1}^d \mathbb{E}\left[\sup_{t\in[0,T]} \vert X^{i}_t\vert^2\right]\right)^{1/2} + \Vert K \Vert \right) + C\bar{d}^p \varepsilon^{-l} \left(\sum_{i=1}^d \mathbb{E}\left[\vert X^{i}_T - X^{i,\tau}_T \vert^2\right]\right)^{1/2}.
\end{align*}

The result now follows from \eqref{L2ConvTerminalValues} by letting $\tau \rightarrow T$ and then by choosing $\varepsilon$ arbitrarily small.
\end{proof}



\section{Deep Calibration of the rough Bergomi Model}\label{Section5}
We now study deep calibration using approximations of option prices in the setting of another class of stochastic volatility models given by the rough Bergomi (rBergomi) model with a flat forward variance curve as introduced in \cite{Bayer2015}. To this purpose we extend the framework of Section \ref{Section4} to include Volterra processes as driver for the stochastic volatility. More specifically, consider two independent $1$-dimensional standard $\mathbb{F}$-Brownian motions $B,B^\perp$ on the filtered probability space $(\Omega,\mathcal{F},\mathbb{F}=(\mathcal{F}_t)_{t\in[0,T]},\mathbb{Q})$ from Section \ref{Section3}. Let $W=(W_t)_{t\in[0,T]}$ be defined as $W:= \rho B + \sqrt{1-\rho^2}B^\perp$ for a negative correlation parameter $\rho\in[-1,0]$. We assume that the risky asset process $X^{x,v}$ and variance process $V^v$ are now described by the following dynamics 
\begin{equation}\label{rBergomiModel}
    \begin{cases}
        \diff{X}^{x,v}_t = \sqrt{V^v_t} X^{x,v}_t \diff{W}_t,\, X_0 = x >0,\\[3mm]
        V_t^v = \nu\cdot \exp\left( \eta \sqrt{2H} \mathfrak{X}_t - \frac{1}{2} \eta^2 t^{2H} \right), \\[3mm]
        \diff\langle W,B \rangle_t = \rho \diff{t},
    \end{cases}
\end{equation}

for some constants $\eta,\nu > 0$ and where the process $\mathfrak{X} = (\mathfrak{X}_t)_{t\in[0,T]}$ defined as 
\begin{equation}\label{RiemannLouiville}
    \mathfrak{X}_t = \int_0^t (t-s)^{H-1/2}\diff{B}_s
\end{equation}

is a Volterra process and $H\in(0,1/2)$ is the so called Hurst parameter. As indicated in \cite{Bayer2015}, the process $\mathfrak{X}$ behaves very similarly to a fractional Brownian motion (fBM). However, it exhibits a slightly different dependence structure than a fBM especially for small values of $H$.\par

In this case the corresponding model parameters to calibrate are given by $\theta = (\nu,\eta,\rho,H)\in\Theta^{\text{rBergomi}}\subseteq (0,\infty)^2 \times [-1,0] \times (0,1/2)$, where $\Theta^{\text{rBergomi}}$ denotes the parameter space for the rBergomi model that we now specify.

As in Section \ref{CCHeston}, we assume that the model parameters are lying in some compact intervals for the remainder of this section. 

\begin{assumption}
    We assume that there exist $0<\underline{H} \leq \overline{H}< 1/2$, $0<\underline{\eta}\leq \overline{\eta}<\infty$ and $0<\underline{\nu}\leq\overline{\nu}<\infty$ such that $H\in[\underline{H},\overline{H}]$, $\eta\in[\underline{\eta},\overline{\eta}]$ and $\nu\in[\underline{\nu},\overline{\nu}]$, and that $\Theta^{\text{rBergomi}} := [\underline{\nu}, \overline{\nu}]\times[\underline{\eta}, \overline{\eta}]\times[-1,0]\times[\underline{H}, \overline{H}]$.
\end{assumption}

In order to apply Theorem \ref{DNNApproxC1} and similar results on approximation properties of DNNs later on, we first consider a modified version of the model \eqref{rBergomiModel} given by the truncated modification
\begin{equation}\label{TruncatedrBergomiModel}
    \begin{cases}
        \diff{X_t}^{D} = \sqrt{V^{D}_t} X^{D}_t \diff{W}_t,\, X_0 = x >0,\\[3mm]
        \sqrt{V_t^{D}} = \sqrt{\nu}\cdot \min\left(\max\left(\frac{1}{D},\exp\left(\eta \sqrt{\frac{H}{2}}{\mathfrak{X}}_{t} - \frac{1}{4}\eta^2 t^{2H} \right)\right), D\right), \\[3mm]
        \diff\langle W,B \rangle_t = \rho \diff{t},
    \end{cases}
\end{equation}

for a fixed $D\geq 1$. If $D$ is very large, we expect that option prices in the truncated model are very close to option prices in the original model \eqref{rBergomiModel}. We make this intuition precise in Proposition \ref{PropConvrBergomiPrice} below. We can establish the following upper bound for the second and fourth moments of the asset process $X^D$.
\begin{lemma}\label{MomentBoundrBergomiX}
    There exist constants $c_{20},c_{21}>0$ such that for every $t\in[0,T]$, $x>0$ and $(\nu,\eta,\rho,H)\in\Theta^{\text{rBergomi}}$ it holds 
    $$
        \mathbb{E}\left[\vert X^D_t \vert^2\right] \leq c_{20}\vert x \vert^2 \text{ and } \mathbb{E}\left[\vert X^D_t \vert^4\right] \leq c_{21}\vert x \vert^4.
    $$
\end{lemma}

\begin{proof}
    This follows directly from \eqref{TruncatedrBergomiModel} and Itô's formula.
\end{proof}


\subsection{Discrete-time approximation of the asset process}
As in \eqref{DiscreteTimePoints}, we consider the partition of the interval $[0,T]$ given by the discrete time points
$$
    0 = t_0 <... < t_k < ... < t_N = T
$$

with $t_k = kh$, $k=0,\dots,N$, for the grid size $h= T/N$. Without limitation of generality we set $h\in(0,1)$. Moreover, we obtain the following Euler-Maruyama Scheme for the process $X^D$
\begin{align}\label{EulerSchemeRBergomiX}
    \hat{X}^{D}_{t_{k+1}} = \hat{X}^{D}_{t_{k}} + \hat{X}^{D}_{t_{k}} \sqrt{V^D_{t_{k}}} \rho (B_{t_{k+1}} - B_{t_{k}}) + \hat{X}^{D}_{t_{k}} \sqrt{V^D_{t_{k}}} \sqrt{1-\rho^2} (B^\perp_{t_{k+1}} - B^\perp_{t_{k}}),\; \hat{X}^D_0 = x.
\end{align}

Further, we introduce the continuous-time interpolation $\bar{X}^D = (\bar{X}^D_t)_{t\in[0,T]}$ following the dynamics
\begin{equation}\label{TruncatedrBergomiCont}
        \diff \bar{X}^D_t = \bar{X}^D_{\floor{t}}\sqrt{V^D_{\floor{t}}}\rho \diff{B}_t + \bar{X}^D_{\floor{t}}\sqrt{V^D_{\floor{t}}}\sqrt{1-\rho^2} \diff{B^\perp_t},\; \bar{X}^D_0 = x.
\end{equation}

In particular, we can derive the following upper bound for the second and fourth moments of $\bar{X}^D$.
\begin{lemma}\label{MomentBoundrBergomiXBar}
    There exist constants $c_{22},c_{23}>0$ such that for every $x>0$ and $(\nu,\eta,\rho,H)\in\Theta^{\text{rBergomi}}$ it holds
    $$
        \mathbb{E}\left[\sup_{t\in[0,T]} \left\vert \bar{X}_t^D \right\vert^2\right] \leq c_{22}\vert x \vert^2 \text{ and }\mathbb{E}\left[\sup_{t\in[0,T]} \left\vert \bar{X}_t^D \right\vert^4\right] \leq c_{23}\vert x \vert^4.
    $$
\end{lemma}

\begin{proof}
    This follows directly by the Burkholder-Davis-Gundy inequality and Gronwall's inequality.
\end{proof}

Moreover, we have the following strong approximation result of the interpolating process $\bar{X}^D$.
\begin{lemma}\label{StrongApproxXDbar}
    There exists a constant $c_{24}>0$ such that for every $x>0$ and $(\nu,\eta,\rho,H)\in\Theta^{\text{rBergomi}}$ it holds 
    $$
        \mathbb{E}\left[\sup_{t\in[0,T]} \left\vert X^D_t - \bar{X}_t^D \right\vert^2\right] \leq c_{24}h^{\underline{H}}(1+\vert x\vert^2).
    $$
\end{lemma}

\begin{proof}
    Let $x>0$ and $(\nu,\eta,\rho,H)\in\Theta^{\text{rBergomi}}$. For every $t\in[0,T]$ we define $G(t):= \mathbb{E}\left[\sup_{s\in[0,t]} \left\vert X^D_s - \bar{X}_s^D \right\vert^2\right]$. Then by the Burkhölder-Davis-Gundy inequality together with Itô's isometry there exists a constant $C>0$ such that it holds
    $$
        G(t) \leq C \int_0^t \mathbb{E}\left[\left\vert X^D_r \sqrt{V^D_r} - \bar{X}^D_{\floor{r}} \sqrt{V^D_{\floor{r}}}\right\vert^2\right]\diff{r}.
    $$

    By \eqref{SquaredTriangular} we have that
    $$
        \mathbb{E}\left[\left\vert X^D_r \sqrt{V^D_r} - \bar{X}^D_{\floor{r}} \sqrt{V^D_{\floor{r}}}\right\vert^2\right]\leq 2 \mathbb{E}\left[\left\vert X^D_r \sqrt{V^D_r} -  X^D_{\floor{r}} \sqrt{V^D_{\floor{r}}}\right\vert^2\right] + 2 \mathbb{E}\left[\left\vert X^D_{\floor{r}} \sqrt{V^D_r} - \bar{X}^D_{\floor{r}} \sqrt{V^D_{\floor{r}}}\right\vert^2\right].
    $$

    For the first term it follows directly with Lemma \ref{MomentBoundrBergomiX} that
    \begin{equation}\label{bound:1}
        \mathbb{E}\left[\left\vert X^D_r \sqrt{V^D_r} -  X^D_{\floor{r}} \sqrt{V^D_{\floor{r}}}\right\vert^2\right] \leq \bar{\nu}^2 D^4 \int_{\floor{r}}^r \mathbb{E}\left[\vert X^D_s \vert^2\right]\diff{s} \leq \bar{\nu}^2 D^4 h \vert x \vert^2 c_{20}.
    \end{equation}

    Again by applying \eqref{SquaredTriangular} we obtain 
    \begin{equation}\label{bound:2}
        \mathbb{E}\left[\left\vert X^D_{\floor{r}} \sqrt{V^D_r} - \bar{X}^D_{\floor{r}} \sqrt{V^D_{\floor{r}}}\right\vert^2\right] \leq 2 \mathbb{E}\left[\left\vert X^D_{\floor{r}} \sqrt{V^D_r} - X^D_{\floor{r}} \sqrt{V^D_{\floor{r}}}\right\vert^2\right] + 2 \mathbb{E}\left[\left\vert X^D_{\floor{r}} \sqrt{V^D_{\floor{r}}} - \bar{X}^D_{\floor{r}} \sqrt{V^D_{\floor{r}}}\right\vert^2\right].
    \end{equation}

    We see that $\mathbb{E}\left[\left\vert X^D_{\floor{r}} \sqrt{V^D_{\floor{r}}} - \bar{X}^D_{\floor{r}} \sqrt{V^D_{\floor{r}}}\right\vert^2\right] \leq \bar{\nu}D^2 G(r)$ and by the Cauchy-Schwarz inequality it follows 
    \begin{equation}\label{bound:3}
        \mathbb{E}\left[\left\vert X^D_{\floor{r}} \sqrt{V^D_r} - X^D_{\floor{r}} \sqrt{V^D_{\floor{r}}}\right\vert^2\right] \leq \mathbb{E}\left[\left\vert X_{\floor{r}}^D \right\vert^4\right]^{1/2} \mathbb{E}\left[\left\vert\sqrt{V^D_r} - \sqrt{V^D_{\floor{r}}}\right\vert^4\right]^{1/2}.
    \end{equation}

    Moreover, since the truncated exponential function in the definition of $V^D$ is Lipschitz, there exists a constant $L>0$ such that 
    \begin{equation}\label{bound:4}
        \mathbb{E}\left[\left\vert\sqrt{V^D_r} - \sqrt{V^D_{\floor{r}}}\right\vert^4\right]\leq 2L^4 \bar{\nu}^2 \bar{H}^2 \mathbb{E}[\vert \mathfrak{X}_r - \mathfrak{X}_{\floor{r}} \vert^4] + \frac{1}{64}L^4 \bar{\nu}^2 \bar{\eta}^8 \vert r^{2H} - \floor{r}^{2H} \vert^4.
    \end{equation}

    Using the fact that $\mathfrak{X}_r - \mathfrak{X}_{\floor{r}} \overset{d}{=} \mathfrak{X}_{r-\floor{r}}$ we obtain by Gaussianity of the process $\mathfrak{X}$ for the fourth moment
    \begin{equation}\label{bound:5}
        \mathbb{E}\left[\vert \mathfrak{X}_r - \mathfrak{X}_{\floor{r}}\vert^4\right] \leq \frac{3}{4\underline{H}^2}h^{2\overline{H}}.
    \end{equation}

    Lastly, note that for $\floor{r}<h$ it directly follows $\vert r^{2H} - \floor{r}^{2H}\vert^4 \leq (2h)^{8\underline{H}}$. For the case that $\floor{r}\geq h$ by the mean-value theorem we obtain that $\vert r^{2H} - \floor{r}^{2H}\vert^4\leq 16 \overline{H}^4 h^{8\underline{H}}$. Therefore, for general $r\in[0,t]$ we have that
    \begin{equation}\label{bound:6}
        \vert r^{2H} - \floor{r}^{2H}\vert^4 \leq h^{8\underline{H}}\max\{2^{8\underline{H}}, 16\overline{H}^4\}.
    \end{equation}

    In conclusion, by using \eqref{bound:1}, \eqref{bound:2}, \eqref{bound:3}, \eqref{bound:4}, \eqref{bound:5}, \eqref{bound:6} and the fact that $\sqrt{a+b}\leq \sqrt{a} + \sqrt{b}$ for all $a,b>0$ it follows for every $t\in[0,T]$
    \begin{align*}
        G(t) \leq & 2C \bar{\nu}^2 D^4 h \vert x\vert^2c_{20}T + 4C \int_0^t \mathbb{E}\left[\left\vert X^D_{\floor{r}}\sqrt{V^D_r} - X^D_{\floor{r}}\sqrt{V^D_{\floor{r}}}\right\vert^2\right]\diff{r} \\[2mm] &\phantom{...} + 4C \int_0^t \mathbb{E}\left[\left\vert X^D_{\floor{r}}\sqrt{V^D_{\floor{r}}} - \bar{X}^D_{\floor{r}}\sqrt{V^D_{\floor{r}}}\right\vert^2\right]\diff{r} \\[2mm]
        \leq & 4C\bar{\nu}^2 D^4 h c_{20} \vert x\vert^2 T + \sqrt{c_{21}}\vert x\vert^2\left(\sqrt{96}L^2\bar{\nu}h^{\underline{H}} + L^2\bar{\nu}\bar{\mu}^4 h^{4\underline{H}}\max\{2^{2\underline{H}},4\bar{H}^2\}\right)T + 8C\bar{\nu}D^2\int_0^t G(r)\diff{r}.
    \end{align*}

    So the result finally follows by observing that $G(t)\in L^1([0,T])$ by \eqref{SquaredTriangular} and then applying Gronwall's inequality. 
\end{proof}


\subsection{DNN approximation of the Volterra process}
In \eqref{TruncatedrBergomiCont} we observe that it is crucial to simulate the process $\sqrt{V^v}$ on the discrete time points $t_1,...,t_N$. In particular, we need to generate 
\begin{align}\label{DiscreteTimeRBergomiV}
    \sqrt{{V}^{D}_{t_{k}}} &= \sqrt{\nu}\cdot \min\left(\max\left(\frac{1}{D},\exp\left(\eta \sqrt{\frac{H}{2}}{\mathfrak{X}}_{t_k} - \frac{1}{4}\eta^2 {t_k}^{2H} \right)\right), D\right),\; k=1,\dots,N,
\end{align}

which depend on the values of $\mathfrak{X}$. Therefore, we first consider an approximation for the Volterra process in \eqref{RiemannLouiville} for a fixed time point $t\in[0,T]$. Note that $\mathfrak{X}$ belongs to the class of truncated Brownian semistationary processes considered in \cite{Bennedsen2017}. To simulate this kind of processes for a fixed time point $t\in[0,T]$, we now follow a hybrid scheme proposed in \cite{Bennedsen2017} to take care of the steepness of the kernel function near zero, where pure Riemann sums produce bad results. In particular, for some fixed $t$ consider the discrete-time grid defined as 
\begin{equation}\label{GridHybridScheme} 
    G(t) := \{t^t_k := k h(t) \mid k=0,...,N\}
\end{equation}

with $h(t) = t / N$. In particular, note that it holds $h(T) = T/N = h$. Then for any $\mathfrak{K}\in\{0,1,...,k\}$ we introduce
$$
    \bar{\mathfrak{X}}^{\mathfrak{K}}_{t} = \sum_{k=N-\mathfrak{K}+1}^{N} \int_{t^t_{k-1}}^{t^t_k} (t-s)^{H-1/2}\diff{B}_s + \sum_{k=1}^{N-\mathfrak{K}} (t-t^t_k)^{H-1/2}(B_{t^t_{k}} - B_{t^t_{k-1}}).
$$

Note that this scheme is realised by a linear combination of Wiener integrals and Riemann sums. In the special case of $\mathfrak{K}=0$, it corresponds to an approximation purely driven by Riemann sums. In \cite{Bennedsen2017} it is shown that choosing $\mathfrak{K}$ small already improves the approximation error significantly. Therefore, we focus on the special case $\mathfrak{K}=1$. In particular, we consider
\begin{equation}\label{DiscreteTimeVolterra}
    \bar{\mathfrak{X}}_{t} = \int_{t^t_{N-1}}^{t}(t - s)^{H-1/2} \diff{B}_s + \sum_{k=1}^{N-1} (t - t^t_k)^{H-1/2}(B_{t^t_{k}} - B_{t^t_{k-1}}).
\end{equation}

In this way we obtain a strong approximation result for the process $\mathfrak{X}$ in the sense of the following Lemma. 

\begin{lemma}\label{ApproximationHybirdScheme}
Let $\mathfrak{X} = (\mathfrak{X}_t)_{t\in[0,T]}$ be as in \eqref{RiemannLouiville}. Then there exists a constant $c_{25}>0$ such that for every $t\in[0,T]$ and $H\in(0,1/2)$
$$
    \mathbb{E}\left[\vert \mathfrak{X}_{t} - \bar{\mathfrak{X}}_{t} \vert^4\right] \leq c_{25} h^{4H}.
$$
\end{lemma}

\begin{proof}
For a fixed timepoint $t\in[0,T]$ and the corresponding grid $\eqref{GridHybridScheme}$ we can write 
$$
    \mathfrak{X}_t = \int_{t^t_{N-1}}^t (t-s)^{H-1/2}\diff{B}_s + \int_0^{t^t_{N-1}}(t-s)^{H-1/2}\diff{B}_s. 
$$

Therefore, it directly follows
\begin{align*}
    \mathbb{E}\left[\vert \mathfrak{X}_{t} - \bar{\mathfrak{X}}_{t} \vert^4\right] &= \mathbb{E}\left[\left\vert \int_0^{t^t_{N-1}} (t-s)^{H-1/2} \diff{B}_s - \sum_{k=1}^{N-1} (t-t^t_k)^{H-1/2}(B_{t^t_k} - B_{t^t_{k-1}})\right\vert^4\right] \\[2mm]
    & = \mathbb{E}\left[\left\vert \sum_{k=1}^{N-1} \int_{t^t_{k-1}}^{t^t_k} \left[ (t-s)^{H-1/2} - (t-t^t_k)^{H-1/2} \right] \diff{B}_s \right\vert^4\right].   
\end{align*}

Using the fact that for any $Y\sim\mathcal{N}(0,\sigma^2)$ it holds $\mathbb{E}[\vert Y\vert^4] = 3\sigma^4$, we obtain by Itô's isometry

\begin{equation}\label{FourthMomErrorTemp1}
    \mathbb{E}\left[\vert \mathfrak{X}_{t} - \bar{\mathfrak{X}}_{t} \vert^4\right] = 3\left(\sum_{k=1}^{N-1} \int_{t^t_{k-1}}^{t^t_k} \left[(t-s)^{H-1/2} - (t - t^t_k)^{H-1/2}\right]^2\diff{s}\right)^2.
\end{equation}

Further, by the mean-value theorem we obtain that 
\begin{equation}\label{KernelInequality}
    \vert x^\alpha - y^\alpha \vert \leq \vert \alpha \vert \min(x,y)^{\alpha-1}\vert x - y\vert,
\end{equation}

for every $x,y>0$ and $\alpha<1$. Note that $H-1/2\in(-1/2,0)$ and for every $s\in(t^t_{k-1},t^t_k)$ it holds $0 < t - t^t_k < t - s$. Hence, applying \eqref{KernelInequality} to \eqref{FourthMomErrorTemp1} yields
\begin{align*}
    \mathbb{E}\left[\vert \mathfrak{X}_{t} - \bar{\mathfrak{X}}_{t} \vert^4\right] \leq 3\left(\sum_{k=1}^{N-1} \int_{t^t_{k-1}}^{t^t_k} (H-1/2)^2 (t-t^t_k)^{2H-3}h(t)^2 \diff{s}\right)^2 \leq 3 h(t)^6 \left(H-1/2\right)^4\sum_{k=1}^{N-1} (t-t^t_k)^{4H-6}.
\end{align*}

Since $4H-6 \in (-6,-4)$ and $t-t^t_k = (N-k)h(t)$, we obtain
$$
    \sum_{k=1}^{N-1} (t-t^t_k)^{4H-6} =  h(t)^{4H-6} \sum_{k=1}^{N-1} k^{4H-6} \leq h(t)^{4H-6} \sum_{k=1}^{N-1} \frac{1}{k^4} \leq h(t)^{4H-6}\zeta(4),
$$

where $\zeta$ denotes the Riemann zeta function. Therefore, finally we obtain for the approximation error 
$$
    \mathbb{E}\left[\vert \mathfrak{X}_{t} - \bar{\mathfrak{X}}_{t} \vert^4\right] \leq \frac{3}{16}\zeta(4)h(t)^{4H}
$$

using that $(H-1/2)^4 \leq 1/16$ for any $H\in(0,1/2)$. This concludes the proof for $c_{25} := \frac{3}{16}\zeta(4)$ and with the observation that $h(t)\leq h(T) = h$.
\end{proof}
We can now define a second approximating process $\tilde{\mathfrak{X}}^{\varepsilon} = (\tilde{\mathfrak{X}}^\varepsilon_t)_{t\in[0,T]}$ which turns out to be a DNN for a fixed $t\in[0,T]$ and $\omega\in\Omega$. For every $t$ with corresponding grid $G(t)$ we observe
$$
    \int_{t^t_{N-1}}^{t}(t - s)^{H-1/2} \diff{B}_s \overset{d}{=} \frac{1}{\sqrt{2H}}h(t)^{H} Z_t
$$

for some $Z_t\sim\mathcal{N}(0,1)$ independent of $B_{u} - B_{v}$ for every $u>v\geq t$. By Theorem \ref{DNNApproxC1} for a fixed $t\in[0,T]$ and any $\varepsilon\in(0,1/2)$ we find a DNN $\Phi^{t,1}:[\underline{H},\overline{H}]\rightarrow\mathbb{R}$ so that for all $H\in[\underline{H},\overline{H}]$ it holds 
\begin{equation}\label{BoundDNN1rBergomi}
    \vert h(t)^{H} / \sqrt{2H} - \Phi^{t,1}(H) \vert \leq \varepsilon. 
\end{equation}

Moreover, for any $k=1,...,N-1$ there exists a DNN $\Phi^{t,k,2}:[\underline{H},\overline{H}] \rightarrow \mathbb{R}$ such that for all $H\in[\underline{H},\overline{H}]$ it holds 
\begin{equation}\label{BoundDNN2rBergomi}
    \vert (t - t^t_k)^{H-1/2} - \Phi^{t,k,2}(H) \vert \leq \varepsilon.
\end{equation}

Based on this, we define the process $\tilde{\mathfrak{X}}^\varepsilon=(\tilde{\mathfrak{X}}^\varepsilon_t)_{t\in[0,T]}$ by 
\begin{equation}\label{VolterraDNN}
    \tilde{\mathfrak{X}}^{\varepsilon}_{t} := \Phi^{t,1}(H) Z_t + \sum_{k=1}^{N-1} \Phi^{t,k,2}(H)(B_{t^t_k} - B_{t^t_{k-1}}).
\end{equation}

For a fixed $\omega\in\Omega$ this again turns out to be a weighted sum of DNNs and is hence a DNN itself. In particular, there exists a constant $c_{26}>0$ such that for every $k=1,...,N$, it holds 
\begin{equation}\label{SizeDNNVolterra}
    \text{size}(\tilde{\mathfrak{X}}^\varepsilon_{t_k}(\omega))\leq N c_{26}\varepsilon^{-2}.
\end{equation}

\begin{lemma}\label{ApproxDNNVolterra}
There exist constants $c_{27},c_{28}>0$ such that for every $t\in[0,T]$, $\varepsilon\in(0,1/2)$ and $(\nu,\eta,H,\rho)\in\Theta^{\text{rBergomi}}$ it holds
$$
    \mathbb{E}\left[\vert \mathfrak{X}_{t} - \tilde{\mathfrak{X}}^{\varepsilon}_{t} \vert^4\right] \leq c_{27} \frac{h^{4H}}{H^2} + c_{28}\varepsilon^4.
$$
\end{lemma}

\begin{proof}
Let $t\in[0,T],\varepsilon, H \in(0,1/2)$ and $(\nu,\eta,H,\rho)\in\Theta^{\text{rBergomi}}$. Then by \eqref{SquaredTriangular} we obtain 
\begin{equation}
    \mathbb{E}\left[\vert \mathfrak{X}_{t} - \tilde{\mathfrak{X}}^{\varepsilon}_{t} \vert^4\right] \leq 8 \mathbb{E}\left[\vert \mathfrak{X}_{t} - \bar{\mathfrak{X}}_{t} \vert^4\right] + 8 \mathbb{E}\left[\vert \bar{\mathfrak{X}}_{t} - \tilde{\mathfrak{X}}^{\varepsilon}_{t} \vert^4\right].
\end{equation}

For the second summand we obtain again with \eqref{SquaredTriangular} that 
\begin{align*}
    \mathbb{E}\left[\vert \bar{\mathfrak{X}}_{t} - \tilde{\mathfrak{X}}^{\varepsilon}_{t} \vert^4\right] 
    &\leq 8\mathbb{E}\left[\left\vert \int_{t^t_{N-1}}^{t}(t - s)^{H-1/2}\diff{B}_s  - \Phi^{t,1}(H)Z_t \right\vert^4\right] \\[2mm]
    & \phantom{..} + 8\mathbb{E}\left[\left\vert \sum_{k=1}^{N-1} \int_{t^t_{k-1}}^{t^t_k} \left[(t - t^t_k)^{H-1/2} - \Phi^{t,k,2}(H)\right] \diff{B}_s\right\vert^4\right].
\end{align*}

Then Itô's isometry and \eqref{BoundDNN2rBergomi} yield 
$$
    \mathbb{E}\left[\left\vert \sum_{k=1}^{N-1} \int_{t^t_{k-1}}^{t^t_k} \left[(t - t^t_k)^{H-1/2} - \Phi^{t,k,2}(H)\right] \diff{B}_s\right\vert^4\right] = 3\left(\sum_{k=1}^{N-1} h(t) \left\vert(t - t^t_k)^{H-1/2} - \Phi^{t,k,2}(H)\right\vert^2 \right)^2 \leq 3 t^2 \varepsilon^4.
$$

Moreover, using \eqref{SquaredTriangular}, the independence between $B$ and $Z_t$ as well as the bound \eqref{BoundDNN1rBergomi} we have that 
\begin{align*}
    &\mathbb{E}\left[\left\vert \int_{t^t_{N-1}}^{t}(t - s)^{H-1/2}\diff{B}_s  - \Phi^{t,1}(H)Z_t \right\vert^4\right] \\[4mm]
    \leq&\phantom{.} 8 \mathbb{E}\left[\left\vert \int_{t^t_{N-1}}^{t}(t - s)^{H-1/2}\diff{B}_s  - \frac{1}{\sqrt{2H}}h(t)^{H}Z_t \right\vert^4\right] + 8\mathbb{E}\left[\left\vert \frac{1}{\sqrt{2H}}h(t)^{H}Z_t  - \Phi^{t,1}(H)Z_t \right\vert^4\right] \\[4mm]
    \leq& 24\left( \frac{1}{4H^2} h(t)^{4H} + \varepsilon^4 \right).
\end{align*}

Using again the fact that $h(t)\leq h(T) = h$ finishes the proof with $c_{27}:= 8c_{25} + 384$ and $c_{28}:= 192 T^2 + 1536$.
\end{proof}


\subsection{DNN approximation of the system}
Now that we have introduced an approximation for the Volterra process $\mathfrak{X}$, we construct a DNN approximation for both the variance process $V^D$ and the asset process $X^D$ in \eqref{TruncatedrBergomiModel}. In the following Lemma it is shown that the $\min$ and $\max$ function in the definition of $V^D$ can be in fact explicitly written as DNNs.

\begin{lemma}\label{MinMaxDNNs}
    For a fixed $N\in\mathbb{R}$ there are DNNs $\Phi^{\text{max}}_N:\mathbb{R}\rightarrow\mathbb{R}$ and $\Phi^{\text{min}}_N:\mathbb{R}\rightarrow\mathbb{R}$ such that 
    \begin{equation*}
        \max(x,N) = \Phi^{\text{max}}_N(x) \; \text{ and } \min(x,N) = \Phi^{\text{min}}_N(x)
    \end{equation*}
    for every $x\in\mathbb{R}$. In particular,  size$(\Phi^{\text{max}}_N)=$size$(\Phi^{\text{min}}_N)\leq 12$, where equality holds if $N\neq 0$.
\end{lemma}

\begin{proof}
    For a fixed $N\in\mathbb{R}$ we can write
    $$
        \min(x,N) = \frac{1}{2}\left((x+N)^+ - (-x-N)^+ - (x-N)^+ - (N-x)^+\right).
    $$

    Therefore, using the notation from Section \ref{Section2} and \eqref{ReLU}, for every $x\in\mathbb{R}$ we can write $\min(x,N) = \Phi^{\text{min}}_N(x)$, where $\Phi^{\text{min}}_N$ is a DNN with one hidden layer and corresponding weight matrices and bias vectors
    $$
        A^1 =\left[\begin{array}{c}
             1 \\ -1 \\ 1 \\ -1
             \end{array}\right],\; b^1 = \left[\begin{array}{c}
             N \\ -N \\ -N \\ N
        \end{array}\right],\; A^2 = \frac{1}{2}\left[\begin{array}{c}
            1, -1, -1 , -1
        \end{array}\right].
    $$

    Using the fact that $\max(x,N) = -\min(-x,-N)$ we see that $\max(x,N) = \Phi^{\text{max}}_N(x)$, where $\Phi^{\text{max}}_N$ is a DNN with one hidden layer and corresponding weight matrices and bias vectors
    $$
        A^1 =\left[\begin{array}{c}
             -1 \\ 1 \\ -1 \\ 1
        \end{array}\right],\; b^1 = \left[\begin{array}{c}
             -N \\ N \\ N \\ -N
        \end{array}\right],\; A^2 = -\frac{1}{2}\left[\begin{array}{c}
            1, -1, -1 , -1
        \end{array}\right].
    $$
    Clearly, it holds that size$(\Phi^{\text{max}}_N)$ = size$(\Phi^{\text{min}}_N)\leq 12$. In particular, if $N\neq 0$, then equality holds in the last statement.
\end{proof}

For the remainder of this section we consider $\varepsilon\in(0,1/2)$. To construct a DNN approximation for the process $V^D$ it remains to find an appropriate approximation for the square root function on $[\underline{\nu},\overline{\nu}]$ and the exponential function. By Theorem \ref{DNNApproxC1} and the appropriate rescaling of the domain there exists a DNN $\Phi^3_\varepsilon:[\underline{\nu},\overline{\nu}]\rightarrow \mathbb{R}^+$ such that for every $\nu\in[\underline{\nu},\overline{\nu}]$ it holds
\begin{equation}\label{ApproxDNNSqrt}
    \vert \Phi^3_\varepsilon(\nu) - \sqrt{\nu} \vert \leq \varepsilon \text{ and size}(\Phi^3_\varepsilon)\leq c_{29}\varepsilon^{-2}
\end{equation}

for a constant $c_{29}>0$. Moreover, for any $H\in[\underline{H},\overline{H}]$, $\eta\in[\underline{\eta},\overline{\eta}]$ and $\nu\in[\underline{\nu},\overline{\nu}]$ it follows
$$
    \exp\left(\sqrt{\frac{H}{2}}\eta x - \frac{1}{4}\eta^2 t^{2H}\right) \leq \frac{1}{D}\text{ if } x \leq - \frac{\log(D)}{\overline{\eta}} =: \underline{L}
$$

and
$$
    \exp\left(\sqrt{\frac{H}{2}}\eta x - \frac{1}{4}\eta^2 t^{2H}\right) \geq D \text{ if } x \geq \sqrt{\frac{2}{\underline{H}}}\left(\frac{\log(D)}{\underline{\eta}} + \frac{1}{4}\frac{\overline{\eta}^2}{\underline{\eta}}\min\left\{t^{2\underline{H}},t^{2\overline{H}}\right\}\right) =: \overline{L}.
$$

Therefore, by the model dynamics \eqref{TruncatedrBergomiModel} we observe that for a fixed $t\in[0,T]$ it is sufficient to approximate the function $x \mapsto \exp\left(\sqrt{\frac{H}{2}}\eta x - \frac{1}{4}\eta^2 t^{2H}\right)$ on $(\underline{L},\overline{L})$.
 
\begin{lemma}\label{LemmaPhi4kApprox}
    There exists a constant $c_{30}>0$ such that for every $\varepsilon\in(0,1/2)$ and $k=1,...,N$ there exists a DNN $\Phi^{4,k}_{\varepsilon}$ with size$(\Phi^{4,k}_{\varepsilon})\leq c_{30}\varepsilon^{-2}$ and Lipschitz constant $L(\varepsilon)>0$ such that for every $H\in(\underline{H},\overline{H})$, $\eta\in(\underline{\eta},\overline{\eta})$ and $x,y\in(\underline{L},\overline{L})$ it holds
    \begin{equation}\label{LInftyPhi4k}
        \bigg\vert \exp\left(\sqrt{\frac{H}{2}}\eta x - \frac{1}{4}\eta^2 t^{2H}\right) - \Phi^{4,k}_{\varepsilon}(H,\eta,x) \bigg\vert \leq c_{30}\varepsilon
    \end{equation}

    and 
    \begin{equation}\label{LipschitzPhi4k}
        \vert \Phi^{4,k}_{\varepsilon}(H,\eta,x) - \Phi^{4,k}_{\varepsilon}(H,\eta,y) \vert \leq L(\varepsilon) \vert x - y\vert.
    \end{equation}
\end{lemma}

\begin{proof}
    Let $\varepsilon\in(0,1/2)$, $k=1,...,N$, $\underline{K}:= \min\{\underline{H}, \underline{\eta}, \underline{L}\}$ and $\overline{K}:= \max\{\overline{H}, \overline{\eta}, \overline{L}\}$. Then by Theorem 4.1 in \cite{Gühring2020} and the corresponding rescaling of the domain there exists a DNN $\Phi^{4,k}_\varepsilon:(\underline{K},\overline{K})^3 \rightarrow \mathbb{R}^+$ with size$(\Phi^{4,k}_\varepsilon)\leq C_1 \varepsilon^{-2}$ for some constant $C_1>0$ and such that 
    \begin{equation}\label{DNN4k}
        \Vert g_k - \Phi^{4,k}_\varepsilon \Vert_{\mathcal{W}^{1,\infty}(\underline{K},\overline{K})^3}\leq \varepsilon,
    \end{equation}

    where $g_k(x,y,z) := \exp(\sqrt{x/2}yz - 1/4 y^2 t_k^{2x})$ and $\mathcal{W}^{1,\infty}$ denotes the Sobolev norm. By Morrey's inequality and the triangle inequality there is a constant $C_2>0$ such that 
    \begin{equation}\label{FirstMorrey}
        \Vert \Phi^{4,k} \Vert_{C^{0,1}(\underline{K},\overline{K})^3} \leq C_2 \Vert \Phi^{4,k} \Vert_{\mathcal{W}^{1,\infty}(\underline{K},\overline{K})^3}  \leq C_2\left(\varepsilon + \Vert g_k \Vert_{\mathcal{W}^{1,\infty}(\underline{K},\overline{K})^3} \right),
    \end{equation}

    where $C^{0,1}$ denotes the H\"older norm with H\"older constant $1$ so that \eqref{FirstMorrey} is in particular a bound for \eqref{LipschitzPhi4k}. Similarly, we get again by Morrey's inequality for a constant $C_3>0$
    $$
        \sup\left\{\vert g_k(x) - \Phi^{4,k}_\varepsilon(x) \vert \,\middle\vert\, x\in (-L,L)^3\right\}\leq \Vert g_k - \Phi^{4,k}_\varepsilon \Vert_{C^{0,1}(-L,L)^3} \leq C_3 \Vert g_k - \Phi^{4,k}_\varepsilon \Vert_{\mathcal{W}^{1,\infty}(-L,L)^3}\leq C_3\varepsilon.
    $$

    The result follows for $c_{30}:=\max\{C_1,C_2,C_3\}$ and $L(\varepsilon) := c_{30}\left(\varepsilon + \Vert g_k \Vert_{\mathcal{W}^{1,\infty}(\underline{K},\overline{K})^3} \right)$.
\end{proof}

By Proposition 3.16 in \cite{Petersen2022} we can approximate the multiplication of two bounded factors by a DNN $\Phi^{\text{mult},2}_\varepsilon$ with size$(\Phi^{\text{mult},2}_\varepsilon)\leq c_{31} \varepsilon^{-2}$ up to an error $\varepsilon$ for a constant $c_{31}>0$. Therefore, for every $k=1,...,N$ we introduce the random variables 
\begin{equation}\label{VtildeDef}
    \tilde{V}^{\varepsilon,D}_{t_k} := \Phi^{\text{mult},2}_\varepsilon\left(\Phi^3_\varepsilon(\nu),\Phi^{\text{min}}_D \circ \Phi^{\text{max}}_{1/D}\circ \Phi^{4,k}_\varepsilon(H,\eta,\tilde{\mathfrak{X}}^{\varepsilon}_{t_k}(H))\right),
\end{equation}

where $\tilde{\mathfrak{X}}^{\varepsilon}$ is defined in \eqref{VolterraDNN}. We obtain the following approximation error for a fixed grid point. 

\begin{lemma}\label{DNNVarianceApprox}
    There exist constants $c_{32},c_{33}>0$ such that for every $k=1,...,N$, $\varepsilon\in(0,1/2)$ and $(\nu,\eta,\rho,H)\in\Theta^{\text{rBergomi}}$ it holds
    $$
       \mathbb{E}\left[\left\vert \sqrt{\bar{V}^D_{t_k}} - \tilde{V}^{\varepsilon,D}_{t_k}\right\vert^4\right] \leq c_{32}\varepsilon^4 + c_{33}h^{4H},
    $$
    where $\sqrt{V^D_{t_k}}$ is defined in \eqref{DiscreteTimeRBergomiV} for $t_k\in[0,T]$.
\end{lemma}

\begin{proof}
    Let $k=1,...,N, \varepsilon\in(0,1/2)$ and $(\nu,\eta,\rho,H)\in\Theta^{\text{rBergomi}}$. For sake of simplicity we introduce the following notation 
    $$
        M_{\mathfrak{X}} := \min\left(\max\left(\frac{1}{D},\exp\left(\eta \sqrt{\frac{H}{2}}{\mathfrak{X}}_{t_k} - \frac{1}{4}\eta^2 {t_k}^{2H} \right)\right), D\right),\, \tilde{M}_{\mathfrak{X}} := \Phi^{\text{min}}_D \circ \Phi^{\text{max}}_{1/D}\circ \Phi^{4,k}_\varepsilon(H,\eta,\tilde{\mathfrak{X}}^{\varepsilon}_{t_k}(H)).
    $$
    
    Then by \eqref{SquaredTriangular}, \eqref{ApproxDNNSqrt}, \eqref{LInftyPhi4k}, Lemma \ref{ApproxDNNVolterra} and the constant $L(\varepsilon)$ derived in the proof of Lemma \ref{LemmaPhi4kApprox}, we obtain that
    \begin{align*}
        &\mathbb{E}\left[\left\vert \sqrt{V^D_{t_k}} - \tilde{V}^{\varepsilon,D}_{t_k}\right\vert^4\right] = \mathbb{E}\left[\left\vert \sqrt{\nu}M_{\mathfrak{X}} - \Phi^{mult,2}_{\varepsilon}\left(\Phi^3(\nu),\tilde{M}_{\mathfrak{X}}\right)\right\vert^4\right] \\[2mm]
        \leq & 64 \mathbb{E}\left[\left\vert \sqrt{\nu}M_{\mathfrak{X}} - \sqrt{\nu}\tilde{M}_{\mathfrak{X}}\right\vert^4\right] + 64 \mathbb{E}\left[\left\vert \sqrt{\nu}\tilde{M}_{\mathfrak{X}} - \Phi^3(\nu)\tilde{M}_{\mathfrak{X}}\right\vert^4\right] + 8 \mathbb{E}\left[\left\vert \Phi^3(\nu)\tilde{M}_{\mathfrak{X}} - \Phi^{mult,2}_{\varepsilon}\left(\Phi^3(\nu),\tilde{M}_{\mathfrak{X}}\right)\right\vert^4\right]\\[2mm]
        \leq & 64 \bar{\nu}^2 \mathbb{E}\left[\left\vert M_{\mathfrak{X}} - \tilde{M}_{\mathfrak{X}}\right\vert^4\right] + \varepsilon^4 (64 D^4 + 8) \\[2mm]
        \leq & 512 \bar{\nu}^2 \mathbb{E}\left[\left\vert \exp\left(\sqrt{\frac{H}{2}}\eta \mathfrak{X}_{t_k} - \frac{1}{4}\eta^2 t^{2H}_k\right) - \Phi^{4,k}_{\varepsilon}(H,\eta,\mathfrak{X}_{t_k})\right\vert^4\right] + 512 \bar{\nu}^2 \mathbb{E}\left[\left\vert \Phi^{4,k}_{\varepsilon}(H,\eta,\mathfrak{X}_{t_k}) - \Phi^{4,k}_{\varepsilon}(H,\eta,\tilde{\mathfrak{X}}_{t_k})\right\vert^4\right] \\[2mm] 
        &\phantom{...}+ \varepsilon^4 (64 D^4 + 8)\\[2mm] 
        \leq & \varepsilon^4(512\bar{\nu}^2 + 64 D^4 + 8) + 512 \bar{\nu}^2 c_{30}^4 \left(\varepsilon + \Vert g_k \Vert_{\mathcal{W}^{1,\infty}(-K,K)^3}\right)^4 \mathbb{E}\left[\vert \mathfrak{X}_{t_k} - \tilde{\mathfrak{X}}^{\varepsilon}_{t_k} \vert^4\right] \\[2mm] 
        \leq & \varepsilon^4\left(512 \bar{\nu}^2 + 64D^4 + 8 + 4096 \bar{\nu}^2c_{30} c_{28} (1+\Vert g_k \Vert_{\mathcal{W}^{1,\infty}(\underline{K},\overline{K})^3} \right)\\[2mm] 
        &\phantom{...}+  h^{4H} \left(4096 \bar{\nu}^2 c_{30}^4 c_{27} \underline{H}^{-2}(1+\Vert g_k\Vert_{\mathcal{W}^{1,\infty}(\underline{K},\overline{K})^3}\right) \\[2mm]
        =: & c_{32} \varepsilon^4 + c_{33}h^{4H},
    \end{align*}

    where we used the fact that clearly $\Vert g_k \Vert_{\mathcal{W}^{1,\infty}(-K,K)^3} < \infty$.
\end{proof}

Now that we have constructed an approximation for the variance process for fixed grid points, we introduce an approximating process for the asset process $X^D$. First of all, we have by Theorem \ref{DNNApproxC1} that there exists a DNN $\Phi^5_{\varepsilon}:[-1,0] \rightarrow [0,1]$ such that for every $\rho\in[-1,0]$ it holds
\begin{equation}\label{DNNApproxSQRTRho}
    \vert \Phi^5_{\varepsilon}(\rho) - \sqrt{1-\rho^2}\vert < \varepsilon \text{ with size}(\Phi^5_{\varepsilon}) \leq c_{34}\varepsilon^{-2}
\end{equation}

for for some constant $c_{34}>0$. Moreover, we consider a constant $R\geq 1$ and the stopping time $\tau := \tau_B \wedge \tau_{B^\perp}$, where
$$
    \tau_B := \inf\{t\in[0,T]\mid B_t \notin [-R,R]\} \text{ and }\tau_{B^\perp} := \inf\{t\in[0,T]\mid B^\perp_t \notin [-R,R]\}.
$$

We denote by $\Phi^{\text{mult},3}_{\varepsilon}$ the DNN approximating the multiplication of three bounded factors each lying in the interval $[-4RN(1+5\widetilde{D}R)^N, 4RN(1+5\widetilde{D}R)^N]$ for $\widetilde{D}:= D(1/2 + \sqrt{v})+1/2$ up to an error $\varepsilon\in(0,1/2)$. Then Proposition 3.16 in \cite{Petersen2022} there exists a constant $C>0$ such that it holds
\begin{equation}\label{SizeMultDNN}
    \text{size}(\Phi^{mult,3}_{\varepsilon}) \leq C \log_2(4RN(1+5\widetilde{D}R)^N)\varepsilon^{-2} \leq 5C(1+\widetilde{D})NR\varepsilon^{-2} =: c_{35}NR\varepsilon^{-2}.
\end{equation}

Similarly to Section \ref{Section3}, we introduce the following continuous-time interpolation scheme 
\begin{equation}\label{rBergomiDNNContTimeScheme}
    \diff \tilde{X}^{\varepsilon,D}_{t} = \Phi^{\text{mult},3}_{\varepsilon}(\tilde{X}^{\varepsilon,D}_{\floor{t}}, \tilde{V}^{\varepsilon,D}_{\floor{t}}, \rho) \diff{B^{\tau}_t} + \Phi^{\text{mult},3}_{\varepsilon}(\tilde{X}^{\varepsilon,D}_{\floor{t}}, \tilde{V}^{\varepsilon,D}_{\floor{t}}, \Phi^5_{\varepsilon}(\rho)) \diff{B^{\perp,\tau}_t}
\end{equation}

with $\tilde{X}^{\varepsilon,D}_{0} = x$, $\floor{t} := \max\lbrace t_k\mid k=0,...,N,\; t_k \leq t \rbrace$. Moreover, we see that the solution of \eqref{rBergomiDNNContTimeScheme} coincides with 

\begin{equation}\label{rBergomiDNNDiscreteTimeScheme}
    \tilde{X}^{\varepsilon,D}_{t_{k+1}} = \tilde{X}^{\varepsilon,D}_{t_k} + \Phi^{\text{mult},3}_{\varepsilon}(\tilde{X}^{\varepsilon,D}_{t_k}, \tilde{V}^{\varepsilon,D}_{t_k}, \rho) (B^\tau_{t_{k+1}} - B^\tau_{t_k}) + \Phi^{\text{mult},3}_{\varepsilon}(\tilde{X}^{\varepsilon,D}_{t_k}, \tilde{V}^{\varepsilon,D}_{t_k}, \Phi^5_{\varepsilon}(\rho)) (B^{\perp,\tau}_{t_{k+1}} - B^{\perp,\tau}_{t_k})
\end{equation}

at any time point $t=t_k$ for every $k=1,...,N$. In fact, for any $k=1,\dots,N$ we obtain for \eqref{VtildeDef} that 
\begin{equation}\label{BoundVTilde}
    \vert \tilde{V}^{\varepsilon,D}_{t_k}\vert \leq \varepsilon + D\vert \Phi^3_\varepsilon(\nu) \vert \leq \varepsilon + D(\varepsilon + \sqrt{\nu}) \leq \frac{1}{2} + D\left(\frac{1}{2} + \sqrt{\nu}\right)= \widetilde{D}
\end{equation} 

and for \eqref{rBergomiDNNDiscreteTimeScheme} that
\begin{equation}
    \vert \tilde{X}^\varepsilon_{t_{k}} \vert \leq \vert \tilde{X}^\varepsilon_{t_{k-1}} \vert \left(1 + 5\vert \tilde{V}^{\varepsilon,D}_{t_{k-1}}\vert R\right) + 2R
\end{equation}

using that $\varepsilon<1/2$ and $D\geq 1$. So inductively it follows that 
$$
    \vert \tilde{X}^\varepsilon_{t_{k}} \vert \leq (1+5\tilde{D}R)^N + 2R\sum_{k=0}^{N-1}(1+5\tilde{D}R)^k \leq 4RN (1+5\tilde{D}R)^N
$$

for any $k=1,\dots,N$. So that the system \eqref{rBergomiDNNDiscreteTimeScheme} and solution is indeed well-defined since $\vert\rho\vert$,$\vert\Phi^5_\varepsilon(\rho)\vert$,$ \vert\tilde{V}^{\varepsilon,D}_{t_{k}}\vert$,$\vert \tilde{X}^\varepsilon_{t_{k}}\vert \leq 4RN (1+5\widetilde{D}R)^N$. The following lemma establishes a bound for the second moment of $\tilde{X}^{\varepsilon,D}$. 
\begin{lemma}\label{SecondMomBoundDNNXrBergomi}
    There exists a constant $c_{36}>0$ such that for every $x>0$, $\varepsilon\in(0,1/2)$ and $(\nu,\eta,\rho,H)\in\Theta^{\text{rBergomi}}$ it holds
    $$
        \mathbb{E}\left[\sup_{t\in[0,T]} \left\vert \tilde{X}^{\varepsilon,D}_t\right\vert^2\right] \leq c_{36}(1+\vert x\vert^2).
    $$
\end{lemma}

\begin{proof}
    Let $x>0$, $\varepsilon\in(0,1/2)$ and $(\nu,\eta,\rho,H)\in\Theta^{\text{rBergomi}}$. For every $t\in[0,T]$ we define $G(t):= \mathbb{E}\left[\sup_{s\leq t} \left\vert \tilde{X}^{\varepsilon,D}_s\right\vert^2 \right]$. By the Burkolder-Davis-Gundy inequality there exists a constant $C>0$ such that together with \eqref{SquaredTriangular}, \eqref{BoundVTilde} and the Fubini-Tonelli theorem it follows
    \begin{align*}
        &G(t) \\[2mm]
        \leq& 3\left(\vert x\vert^2 + \mathbb{E}\left[\sup_{s\leq t} \left\vert \int_0^s \Phi^{mult,3}_{\varepsilon}(\tilde{X}^{\varepsilon,D}_{\floor{r}}, \tilde{V}^{\varepsilon,D}_{\floor{r}} , \rho) \diff{B^\tau_{r}}\right\vert^2\right] + \mathbb{E}\left[\sup_{s\leq t} \left\vert \int_0^s \Phi^{mult,3}_{\varepsilon}(\tilde{X}^{\varepsilon,D}_{\floor{r}}, \tilde{V}^{\varepsilon,D}_{\floor{r}} , \Phi^5(\rho)) \diff{B^{\perp,\tau}_{r}}\right\vert^2\right]\right) \\[2mm]
        \leq& 3\left(\vert x\vert^2 + C\mathbb{E}\left[ \int_0^{t\wedge \tau} \vert\Phi^{mult,3}_{\varepsilon}(\tilde{X}^{\varepsilon,D}_{\floor{r}}, \tilde{V}^{\varepsilon,D}_{\floor{r}},\rho) \vert^2 \diff{r}\right] + C\mathbb{E}\left[\int_0^{t\wedge \tau} \vert\Phi^{mult,3}_{\varepsilon}(\tilde{X}^{\varepsilon,D}_{\floor{r}}, \tilde{V}^{\varepsilon,D}_{\floor{r}} , \Phi^5(\rho))\vert^2 \diff{r}\right]\right)\\[2mm]
        \leq& 3\vert x \vert^2 + 3C\int_0^t \mathbb{E}\left[\vert\Phi^{mult,3}_{\varepsilon}(\tilde{X}^{\varepsilon,D}_{\floor{r}}, \tilde{V}^{\varepsilon,D}_{\floor{r}},\rho) \vert^2\right] \diff{r} + 3C\int_0^t \mathbb{E}\left[\vert\Phi^{mult,3}_{\varepsilon}(\tilde{X}^{\varepsilon,D}_{\floor{r}}, \tilde{V}^{\varepsilon,D}_{\floor{r}},\Phi^5(\rho)) \vert^2\right] \diff{r} \\[2mm]
        \leq& 3\vert x \vert^2 + 12C \varepsilon^2 t + 6C\int_0^t \mathbb{E}\left[\vert \tilde{X}^{\varepsilon,D}_{\floor{r}} \tilde{V}^{\varepsilon,D}_{\floor{r}} \rho \vert^2\right]\diff{r} + 6C\int_0^t \mathbb{E}\left[\vert \tilde{X}^{\varepsilon,D}_{\floor{r}} \tilde{V}^{\varepsilon,D}_{\floor{r}} \Phi^5(\rho) \vert^2\right]\diff{r} \\[2mm]
        \leq& 3\vert x\vert^2 + 3CT + 6C\widetilde{D}^2(1+\vert \Phi^5(\rho)\vert^2)\int_0^t G(r) \diff{r}.
    \end{align*}

    Using the boundedness of every $\tilde{X}^{\varepsilon,D}_{t_k}$, $k=1,...,N$, it follows that $G(t)\in L^1([0,T])$. Therefore, together with $\vert \Phi^5(\rho)\vert \leq \varepsilon + \sqrt{1-\rho^2}\leq 3/2$ we obtain by Gronwall's inequality 
    $$
        G(t) \leq 3(\vert x \vert^2 + CT)\exp\left(20 C\widetilde{D}^2T\right).
    $$
    The result follows with $c_{36} := 3\max(1,CT)\exp\left(20 C\widetilde{D}^2T\right)$.
\end{proof}

Finally, we obtain the following strong approximation result for the DNN scheme defined in \eqref{rBergomiDNNContTimeScheme}.

\begin{lemma}\label{DNNApproxrBergomi}
    There exists a constant $c_{37}>0$ such that for every $x>0$, $\varepsilon\in(0,1/2)$ and $(\nu,\eta,\rho,H)\in\Theta^{\text{rBergomi}}$ it holds
    $$
        \mathbb{E}\left[\sup_{t\in[0,T]}\left\vert \bar{X}^D_t - \tilde{X}^{\varepsilon,D}_t \right\vert^2\right] \leq c_{37}\left(\varepsilon^2 + h^{2\underline{H}} + \exp\left(-\frac{R^2}{4T}\right)\right)(1+\vert x\vert^2).
    $$
\end{lemma}

\begin{proof}
Let $x>0$, $\varepsilon\in(0,1/2)$ and $(\nu,\eta,\rho,H)\in\Theta^{\text{rBergomi}}$. For every $t\in[0,T]$ we define $G(t) := \mathbb{E}\left[\sup_{s\leq t} \vert \tilde{X}^{\varepsilon,D}_s - \bar{X}^D_s \vert^2 \right]$. Then by the Burkh\"older-Davis-Gundy inequality there exists a constant $C>0$ such that with \eqref{SquaredTriangular} it follows
\begin{align*}
    G(t)=& \mathbb{E}\bigg[\sup_{s\leq t} \bigg\vert \int_0^s \Phi^{mult,3}_{\varepsilon}(\tilde{X}^{\varepsilon,D}_{\floor{r}}, \tilde{V}^{\varepsilon,D}_{\floor{r}},\rho)\diff{B^\tau_r} + \int_0^s \Phi^{mult,3}_{\varepsilon}(\tilde{X}^{\varepsilon,D}_{\floor{r}}, \tilde{V}^{\varepsilon,D}_{\floor{r}},\Phi^5(\rho))\diff{B^{\perp,\tau}_r} \\[2mm] 
    &\phantom{...} - \int_0^s \bar{X}^D_{\floor{r}}\sqrt{V^D_{\floor{r}}}\rho\diff{B_r} - \int_0^s \bar{X}^D_{\floor{r}}\sqrt{V^D_{\floor{r}}}\sqrt{1-\rho^2}\diff{B^{\perp}_r}\bigg\vert^2\bigg] \\[2mm]
    \leq& 4 \mathbb{E}\bigg[\sup_{s\leq t} \bigg\vert \int_0^s \left[\Phi^{mult,3}_{\varepsilon}(\tilde{X}^{\varepsilon,D}_{\floor{r}}, \tilde{V}^{\varepsilon,D}_{\floor{r}},\rho) - \bar{X}^D_{\floor{r}}\sqrt{V^D_{\floor{r}}}\rho\right] \mathbbm{1}_{\{r<\tau\}}\diff{B_r}\bigg\vert^2\bigg]\\[2mm]
    &\phantom{...} + 4\mathbb{E}\bigg[\sup_{s\leq t} \bigg\vert \int_0^s \left[\Phi^{mult,3}_{\varepsilon}(\tilde{X}^{\varepsilon,D}_{\floor{r}}, \tilde{V}^{\varepsilon,D}_{\floor{r}},\Phi^5(\rho)) - \bar{X}^D_{\floor{r}}\sqrt{V^D_{\floor{r}}}\sqrt{1-\rho^2}\right] \mathbbm{1}_{\{r<\tau\}}\diff{B^\perp_r}\bigg\vert^2\bigg]\\[2mm]
    &\phantom{...} + 4 \mathbb{E}\bigg[\sup_{s\leq t} \bigg\vert \int_0^s \left[\bar{X}^D_{\floor{r}}\sqrt{\bar{V}^D_{\floor{r}}}\rho\left(\mathbbm{1}_{\{r<\tau\}} - 1\right)\right]\diff{B_r}\bigg\vert^2\bigg]\\[2mm]
    &\phantom{...} + 4 \mathbb{E}\bigg[\sup_{s\leq t} \bigg\vert \int_0^s \left[\bar{X}^D_{\floor{r}}\sqrt{\bar{V}^D_{\floor{r}}}\sqrt{1-\rho^2}\left(\mathbbm{1}_{\{r<\tau\}} - 1\right)\right]\diff{B^\perp_r}\bigg\vert^2\bigg]\\[2mm]
    \leq& 4C \int_0^t \mathbb{E}\left[\left\vert\Phi^{mult,3}_{\varepsilon}(\tilde{X}^{\varepsilon,D}_{\floor{r}}, \tilde{V}^{\varepsilon,D}_{\floor{r}},\rho) - \bar{X}^D_{\floor{r}}\sqrt{V^D_{\floor{r}}} \rho\right\vert^2\right]\diff{r}\\[2mm]
    &\phantom{...} + 4C \int_0^t \mathbb{E}\left[\left\vert\Phi^{mult,3}_{\varepsilon}(\tilde{X}^{\varepsilon,D}_{\floor{r}}, \tilde{V}^{\varepsilon,D}_{\floor{r}},\Phi^5(\rho)) - \bar{X}^D_{\floor{r}}\sqrt{V^D_{\floor{r}}}\sqrt{1-\rho^2} \right\vert^2\right]\diff{r}\\[2mm]
    &\phantom{...} + 4C \int_0^t \mathbb{E}\left[\left\vert \bar{X}^D_{\floor{r}}\sqrt{V^D_{\floor{r}}}(\mathbbm{1}_{\{r<\tau\}} - 1)\right\vert^2\right]\diff{r}.
\end{align*}

By the Cauchy-Schwarz inequality we have that
\begin{equation}\label{ExpectationStoopingTime}
    \mathbb{E}\left[\left\vert \bar{X}^D_{\floor{r}}\sqrt{V^D_{\floor{r}}}(\mathbbm{1}_{\{r<\tau\}} - 1)\right\vert^2\right] \leq \bar{\nu} D^2 \mathbb{E}\left[\sup_{t\in[0,T]} \vert \bar{X}^D_t\vert^4\right]^{1/2}\mathbb{Q}(r>\tau)^{1/2}.
\end{equation}

From the independence of $B$ and $B^\perp$ we get that $\mathbb{Q}(r>\tau) \leq 2 \mathbb{Q}(r>\tau_B)$. Further, from the reflection principle and symmetry of the Brownian motion we obtain
\begin{equation}\label{ProbabilityHittingTime}
    \mathbb{Q}(r>\tau_B) \leq 4\mathbb{Q}(B_r > R) = 4\mathbb{Q}(B_1 > R/\sqrt{r}) \leq \frac{2\sqrt{2r}}{R\sqrt{\pi}}\exp\left(-\frac{R^2}{2r}\right).
\end{equation}

Since the bound in \eqref{ProbabilityHittingTime} is increasing in $r$, combining Lemma \ref{MomentBoundrBergomiXBar}, \eqref{ExpectationStoopingTime} and is \eqref{ProbabilityHittingTime} yields
\begin{equation}\label{DNNApproxXrBergomiBound1}
    \int_0^t \mathbb{E}\left[\left\vert \bar{X}^D_{\floor{r}}\sqrt{V^D_{\floor{r}}}(\mathbbm{1}_{\{r<\tau\}} - 1)\right\vert^2\right]\diff{r} \leq 2\bar{\nu}D^2 \sqrt{c_{23}} (1+\vert x \vert^2)t^{5/4} \exp\left(-\frac{R^2}{4t}\right).
\end{equation}

By \eqref{DNNApproxSQRTRho}, \eqref{SquaredTriangular} and the properties of $\Phi^{mult,3}_{\varepsilon}$, we get for every $r\in[0,t]$ 
\begin{align*}
    &\mathbb{E}\left[\left\vert\Phi^{mult,3}_{\varepsilon}(\tilde{X}^{\varepsilon,D}_{\floor{r}}, \tilde{V}^{\varepsilon,D}_{\floor{r}},\Phi^5(\rho)) - \bar{X}^D_{\floor{r}}\sqrt{V^D_{\floor{r}}}\sqrt{1-\rho^2} \right\vert^2\right] \\[2mm]
    \leq& 2\varepsilon^2 + 4\mathbb{E}\left[\left\vert\tilde{X}^{\varepsilon,D}_{\floor{r}} \tilde{V}^{\varepsilon,D}_{\floor{r}}\Phi^5(\rho) - \bar{X}^D_{\floor{r}}\sqrt{V^D_{\floor{r}}}\Phi^5(\rho) \right\vert^2\right] + 4\mathbb{E}\left[\left\vert\bar{X}^D_{\floor{r}}\sqrt{V^D_{\floor{r}}}\Phi^5(\rho) - \bar{X}^D_{\floor{r}}\sqrt{V^D_{\floor{r}}}\sqrt{1-\rho^2} \right\vert^2\right] \\[2mm]
    \leq& 2\varepsilon^2 + 8\mathbb{E}\left[\left\vert \tilde{X}^{\varepsilon,D}_{\floor{r}} \tilde{V}^{\varepsilon,D}_{\floor{r}}\Phi^5(\rho) - \bar{X}^{D}_{\floor{r}}\tilde{V}^{\varepsilon,D}_{\floor{r}}\Phi^5(\rho)\right\vert^2\right] + 8\mathbb{E}\left[\left\vert\bar{X}^{D}_{\floor{r}}\tilde{V}^{\varepsilon,D}_{\floor{r}}\Phi^5(\rho) - \bar{X}^D_{\floor{r}}\sqrt{V^D_{\floor{r}}}\Phi^5(\rho) \right\vert^2\right] \\[2mm]
    &\phantom{...} + 4\bar{\nu}D^2 \mathbb{E}\left[\sup_{t\in[0,T]} \vert \bar{X}^D_t\vert^2\right] \varepsilon^2.
\end{align*}

Hence it follows
\begin{equation}\label{DNNApproxBound1}
    \mathbb{E}\left[\left\vert \tilde{X}^{\varepsilon,D}_{\floor{r}} \tilde{V}^{\varepsilon,D}_{\floor{r}}\Phi^5(\rho) - \bar{X}^{D}_{\floor{r}}\tilde{V}^{\varepsilon,D}_{\floor{r}}\Phi^5(\rho) \right\vert^2\right] \leq \tilde{D}^2 \vert \Phi^5(\rho)\vert^2 G(r)
\end{equation}

and by the Cauchy-Schwarz inequality
\begin{equation}\label{DNNApproxBound2}
\mathbb{E}\left[\left\vert\bar{X}^{D}_{\floor{r}}\tilde{V}^{\varepsilon,D}_{\floor{r}}\Phi^5(\rho) - \bar{X}^D_{\floor{r}}\sqrt{V^D_{\floor{r}}}\Phi^5(\rho) \right\vert^2\right] \leq \vert \Phi^5(\rho)\vert^2 \mathbb{E}\left[\sup_{t\in[0,T]} \vert \bar{X}^D_t\vert^4\right]^\frac{1}{2} \mathbb{E}\left[\vert \tilde{V}^{\varepsilon,D}_{\floor{r}} - \sqrt{V^D_{\floor{r}}}\vert^4\right]^\frac{1}{2}.
\end{equation}

So combining \eqref{DNNApproxBound1}, \eqref{DNNApproxBound2},  Lemma \ref{MomentBoundrBergomiXBar}, Lemma \ref{DNNVarianceApprox} and the fact that $\vert \Phi^5(\rho)\vert \leq 3/2$ for every $\rho\in[-1,0]$ yields 
\begin{align}\label{DNNApproxXrBergomiBound2}
        &\mathbb{E}\left[\left\vert\Phi^{mult,3}_{\varepsilon}(\tilde{X}^{\varepsilon,D}_{\floor{r}}, \tilde{V}^{\varepsilon,D}_{\floor{r}},\Phi^5(\rho)) - \bar{X}^D_{\floor{r}}\sqrt{V^D_{\floor{r}}}\sqrt{1-\rho^2} \right\vert^2\right] \nonumber \\[2mm]
        \leq& \varepsilon^2\left(2 + 18\sqrt{c_{23}c_{32}} + 4\bar{\nu}D^2 c_{22})(1+\vert x \vert^2) \right) + h^{2H}18\sqrt{c_{23}c_{33}}(1+\vert x \vert^2) + 18\tilde{D}^2 G(r).
\end{align}

By similar arguments as above it also holds
\begin{align}\label{DNNApproxXrBergomiBound3}
    &\mathbb{E}\left[\left\vert\Phi^{mult,3}_{\varepsilon}(\tilde{X}^{\varepsilon,D}_{\floor{r}}, \tilde{V}^{\varepsilon,D}_{\floor{r}},\rho) - \bar{X}^D_{\floor{r}}\sqrt{V^D_{\floor{r}}} \rho\right\vert^2\right] \nonumber \\[2mm]
    \leq& \varepsilon^2\left(2 + 4 \sqrt{c_{23} c_{32}}(1+\vert x \vert^2)\right) + 4 h^{2H}\sqrt{c_{23}c_{33}}(1+\vert x \vert^2) + 4\tilde{D}^2 G(r).
\end{align}

Finally, we get that 
\begin{align*}
   G(t) \leq& 12CT\max\{2\bar{\nu}^2 D^2 \sqrt{c_{23}}T^{1/4}, 4 + 22\sqrt{c_{24}c_{32}} + 4\bar{\nu}^2 D^2 c_{22} , 22 \sqrt{c_{24}c_{33}}\}\\[2mm]
   &\phantom{...}\left[\varepsilon^2 + h^{2H} + \exp\left(-\frac{R^2}{4T}\right)\right](1+\vert x\vert^2)+ 88 C \tilde{D}^2 \int_0^t G(r) \diff{r}.
\end{align*}

Using Lemma \ref{SecondMomBoundDNNXrBergomi} and the triangle inequality, we see that $G(t)\in L^1([0,T])$ so that we can apply Gronwall's inequality to obtain 
$$
    G(t) \leq c_{37}\left(\varepsilon^2 + h^{2\overline{H}} + \exp\left(-\frac{R^2}{2T}\right)\right)(1+\vert x \vert^2)
$$

for $c_{37} := 12CT\max\{2\bar{\nu}^2 D^2 \sqrt{c_{23}}T^{1/4}, 4 + 22\sqrt{c_{24}c_{32}} + 4\bar{\nu}^2 D^2 c_{22} , 22 \sqrt{c_{24}c_{33}}\}\exp(88 C \tilde{D}^2 T)$.
    
\end{proof}


\subsection{DNN approximation of the pricing function}
Similar to the previous section for some $n\in\mathbb{N}$ we consider a function $\varphi:\mathbb{R}^+ \times \mathbb{R}^n \rightarrow \mathbb{R}^+$, which represents a parametric payoff function of some contingent claim. We make the following assumption about the approximation of $\varphi$ by a DNN.

\begin{assumption}\label{AssumPayoffrBergomi}
    There exist constants $C,q>0$ and for every $\varepsilon\in(0,1/2)$ a neural network $\phi_\varepsilon: \mathbb{R}^+ \times \mathbb{R}^n \rightarrow \mathbb{R}^+$ such that for every $\varepsilon\in(0,1/2)$ it holds 
    \begin{enumerate}
        \item[(i)] for all $x \in \mathbb{R}^+$ and $K\in\mathbb{R}^{n}$
        \begin{equation*}
            \vert \varphi(x,K) - \phi_{\varepsilon}(x,K)\vert \leq C \varepsilon (1+\vert x \vert + \Vert K \Vert)
        \end{equation*}
        
        \item[(ii)] 
        \begin{equation*}
            \text{size}(\phi_{\varepsilon}) \leq C \varepsilon^{-q}
        \end{equation*}
        
        \item[(iii)] 
        \begin{equation*}
            \text{Lip}(\phi_{\varepsilon}) := \sup_{(x_1,K_1),(x_2,K_2)\in\mathbb{R}^+\times \mathbb{R}^n} \frac{\vert\phi_{\varepsilon}(x_1,K_1) - \phi_{\varepsilon}(x_2,K_2)\vert}{\vert x_1 - x_2\vert + \Vert K_1 - K_2 \Vert} \leq C \varepsilon^{-1/2}
        \end{equation*}
    \end{enumerate}
    
    and it holds for a probability measure $\mu$ on $\mathbb{R}^+ \times \Theta^{\text{rBergomi}} \times \mathbb{R}^{n}$ that
    \begin{enumerate}
        \item[(iv)] 
        \begin{equation}\label{AssumPayoffProbMeasurerBergomi}
            \int_{\mathbb{R}^+ \times \Theta^{\text{rBergomi}} \times \mathbb{R}^{n}}(1 + \vert x \vert^2 + \Vert \theta \Vert^2 + \Vert K \Vert^2) \mu(\diff{x},\diff{\theta},\diff{K}) \leq C.
        \end{equation}
    \end{enumerate}
\end{assumption}

\begin{theorem}
    Let $D\geq 1$. For $n\in\mathbb{N}$ let $\varphi:\mathbb{R}^+\times \mathbb{R}^{n}\rightarrow \mathbb{R}^{+}$ be a payoff function satisfying Assumption \ref{AssumPayoffrBergomi}. Then there exist constants $\mathfrak{C},\mathfrak{q}>0$ and for any target accuracy $\bar{\varepsilon}\in(0,1/2)$ a neural network $U^{\text{rBergomi}}_{\bar{\varepsilon}}:\mathbb{R}^{+}\times \Theta^{\text{rBergomi}} \times \mathbb{R}^{n}\rightarrow\mathbb{R}^+$ such that 
    \begin{enumerate}
        \item[(i)]
        \begin{equation}\label{PriceApproxSizeBoundrBergomi}
            \text{size}(U^{\text{rBergomi}}_{\bar{\varepsilon}})\leq \mathfrak{C}\bar\varepsilon^{-\mathfrak{q}},
        \end{equation}
        \item[(ii)]
        \begin{equation}\label{ApproxErrorrBergomi}
            \left(\int_{\mathbb{R}^+\times\Theta^{\text{rBergomi}}\times\mathbb{R}^{n}} \left\vert \mathbb{E}[\varphi(X^D_T,K)] - U^{\text{rBergomi}}_{\bar{\varepsilon}}(x,\theta,K) \right\vert^2 \mu(\text{d}x,\text{d}\theta,\text{d}K)\right)^{1/2} < \bar\varepsilon,
        \end{equation}
    \end{enumerate}

    where $\mu$ is the probability measure from Assumption \ref{AssumPayoffrBergomi}. The constants $\mathfrak{C}$, $\mathfrak{q}$ do not depend on $\bar{\varepsilon}$.
\end{theorem}

\begin{proof}
Let $D\geq1$, $n\in\mathbb{N}$, $\bar\varepsilon\in(0,1/2)$ and $M\in\mathbb{N}$. As in the proof of Theorem \ref{MainTheorem} we want to consider the Monte-Carlo sum
$$
    \frac{1}{M}\sum_{i=1}^M \phi_\varepsilon(\Tilde{X}^{\varepsilon, D,i}_T(\omega),K),
$$

where $(\Tilde{X}^{\varepsilon,D,i}_T)_{i=1,...,M}$ are iid copies of $\tilde{X}^{\varepsilon,D}_T$ with $\tilde{X}^{\varepsilon,D}$ defined in \eqref{rBergomiDNNContTimeScheme} and $\phi_\varepsilon$ from Assumption \ref{AssumPayoffrBergomi}. In the first part of the proof, we show that there exists a $\omega\in\Omega$ such that the approximation error satisfies
$$
    \left(\int_{\mathbb{R}^+\times\Theta^{\text{rBergomi}}\times\mathbb{R}^{n}} \left\vert \mathbb{E}[\varphi(X^D_T,K)] - \frac{1}{M}\sum_{i=1}^M \phi_{\varepsilon}(\tilde{X}^{\varepsilon,D,i}_T(\omega), K) \right\vert^2 \mu(\text{d}x,\text{d}\theta,\text{d}K)\right)^{1/2} < \bar{\varepsilon}.
$$

For this purpose we consider the $L^2$-approximation error 
\begin{align}\label{L2ApproxErrorrBergomi}
    &\int_{\mathbb{R}^+\times\Theta^{\text{rBergomi}}\times\mathbb{R}^{n}} \mathbb{E}\left[\left\vert \mathbb{E}\left[\varphi(X^D_T,K)\right] - \frac{1}{M}\sum_{i=1}^M \phi_{\varepsilon}(\tilde{X}^{\varepsilon,D,i}_T, K) \right\vert^2\right] \mu(\diff{x},\diff{\theta},\diff{K}) \nonumber \\[2mm]
    =& \int_{\mathbb{R}^+\times\Theta^{\text{rBergomi}}\times\mathbb{R}^{n}} \left\vert \mathbb{E}\left[\varphi(X^D_T,K)\right] - \mathbb{E}\left[\phi_{\varepsilon}(\tilde{X}^{\varepsilon,D,1}_T,K)\right] \right\vert^2 \nonumber \\[2mm] &\phantom{.....}+ \frac{1}{M}\mathbb{E}\left[\left\vert \mathbb{E}\left[\bar{\phi}_{\varepsilon}(\tilde{X}^{\varepsilon,D,1}_T,K)\right] - \bar{\phi}_{\varepsilon}(\tilde{X}^{\varepsilon,D,1}_T,K)\right\vert^2\right]\mu(\diff{x},\diff{\theta},\diff{K}),
\end{align}

where $\bar{\phi}_{\varepsilon}(\tilde{X}^{\varepsilon,D,1}_{t_N},K) := \phi_{\varepsilon}(\tilde{X}^{\varepsilon,D,1}_{t_N},K) - \phi_{\varepsilon}(0,K)$. Again we study the integrands individually. For the second integrand we obtain with Lemma \ref{SecondMomBoundDNNXrBergomi} that
\begin{align}\label{L2ApproxErrorrBergomiBound1}
    &\mathbb{E}\left[\left\vert \mathbb{E}\left[\bar{\phi}_{\varepsilon}(\tilde{X}^{\varepsilon,D,1}_T,K)\right] - \bar{\phi}_{\varepsilon}(\tilde{X}^{\varepsilon,D,1}_T,K)\right\vert^2\right] \leq \mathbb{E}\left[\left\vert \phi_{\varepsilon}(\tilde{X}^{\varepsilon,D,1}_T,K) - \phi_{\varepsilon}(0,K)\right\vert^2\right] \nonumber \\[2mm]
    \leq & C^2 \varepsilon^{-1} \mathbb{E}\left[\left\vert\tilde{X}^{\varepsilon,D,1}_T\right\vert^2\right] \leq  C^2\varepsilon^{-1}c_{36}(1+\vert x \vert^2).
\end{align}

For the first integrand by Jensen's inequality, \eqref{SquaredTriangular} and Assumption \ref{AssumPayoffrBergomi} it follows that
\begin{align}\label{L2ApproxErrorrBergomiBound2}
    &\left\vert \mathbb{E}\left[\varphi(X^D_T,K)\right] - \mathbb{E}\left[\phi_{\varepsilon}(\tilde{X}^{\varepsilon,D,1}_T,K)\right] \right\vert^2 \nonumber\\[2mm]
    \leq& 2\mathbb{E}\left[\left\vert \varphi(X_T^D,K) - \phi_{\varepsilon}(X_T^D,K) \right\vert^2\right] + 2\mathbb{E}\left[\left\vert \phi_{\varepsilon}(X_T^D,K) - \phi_{\varepsilon}(\tilde{X}_T^{\varepsilon,D,1},K) \right\vert^2\right] \nonumber\\[2mm]
    \leq& 6C^2 \varepsilon^2\left(1 + \mathbb{E}\left[\vert X_T^D \vert^2 \right] + \Vert K \Vert^2 \right) + 2C^2 \varepsilon^{-1} \mathbb{E}\left[\left\vert X_T^D - \tilde{X}^{\varepsilon,D,1}_T\right\vert^2\right].
\end{align}

Moreover, by \eqref{SquaredTriangular} we obtain 
\begin{align*}
    \mathbb{E}\left[\left\vert X_T^D - \tilde{X}^{\varepsilon,D,1}_T\right\vert^2\right] \leq& 2\mathbb{E}\left[\left\vert X_T^D - \bar{X}^{D}_T\right\vert^2\right] + 2\mathbb{E}\left[\left\vert\bar{X}^{D}_T - \tilde{X}^{\varepsilon,D,1}_T\right\vert^2\right] \\[2mm]
    \leq& 2\mathbb{E}\left[\sup_{t\in[0,T]} \left\vert X_t^D - \bar{X}^{D}_t\right\vert^2\right] + 2\mathbb{E}\left[\sup_{t\in[0,T]}\left\vert\bar{X}^{D}_t - \tilde{X}^{\varepsilon,D,1}_t\right\vert^2\right].
\end{align*}

Therefore, by Lemma \ref{StrongApproxXDbar} and Lemma \ref{DNNApproxrBergomi} we get
\begin{equation}\label{BoundDiffXXtilde}
    \mathbb{E}\left[\left\vert X_T^D - \tilde{X}^{\varepsilon,D,1}_T\right\vert^2\right] \leq 2(1+\vert x\vert^2)\left(c_{37}\varepsilon^2 + (c_{24}+c_{37})h^{\underline{H}} + c_{37} \exp\left(-\frac{R^2}{2T}\right)\right).
\end{equation}

Inserting \eqref{L2ApproxErrorrBergomiBound1}, \eqref{L2ApproxErrorrBergomiBound2} and \eqref{BoundDiffXXtilde} in \eqref{L2ApproxErrorrBergomi} as well as applying Assumption \ref{AssumPayoffrBergomi} yield
\begin{align*}
    &\int_{\mathbb{R}^+\times\Theta^{\text{rBergomi}}\times\mathbb{R}^{n}} \mathbb{E}\left[\left\vert \mathbb{E}\left[\varphi(X^D_T,K)\right] - \frac{1}{M}\sum_{i=1}^M \phi_{\varepsilon}(\tilde{X}^{\varepsilon,D,i}_T(\omega), K) \right\vert^2\right] \mu(\diff{x},\diff{\theta},\diff{K}) \\[2mm]
    \leq& \int_{\mathbb{R}^+\times\Theta^{\text{rBergomi}}\times\mathbb{R}^{n}} \bigg(18 C^2 \varepsilon^2\max(1,c_{20}) + 4C^2\left(c_{37} \varepsilon + (c_{25} + c_{37})h^{\underline{H}} \varepsilon^{-1} + c_{37} \varepsilon^{-1}\exp\left(-\frac{R^2}{4T}\right)\right)\\[2mm]
    &\phantom{...} + \frac{1}{M} C^2 \varepsilon^{-1}c_{36}\bigg)(1+\vert x\vert^2 + \Vert K \Vert^2 + \Vert \theta\Vert^2)\mu(\diff{x},\diff{\theta},\diff{K}) \\[2mm]
    \leq & \bar{C}\left(\varepsilon + h^{\underline{H}} \varepsilon^{-1} + \varepsilon^{-1}\exp\left(-\frac{R^2}{4T}\right) + M^{-1}\varepsilon^{-1}\right)
\end{align*}

for $\bar{C}:= 72 C^2 \max(1 + c_{37},c_{20}+c_{37},c_{25}+c_{37}, c_{36})$. Finally, by choosing 
$$
    M:= \lceil 16\bar{C}\bar{\varepsilon}^{-4}\rceil,\; R:= \sqrt{4T\log(16\bar{C}^2\bar{\varepsilon}^{-4})},\; h:= \left(\frac{\varepsilon \bar{\varepsilon}^2}{4\bar{C}}\right)^\frac{1}{\underline{H}},\; \varepsilon := (4\bar{C})^{-1}\bar{\varepsilon}^2
$$

and applying Fubini's theorem we have that
$$
    \mathbb{E}\left[\int_{\mathbb{R}^+\times\Theta^{\text{rBergomi}}\times\mathbb{R}^{n}} \left\vert \mathbb{E}[\varphi(X^D_T,K)] - \frac{1}{M}\sum_{i=1}^M \phi_\varepsilon(\tilde{X}^{\varepsilon,D,i}_T,K) \right\vert^2 \mu(\diff{x},\diff{\theta},\diff{K})\right] < \bar\varepsilon,
$$

which finishes the first part of the proof. So it remains to show that $\frac{1}{M}\sum_{i=1}^M \phi_\varepsilon(\tilde{X}^{\varepsilon,D,i}_T(\omega),K)$ is indeed a DNN satisfying \eqref{PriceApproxSizeBoundrBergomi}. For this purpose we define
$l_3:= \text{depth}(\Phi^{3}_\varepsilon)$, $l_4 := \text{depth}(\Phi^{4,1}_\varepsilon) = ... = \text{depth}(\Phi^{4,N}_\varepsilon)$, $l_\mathfrak{X}:= \text{depth}(\tilde{\mathfrak{X}}^\varepsilon_{t_1}(\omega)) = ... = \text{depth}(\tilde{\mathfrak{X}}^\varepsilon_{t_N}(\omega))$, $l_{\min} := \text{depth}(\Phi^{\min}_D)$ and $l_{\max}:= \text{depth}(\Phi^{\max}_{\frac{1}{D}})$. Moreover, we set
$$
    l^{\max,2}_V := \max\{2,l_3,l_4,l_\mathfrak{X},l_{\min},l_{\max}\}.
$$

Using that depth$(\Phi^{\max}_{\frac{1}{D}})$ = depth$(\Phi^{\min}_D) = 2$ we obtain with \eqref{SizeDNNVolterra}, \eqref{ApproxDNNSqrt} and Lemma \ref{LemmaPhi4kApprox}
$$
    l^{\max,2}_V \leq 2 + Nc_{26}\varepsilon^{-2} + c_{29} \varepsilon^{-2} + c_{30}\varepsilon^{-2} \leq \varepsilon^{-2}N(2 + c_{26} + c_{29} + c_{30}).
$$ 

Then using size$(\Phi^{\max}_{\frac{1}{D}})$ = size$(\Phi^{\min}_D) \leq 12$ it follows directly by \eqref{SizeDNNVolterra}, \eqref{ApproxDNNSqrt}, \eqref{VtildeDef}, Lemma \ref{LemmaPhi4kApprox} as well as the size of the composition/parallelisation of DNNs for every $i=1,...,M$ that
\begin{align*}
    \text{size}(\tilde{V}^{\varepsilon,D,i}_{t_k}(\omega))=& \text{size}\bigg(\Phi^{\text{mult}}_{\varepsilon} \odot \text{FP}\bigg(\mathcal{I}_{4l^{\max,2}_V - l_4}\odot \Phi^4_{\varepsilon},\mathcal{I}_{4l^{\max,2}_V - l_{\max} - l_{\min} - l_3}\\[2mm]
    &\phantom{...}\odot \Phi^{\min}_D \odot \Phi^{\max}_{\frac{1}{D}}\odot \Phi^{3,k}_{\varepsilon}\odot \text{FP}\left(\mathcal{I}_{l^{\max,2}_V},\mathcal{I}_{l^{\max,2}_V},\mathcal{I}_{l^{\max,2}_V - l_{\mathfrak{X}}}\odot \tilde{\mathfrak{X}}^{\varepsilon}_{t_k}\right)\bigg)\bigg)\\[2mm]
    \leq& 2\text{size}(\Phi^{\text{mult}}_{\varepsilon}) + 4\text{size}(\mathcal{I}_{4l^{\max,2}_V - l_4}) + 4\text{size}(\Phi^3_{\varepsilon}) + 4\text{size}(\mathcal{I}_{4l^{\max,2}_V - l_{\max} - l_{\min} - l_3}) +  8\text{size}(\Phi^{\min}_D) \\[2mm]
    &\phantom{...} + 16\text{size}(\Phi^{\max}_{\frac{1}{D}}) + 32\text{size}(\Phi^{4,k}_{\varepsilon}) + 64\text{size}(\mathcal{I}_{l^{\max,2}_V}) + 64\text{size}(\mathcal{I}_{l^{\max,2}_V - l_{\mathfrak{X}}}) + 64 \text{size}(\tilde{\mathfrak{X}}^{\varepsilon}_{t_k}) \\[2mm]
    \leq& 2\max\{176(2 + c_{26} + c_{29} + c_{30}), c_{31}\}N\varepsilon^{-2} =: c_{38}N\varepsilon^{-2}.
\end{align*}

Moreover, we define $l_5 := \text{depth}(\Phi^5_{\varepsilon})$, $l_V := \text{depth}(\tilde{V}^{\varepsilon, D, i}_{t_k}(\omega))$ and 
$$
    l^{\max,2}_X := \max\{2, l_5, l_V\}.
$$

Then as before using \eqref{DNNApproxSQRTRho} it holds in particular
$$
    l^{\max,2}_X \leq 2 + c_{34}\varepsilon^{-2} + c_{38}N\varepsilon^{-2} \leq N\varepsilon^{-2}(2+c_{34} + c_{38}).
$$

From \eqref{rBergomiDNNDiscreteTimeScheme} it holds together with the DNNs emulating the identity for every $i=1,...,M$ that
\begin{align*}
	\tilde{X}^{\varepsilon,D,i}_{t_{k+1}} := & \mathcal{I}_{l^{\max,2}_X}(\tilde{X}^{\varepsilon,D,i}_{t_{k}}) + \Phi^{\text{mult},3}_{\varepsilon}\left(\mathcal{I}_{l^{\max,2}_X}(\tilde{X}^{\varepsilon,D,i}_{t_{k}}) , \mathcal{I}_{l^{\max,2}_X - l_v}\left(\tilde{V}^{\varepsilon,i}_{t_k}(v,\rho,H)\right), \mathcal{I}_{l^{\max,2}_X}(\rho) \right) (B^{\tau,i}_{t_{k+1}}-B^{\tau,i}_{t_k}) \\[2mm]
	&\phantom{...} +\Phi^{\text{mult},3}_{\varepsilon}\left(\mathcal{I}_{l^{\max,2}_X}(\Tilde{X}^{\varepsilon,D,i}_{t_{k}}) , \mathcal{I}_{l^{\max,2}_X - l_v}\left(\tilde{V}^{\varepsilon,i}_{t_k}(v,\rho,H)\right), \mathcal{I}_{l^{\max,2}_X - l_5}(\Phi^5(\rho))\right) (B^{\perp,\tau,i}_{t_{k+1}}-B^{\perp,\tau,i}_{t_k})
\end{align*}

with $\tilde{X}^{\varepsilon,D,i}_{t_0} = x$.  For a fixed $\omega\in\Omega$ we again get by (randomly) weighted summation and parallelisation that there exists a DNN $\Phi^{i}_{k+1}:\mathbb{R}\times [\underline{\nu},\overline{\nu}] \times [\underline{\rho},\overline{\rho}]\times[-1,0]\times[\underline{H},\overline{H}]\to \mathbb{R}^+$ such that 
\begin{align*}
	\tilde{X}^{\varepsilon,D,i}_{t_{k+1}}(\omega) &= \Phi^{i}_{k+1}\left(\tilde{X}^{\varepsilon,D,i}_{t_{k}}(\omega), \nu, \eta,\rho, H \right).
\end{align*}

In particular, we get using \eqref{SizeMultDNN}
\begin{align*}
	\text{size}(\Phi^{i}_{k+1}) &\leq \text{size}(\mathcal{I}_{l^{\max,2}_X}) + \text{size}\left(\Phi^{\text{mult},3}_{\varepsilon} \odot \text{FP}(\mathcal{I}_{l^{\max,2}_X}, \mathcal{I}_{l^{\max,2}_X - l_V}\odot \bar{V}^{\varepsilon}_{t_k}(\omega), \mathcal{I}_{l^{\max,2}_X})\right)\\[2mm]
	&\phantom{...} + \text{size}\left(\Phi^{\text{mult},3}_{\varepsilon} \odot \text{FP}(\mathcal{I}_{l^{\max,2}_X}, \mathcal{I}_{l^{\max,2}_X - l_V}\odot \bar{V}^{\varepsilon}_{t_k}(\omega), \mathcal{I}_{l^{\max,2}_X - l_5} \odot \Phi^5)\right) \\[2mm]
	&\leq 30 l^{\max,2}_X + 2c_{35}NR\varepsilon^{-2} + 2c_{34}\varepsilon^{-2} + 8c_{38}N\varepsilon^{-2} \\[2mm]
	& \leq 30N\varepsilon^{-2}(2 + c_{34} + c_{38}) + 4c_{35}NR\varepsilon^{-2} + 2c_{34}\varepsilon^{-2} + 8c_{38}N\varepsilon^{-2} \\[2mm]
	&\leq 38 NR \varepsilon^{-2}(2+ c_{34} + c_{35} + 2c_{38}) \\[2mm]
	&=: NRc_{39}\varepsilon^{-2}.
\end{align*}

We define $l_{k+1} := \text{depth}(\Phi^{i}_{k+1})$. Since this holds for every $k=0,...,N-1$, we get recursively and by parallelisation that
\begin{align*}
	\tilde{X}^{\varepsilon,D,i}_{t_{k+1}}(\omega) &= \Phi^{i}_{k+1}\left(\tilde{X}^{\varepsilon,D,i}_{t_{k}}(\omega), \nu, \eta,\rho, H \right) = \Phi^{i}_{k+1}\left(\Phi^{i}_{k}\left(\tilde{X}^{\varepsilon,D,i}_{t_{k-1}}(\omega), \nu, \eta,\rho, H \right), \nu, \eta,\rho, H \right) \\[2mm]
	&= \Phi^{i}_{k+1}\left(\Phi^{i}_{k}\left(\tilde{X}^{\varepsilon,D}_{t_{k-1}}(\omega), \nu, \eta,\rho, H \right), \mathcal{I}_{l_k}(\nu), \mathcal{I}_{l_k}(\eta),\mathcal{I}_{l_k}(\rho), \mathcal{I}_{l_k}(H) \right)\\[2mm]
	&= \Phi^{i}_{k+1}\left(\tilde{\Phi}^{i}_k\left(\tilde{X}^{\varepsilon,D}_{t_{k-1}}(\omega), \nu, \eta, \rho,H\right)\right) = \Phi^{i}_{k+1}\left(\tilde{\Phi}^{i}_k\left(\tilde{\Phi}^{i}_{k-1}\left(\tilde{X}^{\varepsilon,D,i}_{t_{k-2}}(\omega), \nu, \eta, \rho,H\right)\right)\right) \\[2mm]
	&= \dots = \\[2mm]
	&= \Phi^{i}_{k+1}\circ \tilde{\Phi}^{i}_k \circ \dots \circ \Phi^{i}_1(x,\nu,\eta,\rho,H) \\[2mm]
	&= \bar{\Phi}^{i}_{k+1}(x,\nu,\eta,\rho,H)
\end{align*}

for $ \bar{\Phi}^{i}_{k+1} := \Phi^{i}_{k+1}\odot \tilde{\Phi}^{i}_k \odot \dots \odot \tilde{\Phi}^{i}_1$ and $\tilde{\Phi}^{i}_k:= \text{FP}(\Phi^i_k, \mathcal{I}_{l_k},\mathcal{I}_{l_k},\mathcal{I}_{l_k},\mathcal{I}_{l_k})$. For the latter it holds
$$
	\text{size}(\tilde{\Phi}^{i}_k) \leq \text{size}(\Phi^{i}_k) + 4\text{size}(\mathcal{I}_{l_k}) \leq 9\text{size}(\Phi^i_k) \leq 9Nc_{39}\varepsilon^{-2}.
$$

Based on this we obtain analogous to the proof of Theorem \ref{MainTheorem} that
$$
	\text{size}(\bar{\Phi}^{i}_{k+1}) \leq \text{size}(\Phi^{i}_{k+1}) + 3 \sum_{j=1}^k \text{size}(\tilde{\Phi}^{i}_j) \leq N\varepsilon^{-2} c_{39} + 27kN\varepsilon^{-2} c_{39} = \varepsilon^{-2}(1+27k)NR c_{39}.
$$

In particular, we have
$$
	\tilde{X}^{\varepsilon,D,i}_T(\omega) = \tilde{X}^{\varepsilon,D,i}_{t_N}(\omega) = \bar{\Phi}^{i}_N(x,\nu,\eta,\rho,H) = \bar{\Phi}^{i}_N(x,\theta)
$$

with size$(\bar{\Phi}^{i}_N) \leq \varepsilon^{-2}(1+27N)N c_{39}$. For $l_{\bar{\Phi}} := \text{depth}(\bar{\Phi}^i_N)$ we obtain
$$
	\frac{1}{M}\sum_{i=1}^M \phi_{\varepsilon}(\tilde{X}^{\varepsilon,D,i}_{T}(\omega), K) = \frac{1}{M}\sum_{i=1}^M \phi_{\varepsilon}(\bar{\Phi}^{i}_N(x,\theta), \mathcal{I}_{n,l_{\bar{\Phi}}}(K)) = \frac{1}{M}\sum_{i=1}^M \bar{\Psi}^{i}(x,\theta,K)
$$

for the DNN $\bar{\Psi}^{i} := \phi_{\varepsilon} \odot \text{FP}(\bar{\Phi}^{i}_N, \mathcal{I}_{n,l_{\bar{\Phi}}})$. We therefore choose $U^{\text{rBergomi}}_{\bar{\varepsilon}}(x,\theta,K):= \frac{1}{M}\sum_{i=1}^M \bar{\Psi}^{i}(x,\theta,K)$ and it follows
\begin{align*}
	\text{size}(U^{\text{rBergomi}}_{\bar{\varepsilon}}) &\leq 2M\text{size}(\phi_{\varepsilon}) + 2M\text{size}(\bar{\Phi}^N) + 2M\text{size}(\mathcal{I}_{n,l_{\bar{\Phi}}})\\[2mm]
	&\leq 2MC\varepsilon^{-q} + 2M(1+2n)(1+27N)Nc_{39}\varepsilon^{-2}\\[2mm]
    &\leq 8\bar{C}\bar{\varepsilon}^{-2} \varepsilon^{-(q+3)}(C + (1+2n)(1+27N)NRc_{39})\\[2mm]
    &\leq 32 \bar{C}^{q+4+4/\underline{H}} 4^{q+3+4/\underline{H}}T^5(C+(1+2n)(1+27)c_{39})(1+\log(16\bar{C}^2)) \bar{\varepsilon}^{-2(q+6+4/\underline{H})},
\end{align*}

where we have used our choices for $\varepsilon$, $h$, $R$ and $M$ from the first part of this proof. Setting $\mathfrak{C}:= 32 \bar{C}^{q+4+4/\underline{H}} 4^{q+3+4/\underline{H}}T^5(C+(1+2n)(1+27)c_{39})(1+\log(16\bar{C}^2))$ and $\mathfrak{q}:=2(q+6+4/\underline{H})$ finishes the proof.
\end{proof}

In the following proposition we show that the option price in the truncated rBergomi model converges to the option price in the standard model in \eqref{rBergomiModel} for bounded payoff functions. The limitation to bounded payoff functions arises from the fact that there are no sufficient conditions for the existence of $p$-th moments of $X^{x,v}$ for $p>1$ yet established in the literature as pointed out in \cite{Gassiat2019} or Section 4.2 in \cite{Friz2021}.

\begin{proposition}\label{PropConvrBergomiPrice}
    For $n\in\mathbb{N}$ let $\varphi:\mathbb{R}^+\times\mathbb{R}^n\rightarrow \mathbb{R}$ be a bounded payoff function. Then for any $K\in\mathbb{R}^n$ it holds
    $$
        \mathbb{E}\left[\varphi(X^D_T,K)\right] \rightarrow \mathbb{E}\left[\varphi(X^{x,v}_T,K)\right]
    $$

    for $D\rightarrow +\infty$.
\end{proposition}

\begin{proof}
    Let $\varphi:\mathbb{R}^+\times\mathbb{R}^n\rightarrow \mathbb{R}$ be a bounded payoff function for some $n\in\mathbb{N}$. Consider the set
    $$
        A_D := \left\{\omega\in\Omega\,\middle\vert\,\exists t\in[0,T] \text{ such that } \eta \sqrt{\frac{H}{2}}{\mathfrak{X}}_{t} - \frac{1}{4}\eta^2 t^{2H} \notin[-\log(D),\log(D)]\right\}.
    $$
    
    By Markov's inequality together with Theorem 4.2 in \cite{Nourdin2012} we get that $\mathbb{Q}(A_D)\rightarrow 0$ as $D\rightarrow +\infty$. Therefore, we obtain by the triangle inequality and the boundedness of $\varphi$
    \begin{align*}
        &\lim_{D\to\infty}\mathbb{E}[\vert \varphi(X^D_T,K) - \varphi(X^{x,v}_T,K)\vert] \\[2mm] =& \lim_{D\to\infty} \mathbb{E}[\vert \varphi(X^D_T,K) - \varphi(X^{x,v}_T,K)\vert\mathbbm{1}_{A_D}] + \lim_{D\to\infty}\mathbb{E}[\vert \varphi(X^D_T,K) - \varphi(X^{x,v}_T,K)\vert\mathbbm{1}_{A^c_D}] \\[2mm]
        \leq& \lim_{D\to\infty} 2\sup_y \vert \varphi(y,K)\vert \mathbb{Q}(A_D) + \lim_{D\to\infty} \mathbb{E}[\vert \varphi(X^D_T,K) - \varphi(X^{x,v}_T,K)\vert\mathbbm{1}_{A^c_D}] \\[2mm]
        =& \lim_{D\to\infty} \mathbb{E}[\vert \varphi(X^D_T,K) - \varphi(X^{x,v}_T,K)\vert\mathbbm{1}_{A^c_D}] \\
        =& 0.
    \end{align*}

    The last step follows by the definition of $A_D$.
\end{proof}


\newpage
\printbibliography


\newpage
\appendix
\section{Appendix}
\subsection{Deep neural networks}\label{Section2}
We follow the introduction to deep neural networks in \cite{Gonon2021} to give a short reminder on the most crucial terminology. Let $d,d^\prime,L\in\mathbb{N}$. A function $\psi:\mathbb{R}^d\rightarrow \mathbb{R}^{d^\prime}$ is a (feedforward) DNN with $L-1$ hidden layers and single activation function $\varrho$ if it can be written as
\begin{equation}
    \psi = W_L \circ (\varrho \circ W_{L-1}) \circ \dots \circ (\varrho \circ W_1),
\end{equation}

where $W_l:\mathbb{R}^{d_{l-1}}\rightarrow\mathbb{R}^{d_l}$ for $l=1,...,L$ is an affine function 
\begin{equation}
    W_l(x) := A^l x + b^l,\, x \in\mathbb{R}^{d_{l-1}}
\end{equation}

defined by weight matrices $A^l\in\mathbb{R}^{d_l \times d_{l-1}}$ and bias vectors $b^l \in \mathbb{R}^{d_l}$ for $d_0 := d$ and $d_L := d^\prime$. The total numbers of the DNN's layers is called depth of the DNN, i.e. 
$$
    \text{depth}(\psi) := L + 1.
$$ 

The size of the DNN is defined as
$$
    \text{size}(\psi) := \vert \lbrace (i,j,l): \, A^l_{i,j} \neq 0 \rbrace\vert + \vert \lbrace (i,l):\, b^l_i \neq 0 \rbrace \vert.
$$

Sometimes we are especially interested in the size of the output layer, which is defined by
$$\text{size}_{out}(\psi) := \vert \lbrace (i,j): \, A^L_{i,j} \neq 0 \rbrace\vert + \vert \lbrace i:\, b^L_i \neq 0 \rbrace \vert.$$
Moreover, we define the width of the DNN as
$$
    \text{width}(\psi) := \max_{l=0,...,L}d_{l}.
$$
Throughout this work we solely consider the ReLU activation function whose $d$-dimensional version is defined by
\begin{equation}\label{ReLU}
    \varrho:\mathbb{R}^d\rightarrow \mathbb{R}^d,\; \varrho(\theta) := (\max\lbrace \theta_i,0\rbrace)_{i=1,...,d},
\end{equation}
where $\theta_i$ is the $i$-th entry of the vector $\theta$.

\subsection{Proofs of Section \ref{Section3}}\label{ProofSection3}
\begin{proof}[Lemma \ref{LemmaBoundSecMomX}]
Define the function $G(t):= \mathbb{E}[\sup_{s\in[0,t]}\Vert X^{x,v}_s\Vert^2]$ for $t\in[0,T]$. By Doob's maximal inequality, Itô's isometry, Lemma \ref{LemmaL2XV} and \eqref{SquaredTriangular} it follows
    \begin{align}\label{Lemma4.3.Bound1}
        G(t) &\leq 2\Vert x \Vert^2 + 2 \mathbb{E}\left[\sup_{s\in[0,t]}\left\Vert \int_0^s \sigma^\theta(X^{x,v}_r,V^v_r)\diff{W}_r \right\Vert^2\right] \nonumber \\
        &\leq 2\Vert x \Vert^2 + 8\mathbb{E}\left[\left\Vert \int_0^t \sigma^\theta(X^{x,v}_s,V^v_s)\diff{W}_s \right\Vert^2\right] \nonumber \\
        &= 2\Vert x \Vert^2 + 8 \int_0^t\mathbb{E}\left[ \left\Vert\sigma^\theta(X^{x,v}_s,V^v_s)\right\Vert^2_F\right]\diff{s}.
    \end{align}
    
    For the integral term we obtain
    \begin{align}\label{Lemma4.3.Bound3}
        &\int_0^t  \mathbb{E}\left[\left\Vert\sigma^\theta(X^{x,v}_s,V^v_s)\right\Vert^2_F\right]\diff{s}\nonumber\\
        \leq & 2 \int_0^t \mathbb{E}\left[ \left\Vert\sigma^\theta(X^{x,v}_s,V^v_s) - \sigma^\theta(0,V^v_s)\right\Vert^2_F\right] \diff{s} + 2\int_0^t \mathbb{E}\left[ \left\Vert\sigma^\theta(0,V^v_s)\right\Vert^2_F\right]\diff{s} \nonumber\\
        \leq & 2L\int_0^t \mathbb{E}\left[\Vert X^{x,v}_s \Vert^2\right]\diff{s} + 2\int_0^t \mathbb{E}\left[d r K_1(1+\Vert V^v_s \Vert^2)\right]\diff{s} \nonumber\\
        \leq & 2L\int_0^t \mathbb{E}\left[\Vert X^{x,v}_s \Vert^2\right]\diff{s} + 2d r K_1 t(1+c_1\Vert v\Vert^2 + c_2 \bar{d}^2)
    \end{align}
    
    by Lemma \ref{LemmaBoundSecMomV}. Therefore, by Lemma \ref{LemmaL2XV} and \eqref{Lemma4.3.Bound3} it follows that $G(t)\in L^1([0,T])$. By \eqref{Lemma4.3.Bound1} and \eqref{Lemma4.3.Bound3} we obtain
    \begin{equation*}
        G(t) \leq 2\Vert x \Vert^2 + 16 c_1 K_1 T \Bar{d}^2 \Vert v \Vert^2 + 16 K_1 T(1+c_2)\bar{d}^4 + 16 L \int_0^t G(s) \diff{s}.
    \end{equation*}
    
    So finally, by Gronwall's inequality
    \begin{equation}
        G(t) \leq (2 \Vert x \Vert^2 + 16c_1 K_1 T \Vert v \Vert^2 \bar{d}^2 + 16 K_1 T (1+c_2)\bar{d}^4)\exp(16LT).
    \end{equation}
    
    Finally, we get that for all $t\in[0,T]$
    \begin{equation}
        G(t) \leq c_3 \Vert x \Vert^2 + c_4 \Vert v \Vert^2 \bar{d}^2 + c_5 \bar{d}^4,
    \end{equation}
    
    where 
    \begin{equation*}
        c_3 := 2\exp(16LT),\; c_4 := 16 c_1 K_1 T \exp(16LT)\; \text{ and }  c_5 := 16 K_1 T (1+c_2) \exp(16LT).
    \end{equation*}
\end{proof}

\begin{proof}[Proposition \ref{DiffVNormalBar}]\label{ProofProp2.1}
    Define the functional $G(t) := \mathbb{E}\left[ \sup_{s\in[0,t]} \Vert V^v_s - \bar{V}^v_s \Vert^2 \right]$ for $t\in[0,T]$. Then it holds by Cauchy-Schwarz inequality, Doob's maximal inequality, Itô's isometry, Assumption \ref{StrucAssumptionV} and Lemma \ref{LemmaL2XV} that
    \begin{align*}
        G(t) &\leq 2\mathbb{E}\left[ \sup_{s\in[0,t]} \left\Vert \int_0^s \bar{\mu}^\theta(V^v_r) - \bar{\mu}^\theta(\bar{V}^v_{\floor{r}}) \diff{r} \right\Vert^2 \right] + 2\mathbb{E}\left[ \sup_{s\in[0,t]} \left\Vert \int_0^s \bar{\sigma}^\theta(V^v_r) - \bar{\sigma}^\theta(\bar{V}^v_{\floor{r}}) \diff{B}_r \right\Vert^2 \right] \\
        &\leq 2t \mathbb{E}\left[\int_0^t \left\Vert\bar{\mu}^\theta(V^v_r) - \bar{\mu}^\theta(\bar{V}^v_{\floor{r}})\right\Vert^2 \diff{r} \right] + 8 \mathbb{E}\left[ \left\Vert \int_0^t \bar{\sigma}^\theta(V^v_r) - \bar{\sigma}^\theta(\bar{V}^v_{\floor{r}}) \diff{B}_r \right\Vert^2 \right] \\
        &= 2t \int_0^t\mathbb{E}\left[ \left\Vert\bar{\mu}^\theta(V^v_r) - \bar{\mu}^\theta(\bar{V}^v_{\floor{r}})\right\Vert^2 \right]\diff{r} + 8\int_0^t\mathbb{E}\left[ \Vert \bar{\sigma}^\theta(V^v_r) - \bar{\sigma}^\theta(\bar{V}^v_{\floor{r}}) \Vert^2_F \right]\diff{r}\\
        &\leq 2tL \int_0^t \mathbb{E}\left[ \Vert V^v_r - \bar{V}^v_{\floor{r}}\Vert^2 \right]\diff{r} + 8L\int_0^t \mathbb{E}[\Vert V^v_r - \bar{V}^v_{\floor{r}} \Vert^2]\diff{r}\\
        &\leq 2L\max(2t,8)\int_0^t \mathbb{E}[\Vert V^v_r - \bar{V}^v_{\floor{r}} \Vert^2]\diff{r}.
    \end{align*}
    
    Applying the triangular inequality yields
    \begin{equation}
        G(t) \leq 2L\max(2t,8)\int_0^t \mathbb{E}[\Vert V^v_r\Vert^2] + \mathbb{E}[\Vert \bar{V}^v_{\floor{r}}\Vert^2] \diff{r},
    \end{equation}
    
    which is bounded by Lemma \ref{LemmaL2XV} and \ref{L2floor}. So it follows that $G(t)\in L^1([0,T])$ for any $t\in[0,T]$. Moreover, it clearly holds
    \begin{equation}
        \mathbb{E}[\Vert V^v_r - \bar{V}^v_{\floor{r}} \Vert^2] \leq 2\mathbb{E}[\Vert V^v_r - V^v_{\floor{r}}\Vert^2] + 2\mathbb{E}[\Vert V^v_{\floor{r}} - \bar{V}^v_{\floor{r}}\Vert^2].
    \end{equation}
    
    By Itô's isometry and the Cauchy-Schwarz inequality, for any $r\in [0,t]$ we have
    \begin{align*}
        \mathbb{E}[\Vert V^v_r - V^v_{\floor{r}}\Vert^2]&\leq 2\mathbb{E}\left[\left\Vert\int_{\floor{r}}^r\bar{\mu}^\theta(V^v_s)\diff{s} \right\Vert^2\right] + 2\mathbb{E}\left[\left\Vert\int_{\floor{r}}^r\bar{\sigma}^\theta(V^v_s)\diff{B}_s \right\Vert^2\right]\\
        &\leq 2(r-\floor{r})\int_{\floor{r}}^r \mathbb{E}[\Vert \bar{\mu}^\theta(V^v_s)\Vert^2]\diff{s} + 2\int_{\floor{r}}^r \mathbb{E}[\Vert \bar{\sigma}^{\theta}(V^v_s)\Vert^2_F]\diff{s}\\
        &\leq 2(r-\floor{r})dK_2 \int_{\floor{r}}^r \left(1+\mathbb{E}[\Vert V^v_s\Vert^2]\right)\diff{s} + 2d r K_2\int_{\floor{r}}^r \left(1 + \mathbb{E}[\Vert V^v_s\Vert^2]\right)\diff{s}\\
        &= 2K_2[(r-\floor{r})d + d r]\int_{\floor{r}}^r \left(1+ \mathbb{E}[\Vert V^v_s \Vert^2]\right) \diff{s}.
    \end{align*}
    
    By Lemma \ref{LemmaBoundSecMomV} we have that
    \begin{equation}\label{Bounds VrVfloor}
        \mathbb{E}[\Vert V^v_r - V^v_{\floor{r}}\Vert^2] \leq 2K_2[(r-\floor{r})d+dr]h(1+c_1\Vert v\Vert^2 + c_2\bar{d}^2)
    \end{equation}
    
    Using \eqref{Bounds VrVfloor} we get
    \begin{align*}
        G(t) &\leq 4L\max(2t,8)\left( \int_0^t G(r)\diff{r} + \int_0^t \mathbb{E}[\Vert V^v_r - V^v_{\floor{r}}\Vert^2] \diff{r} \right)\\
        &\leq 4L\max(2t,8)\left( \int_0^t G(r)\diff{r} + 4K_2\bar{d}^2 h T (1 +c_1\Vert v\Vert^2 + c_2 \bar{d}^2)\right).
    \end{align*}
    
    By Gronwall's inequality it follows
    \begin{equation}
        G(t) \leq 16 L\max(2T,8)K_2 \bar{d}^2 h T (1 +c_1\Vert v \Vert^2 + c_2\bar{d}^2)\exp(4\max(2T,8)LT).
    \end{equation}
    
    So we finally get that for all $t\in[0,T]$
    \begin{equation}
        G(t) \leq h(c_6 \Vert v \Vert^2 \bar{d}^2 + c_7 \bar{d}^4),
    \end{equation}
    
    where 
    \begin{equation*}
        c_6 := 16L\max(2T,8)c_1 K_2 T \exp(4\max(2T,8)LT )
    \end{equation*}
    
    and
    \begin{equation*}
        c_7 := 16L\max(2T,8)(1+c_2) K_2 T \exp(4\max(2T,8)LT).
    \end{equation*}
\end{proof}

\begin{proof}[Lemma \ref{SecMomVTild}]
    Define the function $G(t) := \mathbb{E}\left[\sup_{s\in[0,t]}\Vert \tilde{V}^\varepsilon_s \Vert^2\right]$ for $t\in[0,T]$. By the Cauchy-Schwarz inequality, Doob's maximal inequality, Itô's isometry, Assumption \ref{AssumpCoeffApprox} and \eqref{SquaredTriangular} it follows
    \begin{align*}
        G(t) &
        \leq 3\Vert v \Vert^2 + 3\mathbb{E}\left[\sup_{s\in[0,t]} \left\Vert \int_0^s \bar{\mu}_\varepsilon(\tilde{V}^\varepsilon_{\floor{r}},\theta)\diff{r} \right\Vert^2 \right] + 3\mathbb{E}\left[\sup_{s\in[0,t]} \left\Vert \int_0^s \bar{\sigma}_\varepsilon(\tilde{V}^\varepsilon_{\floor{r}},\theta) \diff{B}_r \right\Vert^2 \right]\\
        &\leq 3 \Vert v \Vert^2 + 3t \int_0^t \mathbb{E}\left[ \Vert \bar{\mu}_\varepsilon(\tilde{V}^\varepsilon_{\floor{s}}, \theta) \Vert^2\right] \diff{s} + 12 \int_0^t \mathbb{E}\left[\left\Vert \bar{\sigma}_\varepsilon(\tilde{V}^\varepsilon_{\floor{s}},\theta)\right\Vert^2_F
        \right] \diff{s} \\
        &\leq 3 \Vert v \Vert^2 + 9C^2 T^2 \bar{d}^{2p}  + 9C^2 T\int_0^t \mathbb{E}[\Vert \tilde{V}^\varepsilon_{\floor{s}}\Vert^2]\diff{s} + 9C^2 T^2\Vert \theta \Vert^2 \\ 
        &\phantom{1cm} + 36 C^2 T \bar{d}^{2p} + 36 C^2 \int_0^t \mathbb{E}[\Vert \tilde{V}^\varepsilon_{\floor{s}}\Vert^2]\diff{s} + 36C^2 T \Vert \theta\Vert^2,
    \end{align*}
    
    which is finite since $\tilde{V}^\varepsilon_{\floor{t}}\in L^2(\Omega,\mathcal{F},\mathbb{Q})$ for all $t\in[0,T]$. This can be shown inductively by similar arguments as in the proof of Lemma \ref{L2floor}. Therefore, it follows $G(t)\in L^1([0,T])$. Furthermore, we see 
    \begin{align*}
        G(t)&
        \leq 3\Vert v \Vert^2 + 36C^2 T(T+1) \Vert \theta \Vert^2 + 36 C^2 T(T+1)\bar{d}^{2p} + 3(C^2T + 12C^2) \int_0^t G(s)\diff{s}.
    \end{align*}
    
    By Gronwall's inequality it finally follows
    \begin{align*}
        G(t) &\leq \left(3\Vert v \Vert^2 + 36C^2 T(T+1) \Vert \theta \Vert^2 + 36 C^2 T(T+1)\bar{d}^{2p}\right)\exp(3(C^2T + 12C^2)T)\\
        &\leq c_{11}\Vert v \Vert^2 + c_{12}\Vert \theta \Vert^2 + c_{13} \bar{d}^{2p} 
    \end{align*}
    
    for 
    \begin{equation*}
        c_{11} := 3\exp(3(C^2T + 12C^2)T),\; c_{12} := 36 C^2 T(T+1) \exp(3(C^2T + 12C^2)T)
    \end{equation*}
    
    and
    \begin{equation*}
        c_{13} := 36 C^2 T(T+1)\exp(3(C^2T + 12C^2)T).
    \end{equation*}
\end{proof}

\begin{proof}[Lemma \ref{SecMomXTild}]
Define the function $G(t) := \mathbb{E}\left[\sup_{s\in[0,t]}\Vert \tilde{X}^\varepsilon_s \Vert^2\right]$ for $t\in[0,T]$. By Doob's maximal inequality, Itô's isometry, Assumption \ref{AssumpCoeffApprox} and \eqref{SquaredTriangular} it follows
\begin{align*}
    G(t)
    &\leq 2\Vert x \Vert^2 + 2\mathbb{E}\left[\sup_{s\in[0,t]}\left\Vert \int_0^s \sigma_\varepsilon(\tilde{X}^\varepsilon_{\floor{r}},\tilde{V}^\varepsilon_{\floor{r}},\theta)\diff{W}_r\right\Vert^2\right]\\
    & \leq 2 \Vert x \Vert^2 + 8 \mathbb{E}\left[\left\Vert\int_0^t  \sigma_\varepsilon(\tilde{X}^\varepsilon_{\floor{s}},\tilde{V}^\varepsilon_{\floor{s}},\theta)\diff{W_s}\right\Vert^2\right]\\
    & = 2\Vert x \Vert^2 + 8\int_0^t \mathbb{E}\left[\left\Vert \sigma_\varepsilon(\tilde{X}^\varepsilon_{\floor{s}},\tilde{V}^\varepsilon_{\floor{s}},\theta) \right\Vert^2_F\right]\diff{s}\\
    &\leq 2 \Vert x \Vert^2 + 8C^2\bigg( 4T \bar{d}^{2p}  + 4T \Vert \theta \Vert^2 + 4 \int_0^t \left( \mathbb{E}[\Vert \Tilde{X}^{\varepsilon}_{\floor{s}}\Vert^2] + \mathbb{E}[\Vert \Tilde{V}^v_{\floor{s}} \Vert^2] \right) \diff{s} \bigg).
\end{align*}

As in the proof of Lemma \ref{SecMomVTild} we get that $G\in L^1([0,T])$. Moreover, we get by Lemma \ref{SecMomVTild}
\begin{align*}
    G(t) \leq 2 &\Vert x \Vert^2 + 32 C^2 T c_{11} \Vert v \Vert^2 + 32 C^2 T (1 + c_{12}) \Vert \theta \Vert^2 \\ &+ 32 C^2 T (1 + c_{13})\bar{d}^{2p} + 32C^2 \int_0^t G(s)\diff{s}.
\end{align*}

So by Gronwall's inequality we finally get
\begin{align*}
    G(t) &\leq \bigg(2\Vert x \Vert^2 + 32 C^2 T c_{11} \Vert v \Vert^2 + 32 C^2 T (1 + c_{12}) \Vert \theta \Vert^2 + 32 C^2 T (1 + c_{13})\bar{d}^{2p}\bigg)\exp(32C^2 T)\\
    &= c_{14} \Vert x \Vert^2 + c_{15}\Vert v \Vert^2 + c_{16}\Vert \theta \Vert^2 + c_{17}\bar{d}^{2p}  
\end{align*}

for 
\begin{equation}
    c_{14}:=2 \exp(32 C^2 T),\; c_{15}:= 32C^2 T c_{11} \exp(32 C^2 T),
\end{equation}

and
\begin{equation}
    c_{16}:= 32C^2 T (1 + c_{12}) \exp(32 C^2 T),\; c_{17}:= 32 C^2 T (1 + c_{13})\exp(32 C^2 T).
\end{equation}
\end{proof}

\end{document}